%% file: main.tex
\documentclass{article}
\usepackage{graphicx} 
\usepackage{fullpage}

\input{math_commands}

\usepackage{bbm}
\usepackage{amsmath}
\usepackage{amssymb}
\usepackage{mathtools}
\usepackage{amsthm}
\usepackage{physics}
\usepackage{xcolor}
\definecolor{blueviolet}{rgb}{0.2, 0.2, 0.6}
\definecolor{webgreen}{rgb}{0,.5,0}
\definecolor{webbrown}{rgb}{.6,0,0}
\definecolor{darkred}{HTML}{CC333A}
\usepackage{hyperref}
\hypersetup{
  colorlinks=true,breaklinks,
  linktoc=section,
  linkcolor=blueviolet,
  linkbordercolor=white,
  citecolor=blueviolet,
  urlcolor=webgreen
  }

\usepackage[
        noabbrev,
        capitalise,
        nameinlink,
    ]{cleveref}
\usepackage{physics}
\usepackage{caption,subcaption}

\usepackage{thmtools}
\usepackage{thm-restate}

\newtheorem{theorem}{Theorem}
\newtheorem{proposition}[theorem]{Proposition}
\newtheorem{lemma}[theorem]{Lemma}
\newtheorem{corollary}[theorem]{Corollary}
\newtheorem{conjecture}[theorem]{Conjecture}
\theoremstyle{definition}
\newtheorem{definition}[theorem]{Definition}

\theoremstyle{remark}

\newcommand{\Pois}{\ensuremath{\operatorname{Pois}}\xspace}
\newcommand{\pr}{\mathbb{P}}

\newcommand{\inda}{\ensuremath{\bar{\bm a}}}

\newcommand{\indi}{\ensuremath{\bar{\bm i}}}
\newcommand{\indj}{\ensuremath{\bar{\bm j}}}
\newcommand{\indasub}{\ensuremath{ a}}
\newcommand{\indisub}{\ensuremath{ i}}
\newcommand{\multiset}[1]{\ensuremath{\{\!\!\{#1\}\!\!\}}}

\title{Bounds on the ground state energy of quantum $p$-spin Hamiltonians}

\author{
{\sf Eric R.\ Anschuetz}\thanks{Institute for Quantum Information and Matter, Caltech; Walter Burke Institute for Theoretical Physics, Caltech; MIT Center for Theoretical Physics, MIT; e-mail: {\href{mailto:eans@caltech.edu}{\texttt{eans@caltech.edu}}}.} \thanks{Funding from the Walter Burke Institute for Theoretical Physics at Caltech and the DARPA ONISQ program (grant number W911NF2010021) is gratefully acknowledged.}
\and
{\sf David Gamarnik}    \thanks{Operations Research Center, Statistics and Data Science Center,  Sloan School of Management, MIT; e-mail: {\href{mailto:gamarnik@mit.edu}{\texttt{gamarnik@mit.edu}}}} \thanks{Funding from NSF Grant CISE-2233897 is gratefully acknowledged.}
\and
{\sf Bobak T.\ Kiani}\thanks{John A. Paulson School of Engineering and Applied Sciences, Harvard; Department of Electrical Engineering and Computer Science, MIT; e-mail: {\href{mailto:bkiani@g.harvard.edu}{\texttt{bkiani@g.harvard.edu}}}.}
}

\date{}

\begin{document}

\maketitle

\begin{abstract}
We consider the problem of estimating the ground state energy of quantum $p$-local spin glass random Hamiltonians, the quantum analogues of widely studied classical spin glass models. Our main result shows that the maximum energy achievable by product states has a well-defined limit (for even $p$) as $n\to\infty$ and 
is $E_{\text{product}}^\ast=\sqrt{2 \log p}$ in the limit of large $p$. This value is interpreted as the maximal energy of a much simpler so-called Random Energy Model, widely studied in the setting of classical spin glasses. The proof of the limit existing follows from an extension of Fekete's Lemma after we demonstrate near super-additivity of the (normalized) quenched free energy. The proof of the value follows from a second moment method on the number of states achieving a given energy when restricting to an $\epsilon$-net of product states.

Furthermore, we relate the maximal energy achieved over \emph{all} states to a $p$-dependent constant $\gamma\left(p\right)$, which is defined by the degree of violation of a certain asymptotic independence ansatz over graph matchings. We show that the maximal energy achieved by all states $E^\ast\left(p\right)$ in the limit of large $n$ is at most $\sqrt{\gamma\left(p\right)}E_{\text{product}}^\ast$. We also prove using Lindeberg's interpolation method that the limiting $E^\ast\left(p\right)$ is robust with respect to the choice of the randomness and, for instance, also applies to the case of sparse random Hamiltonians. This robustness in the randomness extends to a wide range of random Hamiltonian models including SYK and random quantum max-cut.

\end{abstract}

\thispagestyle{empty}
\tableofcontents
\thispagestyle{empty}
\clearpage

\section{Introduction}\label{sec:intro}

We  study the values of the maximal energy of the quantum analogue of a classical $p$-spin glass model~\cite{bray1980replica,erdHos2014phase,baldwin2020quenched}, which is a $p$-local random Hamiltonian taking the form
\begin{equation*}
    H_{n,p} = {n \choose p}^{-1/2} \sum_{\substack{\sigma \in \mathcal{P}_n \\ \sigma \text{ is $p$-local}}} \alpha_\sigma \sigma, \quad \quad \alpha_\sigma \sim \mathcal{N}(0,1),
\end{equation*}
where $\mathcal{P}_n$ denotes the set of Pauli matrices on $n$ qubits and $\alpha_\sigma$ are i.i.d.\ Gaussian coefficients (see the more formal definition in \Cref{eq:hamiltonian_concise_form} below). For classical spin glasses, seminal developments originating from the work of Parisi~\cite{parisi1979infinite} and later formalized by Talagrand and Panchenko~\cite{talagrand2006parisi,panchenko2013sherrington} provide a remarkably precise value for this maximal energy for every fixed $p$ via a solution of a certain variational problem. The answer simplifies dramatically
as $p$ increases (after the number of spins $n$ diverges) to the value $\sqrt{2\log\left(2\right)n}$, as was shown recently in~\cite{gamarnik2023shattering}.

Unlike the widely studied classical spin glass model which is now rigorously analyzed and well understood, its quantum counterpart remains notably underdeveloped. Here, we prove various results for the quantum spin model. Our main result shows that the maximum energy achievable by product states has a well-defined limit as $n\to\infty$ and is $(1+o_p(1))\sqrt{2 \log\left(p\right)n}$ in the limit of large $p$.
The limit value $\sqrt{2\log\left(p\right)n}$ is interpreted as an extreme value of an uncorrelated
sequence of Gaussians and arises naturally in the so-called Random Energy Model (REM)---a tractable simplification of classical spin glasses. 
We also take some first steps towards estimating the maximum energy achievable by \emph{all} states by upper bounding its limiting value under an asymptotic independence ansatz over graph matchings.
This ansatz states that the joint distribution of certain
functionals of a random matching in a sparse random hypergraph model is approximately of product form up to $\gamma(p)$ multiplicative approximation factors, which we bound as $\gamma(p)\leq 3^p$.
The value $\sqrt{\gamma(p)}$ corresponds almost directly to the ratio of approximation achievable by product states in comparison to the true ground energy.

Finally, we prove a universality result---namely, that the limiting maximal energy is asymptotically (i.e., as $n\to\infty$) independent from the specifics of the distribution of coefficients $\alpha_\sigma$. In particular, the energy remains the same if $\bm{\alpha}$ is supported on a sparse random hypergraph with arbitrarily slowly growing average degree. We prove our result for a broad range of local Hamiltonians bounded in norms, not necessarily corresponding to the Pauli operators. In particular, our result applies to the widely-studied SYK model~\cite{PhysRevLett.70.3339}. Our universality result in this context is particularly remarkable since it is known, for example, that the so-called Gaussian states achieve a constant factor approximation for sparse random hypergraphs~\cite{herasymenko2023optimizing}, while this is provably not the case for the original (dense) model~\cite{hastings2022optimizing}.

\subsection*{Informal statement of results}

We first give an explicit description of the quantum $p$-spin model we here consider. We adopt super-index notation $\indi =(i_1,\ldots,i_p)\in [n]^p$ and $\inda =(a_1,\ldots,a_p)\in \{1,2,3\}^p$, each of  which is a vector of length $p$. Letting $\mathcal{I}_p^n=\{\indi \in [n]^p:\indisub_1 < \indisub_2 < \cdots <\indisub_p \}$ denote the set of ordered $p$-tuples, the Hamiltonian of the quantum $p$-spin model can be written as:
\begin{equation} \label{eq:hamiltonian_concise_form}
     H_{n,p} = \frac{1}{{n \choose p}^{1/2}   } \sum_{\indi \in \mathcal{I}_p^n} \sum_{\inda \in \{1,2,3\}^p} \alpha[\indi; \inda] \; P_{\indi}^{\inda},
\end{equation}
where $\alpha[\indi; \inda]$ are i.i.d.\ standard normal coefficients and $P_{\indi}^{\inda} =\sigma_{\indisub_1}^{\indasub_1} \sigma_{\indisub_2}^{\indasub_2} \cdots \sigma_{\indisub_p}^{\indasub_p}$. We use the notation where $\sigma_i^a$ indicates Pauli matrix $\sigma^a$ acting on qubit $i$, where $1, 2, 3 $ correspond to the standard Pauli operators $ X, Y, Z$, respectively. We will later demonstrate that there is some flexibility in the choice of distribution for $\bm{\alpha}$. When the $\bm \alpha$ are specified, we denote the Hamiltonian with coefficients $\bm \alpha$ as $H_{n,p}(\bm \alpha)$.

An important quantity is the constant $E^*(p)$ corresponding to the largest average eigenvalue or maximum energy of the random Hamiltonian $H_{n,p}$ in the limit of large $n$:
\begin{equation} \label{eq:ground_energy_limit_def}
    E^*(p) \coloneqq \lim_{n \to \infty}E_n^*(p)\coloneqq\lim_{n \to \infty} \E \left[ \frac{1}{\sqrt{n}} \lambda_{\text{max}}(H_{n,p}) \right] = \lim_{n \to \infty} \E \left[ \frac{1}{\sqrt{n}} \max_{\ket{\phi} \in \mathcal{S}_{\rm all}^n} \bra{\phi}H_{n,p} \ket{\phi} \right].
\end{equation}
This limiting extremal value is characterized in two equivalent ways: as the largest eigenvalue $\lambda_{\text{max}}$ or as the maximum energy over the set $\mathcal{S}_{\rm all}^n$ of all quantum states on $n$ qubits.
In fact, this average limiting energy also holds in the high probability sense (namely, the model exhibits self-averaging) as standard concentration bounds can be applied (see \Cref{sec:concentration}) to show that in the limit of large $n$, the maximum energy $\lambda_{\text{max}}(H_{n,p})$ concentrates around its average. An analogous constant $E_{\rm product}^*(p)$ defined as:
\begin{equation}\label{eq:prod_state_limit_def}
    E_{\rm product}^*(p) \coloneqq \lim_{n \to \infty} \E \left[ \frac{1}{\sqrt{n}} \max_{\ket{\phi} \in \mathcal{S}_{\rm product}^n} \bra{\phi}H_{n,p} \ket{\phi} \right]
\end{equation}
quantifies the limiting maximum energy when optimizing over the set $\mathcal{S}_{\rm product}^n$ of all product states on $n$ qubits. 

Justifying the existence of the above limits $E^*(p)$ and $E_{\rm product}^*(p)$ is far from obvious. Fortunately, for even $p$, we can extend proofs  from classical spin glasses~\cite{guerra2002thermodynamic} (see \Cref{sec:limit_exists}) to show that the limit $E_{\rm product}^*(p)$ does indeed exist. Formally, we cannot guarantee that the limit in the definition of $E^*(p)$ (\Cref{eq:ground_energy_limit_def}) exists, and leave such a proof of existence as a conjecture (see \Cref{conj:limit_exists_full}). In stating our formal results, we will rely on $\limsup$ and $\liminf$ instead.

As stated previously, our results depend on the degree of which a certain trace functional of random matchings asymptotically depends only on the marginals of these matchings on individual qubits, which we state as a conjecture (formally given as \Cref{conj:main-0}).
\begin{conjecture}[Asymptotic matching ansatz (informal)]\label{conj:inf}
Consider the uniform distribution of matchings $M$ on $2r$ cardinality-$p$ superindices $\indi_1,\ldots,\indi_{2r}$. Given $M$, let $M_j$ be the induced matching on superindices containing $j$ as an element, and define:
\begin{equation}
    \Tr_{\rm sum}\left(M_j\right)\triangleq \frac{1}{2}\Tr\left(\sum\limits_{\sigma:M_j}\prod\limits_{1\leq k\leq 2r}\sigma_k\right),
\end{equation}
where $\sigma_k=\sigma_l$ whenever $k\sim l$ according to the induced matching $M_j$.

For every constant $C$ and $r\le Cn$, there exists a constant $\gamma(p)$ depending on $p$ but independent of $n$ and $r$ such that the following holds: 
\begin{equation*}
\E_{\indi_1, \dots, \indi_{2r}} \E_{M}\left[ \prod_{j\in [n]}\Trace_{\rm sum}(M_j)\right]\leq\gamma(p)^r \exp(O_p(1)n) \; \E_{\indi_1, \dots, \indi_{2r}} \left[ \prod_{j\in [n]} \E_{M_j} \left[\Trace_{\rm sum}(M_j) \right]\right].
\end{equation*}
\end{conjecture}
At this stage we do not have any rigorous treatment of the actual growth rate of $\gamma(p)$ other than an upper bound of $\gamma\left(p\right)\leq 3^p$. The essence of our conjecture is that the growth rate of $\gamma(p)$ is slower than $3^p$.
We leave the tightest choice of $\gamma\left(p\right)$ for which this bound holds  as an open question.

Product states possess no entanglement and have efficient classical descriptions. Thus comparing $E^*(p)$ and $E_{\rm product}^*(p)$ helps answer whether extrema of the random quantum  Hamiltonians are ``complex'' quantum states. Our main result stated---informally below---quantifies this in terms of the coefficient $\gamma(p)$ for large $p$: 
\begin{theorem}[Main result (informal)]\label{theorem:Main-result-informal}
    The limiting maximum energy attainable by product states $E_{\rm product}^*(p)$ (\Cref{eq:prod_state_limit_def}) as a function of locality $p$ satisfies $E_{\rm product}^*(p) = (1 + o_p(1))\sqrt{2 \log(p)}$, and furthermore
    \begin{equation} \label{eq:main_bounds}
        (1 + o_p(1))\sqrt{2 \log(p)}
=
E_{\rm product}^*(p) \leq E^*(p) \leq \sqrt{2\gamma(p)\log p},
    \end{equation}
    where  $\gamma(p)$ is from the ansatz of \Cref{conj:inf},  and
     $E^*(p)$ is the (conjectural) limit of the ground energy of the $p$-local quantum spin glass Hamiltonian (\Cref{eq:ground_energy_limit_def}).
\end{theorem}
In particular, the extent to which $E^*(p)$ matches $E_{\rm product}^*(p)$ is controlled by the ansatz parameter $\gamma(p)$. If it is the case that $\gamma(p)=1+o_p(1)$ then the ground state value is achieved by product states asymptotically.

In addition to the aforementioned result regarding the existence of the limit $E^*_{\rm product}(p)$ for every $p$,
we also establish distributional insensitivity (universality) of our results with respect to the choice of randomness of $\alpha$. In particular, in 
\Cref{sec:equivalence_distribution} we establish that the limit $\lim_p E^*(p)=\sqrt{2\gamma(p)\log p}$ holds also when $\alpha$ is a zero mean sparse random
variable, i.e., has support on sparse random hypergraphs. Here, sparsity refers to the fact that the likelihood $\alpha\ne 0$ is of the order $\omega(1/n^{p-1})$, implying order
$\omega(n)$ number of non-zero coefficients $\alpha$ appearing in the Hamiltonian. The term $\omega(n)$ underscores that the number of terms grows faster than $n$ at any arbitrarily slow rate.
This sparsity regime corresponds to the so-called sparse random hypergraph 
model with ``large'' average degree, also called the diluted spin glass model,  which was widely studied in the classical literature both 
in physics~\cite{franz2003replica} and the theory of random graphs~\cite{sen2018optimization}. To prove this result, 
we employ a variant of Lindeberg's interpolation 
technique~\cite{chatterjee2005simple}, which is also the main tool underlying~\cite{sen2018optimization}.

\subsection*{Proof sketch}

We now describe the proof idea behind \Cref{theorem:Main-result-informal}. It is based on obtaining an upper bound
on $E^*(p)$ and obtaining lower and upper bounds on $E^*_{\rm product}(p)$.

We first sketch the lower bound on $E^*_{\rm product}(p)$. We do this via the so-called second moment method. This is done by estimating the first two moments of a certain
counting quantity $N$ and then showing that $N\ge 1$ with high probabiltiy (whp) using the Paley--Zygmund inequality, provided the second moment
nearly matches the square of the first moment.
In our case the quantity of interest $N$
is the number of high-energy states within a set of order approximately $p^n$ roughly evenly spaced product states $\ket{\psi}$ achieving normalized energy
approaching $\sqrt{2\log p}$. These states are chosen by taking the $n$-fold tensor product of an $\epsilon$-net over the single qubit pure states of size roughly $O(p)$. We are able to calculate the first and second moment of $N$ since the energy associated with each $\ket{\psi}$ is a Gaussian
with mean $0$, variance $1/n$, and covariance described by the overlaps within the states for each of the individual qubits. The spacing of the $\epsilon$-net is chosen so that differences in the overlap of at least $\epsilon$ in individual qubits render the covariance between pairs of state to be sufficiently uncorrelated. Note that a trivial
upper bound on the first moment is obtained as the extreme of order $p^n$ variance-$1/n$ Gaussians, which is
$(1/\sqrt{n})\sqrt{2\log p^n}=\sqrt{2\log p}$. This number is then justified by controlling the second moment.
This approach follows an approach recently developed for purely classical
spin glass models with large $p$~\cite{gamarnik2023shattering}. 
This is the special case of our model
when all Paulis are restricted to be identical, which restricts the number of configurations to $2^n$ leading to the asymptotic maximum eigenenergy $\sqrt{2\log 2}$. 
Our proof obtains a bound $\E[N^2]=\exp\left(O_p\left(1\right)n\right)\E[N]^2$ (where $O_p\left(1\right)$ hides constants in $p$ and $n$), and the claimed result follows by adopting a trick introduced by Frieze~\cite{FriezeIndependentSet} which takes advantage of concentration inequalities.

The upper bound on $E_{\text{product}}^*\left(p\right)$ uses a similar net construction, and then uses Dudley's inequality to bound how the energy varies away from the net.

We now sketch the upper bound on $E^*(p)$. We consider the expected
trace of the partition function with a judicious choice of the inverse temperature of the form $\beta\sqrt{n}$ for a constant
$\beta>0$.
By Jensen's inequality
this provides an upper bound on the (exponent of the) largest eigenvalue. The partition function is then expanded using the Taylor series. This is in fact a standard method to analyze $H_{n,p}$, used
both in~\cite{erdHos2014phase,baldwin2020quenched} for quantum spin glasses and in~\cite{feng2019spectrum} for a related SYK model.
In both cases they study the free energy with temperature on the order $O(1)$, such that the dominating terms in the Taylor
expansion come from $O(1)$ powers of the Hamiltonian. In our case, as we prove, the dominating contribution comes from powers of the Hamiltonian that are linear in $n$.

The associated terms in the Taylor expansion are expected traces of $O\left(n\right)$-fold products of Pauli operators. We are able to rewrite this as an expectation over matchings of $2m$ $p$-local Pauli operators---with induced matching $M_i$ on Pauli operators with support on qubit $i$---of a certain trace object that depends on the induced matchings $M_i$. Though we are unable to treat this exactly, we bound the associated quantity by some factor multiplied by the equivalent expression if the $M_i$ were chosen independently; this is informally stated as \Cref{conj:inf} (and formally stated as \Cref{conj:main-0}). The degree to which this asymptotic independence ansatz on matchings is incorrect---i.e., the looseness of this bound---directly impacts our associated upper bound on $E^\ast$ through the $p$-dependent constant $\gamma\left(p\right)$. We show here that $\gamma\left(p\right)\leq 3^p$ and leave its true value as an open question.

Upon applying the conjecture, for the $r$-th moment in the trace, the contribution over all qubits splits into an expectation over individual qubits. 
For each qubit, based on our prior closed-form calculation, its contribution is then a function of the of the number of tuples in the $r$-th moment that contain the given qubit, which is (roughly) a Poisson random variable with mean concentrated on the order of $\frac{pr}{n}$. 
From this, we obtain that the expected trace grows as a function $\exp(g(\beta, p)n)$ with dominant term $\frac{\beta^2\gamma\left(p\right)}{2}+\log(1+p\gamma\left(p\right)\beta^2)$.
Optimizing over the choice of $\beta$ achieves our
claimed upper bound $(1+o_p(1))\sqrt{2\gamma\left(p\right)\log p}$.

\subsection*{Algorithmic (in)tractability of classical and quantum spin glasses}

The energy achievable by product states indicates a certain kind of tractability associated to the representation complexity of states at a given energy. 
Here, product states are often called ``trivial'' as they can be constructed by constant depth (specifically depth-1) quantum circuits.
Triviality and non-triviality of states underlies quantum state representation
complexity questions, such as the validity of the NLTS (No-Low-Energy-Trivial-State) conjecture which was recently positively resolved
in~\cite{anshu2023nlts}. By no means, however, does the triviality of near-ground states imply algorithmic tractability. 
Namely, the algorithmic complexity of generating low energy trivial (i.e., product) states is far from clear. Recent
advances based on a geometric property of solution spaces called the overlap gap property has ruled out large classes of algorithms
as contenders in the classical setting---even in the setting when the disorder $\alpha$ is randomly generated as in our setting---and in fact  provide strong evidence that no polynomial time algorithm exists for 
this problem~\cite{huang2022tight,gamarnik2021overlap,gamarnik2022disordered,gamarnik2023shattering,gamarnik2021overlapMessagePassing,gamarnik2020low}.

We emphasize that our results have no bearing on, e.g., the $\textsc{NP}$-hardness of the ground-state problem for the $p$-spin glass model as $\textsc{NP}$-hardness is a statement of the worst-case complexity of the problem---our results are whp over the disorder $\bm{\alpha}$. This distinguishes this work from previous results studying product-state approximations of the ground state of $p$-spin models, which previously had been considered in a worst-case setting. For example, \cite{doi:10.1137/110842272} and \cite{10.1145/2488608.2488719} studied worst-case product-state approximations for dense $p$-local models and give efficient classical algorithms for finding the approximations; however, the associated guarantees on the multiplicative error grows with $n$ for the model we consider here due to its highly frustrated nature. \cite{Harrow2017extremaleigenvalues} gives bounds for sparse models which also scale incorrectly with $n$ in the average case. \cite{10.1063/1.5085428} manages a multiplicative approximation to the ground state in the worst case, though only study $p=2$-local Hamiltonians.

\subsection*{Open questions}

We put forward now several open questions for future investigation. 
The first question of interest is obtaining estimates of the ground state energy for fixed values of $p$, e.g., $p=2$. This is open even when the optimization is over product states $\mathcal{S}_{\rm product}^n$, of which the classical spin glass model is a special case.
For this one would need to have a tighter control on the actual (quenched) value of the partition function as opposed to its
annealed (expectation) approximation. This direction, we expect, either requires tools from Parisi's ansatz~\cite{parisi1979infinite,parisi1980sequence} and their rigorization by Guerra--Toninelli~\cite{guerra2002thermodynamic,talagrand2006parisi} or the use of the Kac--Rice formula with a control on its second moment~\cite{10.1214/16-AOP1139}.

Second, future work is needed to determine better upper bounds for $E^*(p)$, the limiting expected maximum energy achievable by all states. One way to achieve this would be to more tightly control the coefficient $\gamma(p)$ used in our ansatz \Cref{conj:main-0} 
improving over  the trivial bound $\gamma(p)\leq 3^p$.

Another question is whether 
our result extends to other similar quantum spin models, of which the Sachdev-Ye-Kitaev (SYK) model \cite{sachdev1993gapless,rosenhaus2019introduction} is of particular
importance due to its physics applications. Also of potential interest are whether our results extend to quantum analogues of classical optimization problems such as quantum versions of random $k$-SAT (denoted $k$-QSAT) \cite{bravyi2009bounds,bravyi2011efficient,ambainis2012quantum,laumann2009phase} and quantum max-cut \cite{gharibian2019almost,hwang2023unique}.

\subsection*{Structure of the paper}
The remainder of the paper is structured as follows. We formally detail our main results in \Cref{sec:main_results}. In \Cref{sec:preliminaries}, we prove preliminary technical results on the concentration and covariance of the energies of the Hamiltonians which we use throughout this work. Then, in \Cref{sec:quantum_p_spin_proof}, we prove \Cref{theorem:Main-result-formal} (formal statement of \Cref{theorem:Main-result-informal}) about the limiting maximum energy of the quantum $p$-spin model. \Cref{sec:limit_exists} and \Cref{sec:equivalence_distribution} contain the proofs of the existence of the limit (\Cref{theorem:limit-esist}) and the universality under different choices of distribution (\Cref{theorem:equivalence}), respectively. Finally, we study the statistical properties of entangled states in \Cref{sec:properties_entangled_states} to show entangled states typically have exponentially smaller variance in comparison with product states, thus giving some intuition as to why product states are ``special'' as approximates of the ground state. 

\section{Main results}
\label{sec:main_results}

We denote the set of $n$ elements as $[n] \coloneqq \{1, 2, \dots,n\}$. Throughout this study, $n$ refers to the number of qubits. 
Systems of $n$ qubits live on the vector space $\mathbb{C}^{2^n}$. A Hamiltonian is a Hermitian matrix acting on this vector space, i.e. a matrix $H \in \mathbb{C}^{2^n \times 2^n}$ obeying the Hermitian relation $H^\dagger = H$. We denote the largest  eigenvalue of a matrix $M$ as  $\lambda_{\text{max}}(M)$. 

The Pauli matrices form a basis for the space of Hermitian matrices acting on $n$ qubits. For a single qubit, there are four Pauli matrices (Paulis $I$, $X$, $Y$, and $Z$):
\begin{equation}
    \sigma^0  =
    \begin{pmatrix}
      1&0\\
      0&1
    \end{pmatrix}, \quad \quad
    \sigma^1  =
        \begin{pmatrix}
        0&1\\
        1&0
        \end{pmatrix}, \quad \quad
    \sigma^2  =
        \begin{pmatrix}
        0& -i \\
        i&0
        \end{pmatrix}, \quad \quad
    \sigma^3 =
        \begin{pmatrix}
        1&0\\
        0&-1
        \end{pmatrix}.
\end{equation}
For $n$ qubits, the set of Pauli matrices $\mathcal{P}_n$ is the set of all $4^n$ possible $n$-fold tensor products of Pauli matrices, 
\begin{equation}
    \mathcal{P}_n = \{ \sigma^{a_1} \otimes \cdots \otimes \sigma^{a_n}: a_1, \dots, a_n \in \{0,1,2,3\} \}.
\end{equation}
If a Pauli matrix equals the identity $\sigma^0$ at all but a single qubit, we let $\sigma_{i}^a \in \mathcal{P}_n$ denote such a Pauli matrix where $\sigma^a$ acts on qubit $i$ tensored with the identity Pauli $\sigma^0$ on all other qubits. Given superindices $\indi=(i_1,\ldots,i_p) \in [n]^p$ and $\inda=(a_1,\ldots,a_p) \in \{0,1,2,3\}^p$, we define the Pauli matrix $P_{\indi}^{\inda} \in \mathcal{P}_n$ as
\begin{equation}
    P_{\indi}^{\inda} = \prod_{k=1}^p \sigma_{i_k}^{a_k}.
\end{equation}
As a simple example, for $n=5$ qubits, superindices $\indi = [2,4], \inda = [1,3]$ denote the Pauli matrix $P_{\indi}^{\inda} = \sigma^0 \otimes \sigma^1 \otimes \sigma^0 \otimes \sigma^3 \otimes \sigma^0$. Any Pauli matrix that can be written in such a form containing exactly $p$ elements in its tensor product not equal to $\sigma^0$ is called a $p$-local Pauli matrix. Similarly, a Hamiltonian is called a $p$-local Hamiltonian if it can be written as a linear combination of (at most) $p$-local Pauli matrices.

For $n$ qubits, we use bra-ket notation to denote states where a ket $\ket{\phi}$ is a vector in the vector space $\ket{\phi} \in \mathbb{C}^{2^n}$ and $\bra{\phi}$ denotes its conjugate transpose. 
We denote the set of all  pure states by $\mathcal{S}_{\rm all}^n = \{ \ket{\phi} \in \mathbb{C}^{2^n}: |\braket{\phi }|^2 = 1\}$. A state $\ket{\phi}$
is a product state on $n$ qubits if it can be written as a tensor product of states $\ket{\phi^{(i)}}$ on its individual qubits $i\in [n]$:
$\ket{\phi} = \ket{\phi^{(1)}} \otimes \cdots \otimes \ket{\phi^{(n)}}$. The set of such product states on $n$ qubits we denote as 
\begin{equation}
    \mathcal{S}^n_{\rm product} = \{ \ket{\phi} = \ket{\phi^{(1)}} \otimes \cdots \otimes \ket{\phi^{(n)}} \in \mathcal{S}_{\rm all}^n:\;  
    \ket{\phi^{(i)}} \in \mathcal{S}_{\rm all}^1, \forall i \in [n] \}.
\end{equation}

We now introduce the following $n$-dependent prelimit analogues of \Cref{eq:ground_energy_limit_def} and \Cref{eq:prod_state_limit_def}
with respect to the same Hamiltonian $H_{n,p}$ as defined in \Cref{eq:hamiltonian_concise_form}. 
\begin{align} 
    E^*_n(p) & \coloneqq   \frac{1}{\sqrt{n}} \lambda_{\text{max}}(H_{n,p})  =  \frac{1}{\sqrt{n}} \max_{\ket{\phi} \in \mathcal{S}_{\rm all}^n} \bra{\phi}H_{n,p} \ket{\phi}, \label{eq:ground_energy_n_def} \\
    E_{n,\rm prod}^*(p) & \coloneqq   \frac{1}{\sqrt{n}} \max_{\ket{\phi} \in \mathcal{S}_{\rm product}^n} \bra{\phi}H_{n,p} \ket{\phi} 
    \label{eq:prod_state_n_def}.
\end{align}
Note the absence of the expectation operator in the definitions above. In particular $ E_{n,\rm prod}^*(p)\le E^*_n(p)$ are random variables, with inequalities holding trivially. We will show below, however, that these random variables are highly concentrated around their means.

We now state our main result. 

\paragraph{Statement of main results}
We now state our main result regarding the energy achievable by product states.

\begin{theorem}[Main result (formal)]\label{theorem:Main-result-formal} 
For every $\epsilon>0$ there exists $p_0$ such that for all $p\ge p_0$ the following holds with high probability 
with respect to the randomness of $\alpha=\left( \alpha[\indi; \inda], \indi\in \mathcal{I}_p^n, \inda\in \{1,2,3\}^p\right)$
as $n\to\infty$:
\begin{align}
(1-\epsilon)\sqrt{2\log p}\leq E_{n,\rm prod}^*(p)\leq(1+\epsilon)\sqrt{2\log p}.
\end{align}
Furthermore, assuming the the validity of \Cref{conj:main-0}, for every $\epsilon>0$ there exists $p_0$ such that for all $p\ge p_0$, the following holds again with high probability as $n\to\infty$.
\begin{align}
E^*_n(p)
\le (1+\epsilon)\sqrt{2\gamma(p) \log p}.
\end{align}
\end{theorem}
As above,
$\gamma\left(p\right)$ is a constant characterizing our asymptotic independence ansatz over graph matchings described informally in \Cref{conj:inf}; we discuss this constant in more detail in \Cref{sec:quantum_p_spin_upper_bound}.

Next we turn to our results regarding the existence of limits of random variables 
$E^*_n(p)$ and $E_{n,\rm prod}^*(p)$.
While we leave open the question of the existence of the limit of $E^*_n(p)$, we can confirm these limits for the latter variables.

\begin{theorem}\label{theorem:limit-esist}
    For every even $p$ there exist constants
    $E^*_{\rm prod}(p)$
    such that whp as $n\to\infty$,
    $E_{n,\rm prod}^*(p)\to E^*_{\rm prod}(p)$.
\end{theorem}
The status of this question for $E^*_{n}(p)$ remains
a conjecture.
\begin{conjecture}\label{conj:limit_exists_full}
    For every $p$ there exist constants 
    $E^*(p)$ 
    such that whp as $n\to\infty$,
    $E_{n}^*(p)\to E^*(p)$.
\end{conjecture}

Finally, we turn to our universality results. Namely,
insensitivity with respect to the choice of distribution
of the disorder $\bm{\alpha}$. Our results hold for a broad range of 
Hamiltonians consisting of a sum of arbitrary local bounded Hamiltonians
weighted by random coefficients. 
More precisely, consider a sequence of random Hamiltonians specified by $\{m_n,H_n,\left\{\mathcal{D}_{n,i}\right\}_i\}_n$ and constructed as follows (see also formal definition in \Cref{def:sequence_of_random_hamiltonians}). 
This sequence is first specified by 
a sequence of positive integers $\{m_n\}_n$ corresponding to the number of terms in the Hamiltonian where we require $m_n = \omega(n)$. The sequence is furthermore specified by
an arbitrary sequence of
$m_n$   operators $H_n^{(1)}, \dots, H_n^{(m_n)}$ 
each bounded in operator norm by $\|H_n^{(1)}\| \leq 1$ 
and random variables $\bm \alpha_n = \{\alpha_{n,1}\sim\mathcal{D}_{n,1}, \ldots, \alpha_{n, m_n}\sim\mathcal{D}_{n,m_n}\}$ drawn independently 
from a collection of distributions $\left\{\mathcal{D}_{n,i}\right\}_i$ satisfying
\begin{equation} \label{eq:tail-D_n}
\begin{split}
    &\E[\alpha_{n,i}]=0, \\
&\E[\alpha_{n,i}^2]=1, \\
&
\forall \epsilon > 0, \; \lim_{n \to \infty} \frac{1}{m_n}\sum_{i=1}^{m_n}  \mathbb{E}\left[(\alpha'_{n,i})^2 ;  |\alpha'_{n,i}|  > \epsilon \sqrt{\frac{m_n}{n} } \right] = 0.
\end{split}
\end{equation}
From these specifications, random Hamiltonians take the form 
    \begin{equation}
        H_n(\bm \alpha_n) = \frac{1}{m_n^{1/2}} \sum_{i=1}^{m_n} \alpha_{n,i} H_n^{(i)}.
    \end{equation} 
When the $\mathcal{D}_{n,i}$ are unspecified, disorder terms $\alpha_{n,1}, \dots \alpha_{n, m_n}$ are assumed to be drawn i.i.d.\ standard Gaussian.
Sequences of random Hamiltonians as above encompass the spin glass Hamiltonians we study here and the SYK model \cite{feng2019spectrum,rosenhaus2019introduction}. 
Our universality statement guarantees that 
models with different distributions of disorder will have
the same limiting ground energy as long as the distributions meet the requirements of \Cref{eq:tail-D_n}.

\begin{restatable}[Equivalence of Hamiltonian models]{theorem}{equivalence}
\label{theorem:equivalence}
    Consider a $\{m_n,H_n,\left\{\mathcal{N}\left(0,1\right)\right\}_i\}_n$ sequence of random Hamiltonians 
    constructed as above
    (see also \Cref{def:sequence_of_random_hamiltonians}), where
    \begin{equation}
        H_n(\bm \alpha_n) = \frac{1}{m_n^{1/2}} \sum_{i=1}^{m_n} \alpha_{n,i} H_n^{(i)}, \quad \quad \| H_n^{(i)} \| \leq 1, \quad \quad m_n = \omega(n).
    \end{equation}
    In particular, let disorder terms be drawn i.i.d. $\alpha_{n,i} \sim \mathcal{N}(0,1)$ for all $i \in [m_n]$. Let $\{m_n,H_n,\{\mathcal{D}_{n,i}\}_{i=1}^{m_n}\}_n$ correspond to a  sequence of random Hamiltonians where disorder coefficients $\alpha_{n,i}$ are replaced by random variables $\alpha'_{n,i} \sim \mathcal{D}_{n,i}$ drawn independently from a collection of distributions $\{\mathcal{D}_{n,i}\}_{i=1}^{m_n}$. Over the set of quantum states $\mathcal{S}_{\rm all}^n$ on $n$ qubits, the limiting average maximum energies are equal under the two distributions
    \begin{equation}
        \lim_{n \to \infty} \E_{\alpha_{n,i} \sim \mathcal{N}(0,1)} \left[ \frac{1}{\sqrt{n}} \max_{\ket{\phi} \in \mathcal{S}_{\rm all}^n} \bra{\phi}H_{n}(\bm \alpha_n) \ket{\phi} \right] = \lim_{n \to \infty} \E_{\alpha'_{n,i} \sim \mathcal{D}_{n,i}} \left[ \frac{1}{\sqrt{n}} \max_{\ket{\phi} \in \mathcal{S}_{\rm all}^n} \bra{\phi}H_{n}(\bm \alpha'_n) \ket{\phi} \right] 
    \end{equation}
    whenever the sequence of distributions $\mathcal{D}_{n,i}$ meets the following three conditions:
    \begin{equation}
        \begin{split}
            (\text{first moment}): & \quad \E[\alpha'_{n,i}] =  0, \\
            (\text{second moment}): & \quad \E[(\alpha'_{n,i})^2]  = 1, \\
            (\text{boundedness}): & \quad \forall \epsilon > 0, \; \lim_{n \to \infty} \frac{1}{m_n}\sum_{i=1}^{m_n}  \mathbb{E}\left[(\alpha'_{n,i})^2 ;  |\alpha'_{n,i}|  > \epsilon \sqrt{\frac{m_n}{n} } \right] = 0.
        \end{split}
    \end{equation}
\end{restatable} 

Specializing to the spin glass model,
a canonical model where
our theorem above applies
is the setting of a sparse random graph. 
Let $\alpha'[\indi, \inda]$ be distributed according to
\begin{equation}
    \alpha'[\indi, \inda] = \begin{cases}
    0 & \text{ with probability } 1- 3^{-p}d_n/{n\choose p-1} \\
    +\sqrt{3^p{n\choose p-1}\over d_n} & \text{ with probability }  (1/2)3^{-p}d_n/{n\choose p-1}  \\
    -\sqrt{3^p{n\choose p-1}\over d_n} & \text{ with probability }  (1/2)3^{-p}d_n/{n\choose p-1} \\
    \end{cases}
\end{equation}
where $d_n$ corresponds to the average degree and is assumed to be an arbitrarily slowly growing function of $n$. 
In other words, the random Hamiltonian $H_n(\bm \alpha_n')$ has on average $d_n n$ nonzero terms.
One can easily verify that $\E[\alpha'[\indi, \inda]]=0$, $\E[\alpha'[\indi, \inda]^2]=1$ and $\mathbb{P}\left(\alpha'[\indi, \inda]>
n^{p-1\over 2}\epsilon \right)$
goes to zero for $d_n$ large enough 
guaranteeing the third condition in \Cref{eq:tail-D_n} holds trivially for large enough $n$. \Cref{theorem:equivalence} then states that
 \begin{align}
\lim_{n \to \infty} \E_{\alpha[\indi, \inda] \sim \mathcal{N}(0,1)} \left[ \frac{1}{\sqrt{n}} \max_{\ket{\phi} \in \mathcal{S}_{\rm all}^n} \bra{\phi}H_{n,p}(\bm \alpha) \ket{\phi} \right] = \lim_{n \to \infty} \E_{\alpha'[\indi,\inda] \sim \mathcal{D}_n} \left[ {1\over \sqrt{n}}\max_{\ket{\phi} \in \mathcal{S}_{\rm all}^n} \bra{\phi} H_{n,p}( \bm \alpha') \ket{\phi} \right].
\end{align}
In fact, as in the classical setting, our proof implies the following double limit identity. Suppose $d$ is a constant not depending on $n$. Define the associated distribution $\mathcal{D}_{n,d}$ of the disorder as above with $d$ replacing $d_n$.
Then 
 \begin{align}
\lim_{n \to \infty} \E_{\alpha[\indi, \inda] \sim \mathcal{N}(0,1)} \left[ \frac{1}{\sqrt{n}} \max_{\ket{\phi} \in \mathcal{S}_{\rm all}^n} \bra{\phi}H_{n,p}(\bm \alpha) \ket{\phi} \right] = \lim_{d\to\infty}\lim_{n \to \infty} \E_{\alpha'[\indi,\inda] \sim \mathcal{D}_{n,d}} \left[{1\over \sqrt{n}}\max_{\ket{\phi} \in \mathcal{S}_{\rm all}^n} \bra{\phi} H_{n,p}(\bm \alpha') \ket{\phi} \right],
\end{align}
which is a more common way of expressing this universality in the classical setting. 
We do not bother with this formal strengthening for simplicity.

\section{Preliminary technical results}
\label{sec:preliminaries}

\subsection{Variance and covariance of energies}

An important quantity in our study is the covariance of the energy of the random Hamiltonian with respect to given states $\ket{\phi}$ and $\ket{\psi}$. \Cref{lem:covariance_product} quantifies this covariance for product states and is crucially used throughout our proofs in later sections. For completeness, we also calculate variance and covariance for the energy over arbitrary and potentially entangled states in \Cref{sec:properties_entangled_states} and comment on how entanglement typically reduces the magnitude of the variance there.

\begin{lemma} \label{lem:covariance_product}
    For any two product states $\ket{\phi}, \ket{\psi} \in \mathcal{S}_{\rm product}^n$ where $\ket{\phi}=\ket{\phi^{(1)}} \otimes \cdots \otimes \ket{\phi^{(n)}}$ and $\ket{\psi}= \ket{\psi^{(1)}} \otimes \cdots \otimes \ket{\psi^{(n)}}$:
    \begin{equation}
        \E\left[\bra{\phi} H_{n,p} \ket{\phi} \bra{\psi} H_{n,p} \ket{\psi}\right] = \frac{1}{\binom{n}{p}} \sum_{\indi \in \mathcal{I}_p^n} \prod_{k=1}^p \left( 2 \left|\bra{ \phi^{(\indisub_k)}}\ket{\psi^{(\indisub_k)}} \right|^2 - 1\right).
    \end{equation}
\end{lemma}
\begin{proof}
    By taking the expectation over coefficients, it is straightforward to note that
    \begin{equation} \label{eq:covariance_first_step}
        \E\left[ \bra{\phi} H_{n,p} \ket{\phi} \bra{\psi} H_{n,p} \ket{\psi}\right] = \frac{1}{\binom{n}{p}} \sum_{\indi \in \mathcal{I}_p^n} \sum_{\inda \in \{1,2,3\}^p} \bra{\phi} P_{\indi}^{\inda} \ket{\phi} \bra{\psi} P_{\indi}^{\inda} \ket{\psi}.
    \end{equation}

    Since $\ket{\phi}$ is a product state, we denote its tensor components as $\ket{\phi} = \ket{\phi^{(1)}} \otimes \cdots \otimes \ket{\phi^{(n)}}$ and similarly for $\ket{\psi}$. Looking at each term within the first sum, for each $\indi$:
\begin{equation}
    \begin{split}
        \sum_{\inda \in \{1,2,3\}^p} \bra{\phi} P_{\indi}^{\inda} \ket{\phi} \bra{\psi} P_{\indi}^{\inda} \ket{\psi} &= \sum_{\indasub_1 \in \{1,2,3\}} \bra{\phi^{(\indisub_1)}} \sigma_{\indasub_1} \ket{\phi^{(\indisub_1)}} \bra{\psi^{(\indisub_1)}} \sigma_{\indasub_1} \ket{\psi^{(\indisub_1)}} \\
        & \quad \times \sum_{\indasub_2 \in \{1,2,3\}} \bra{\phi^{(\indisub_2)}} \sigma_{\indasub_2} \ket{\phi^{(\indisub_2)}} \bra{\psi^{(\indisub_2)}} \sigma_{\indasub_2} \ket{\psi^{(\indisub_2)}} \\
        & \quad \times \cdots \\
        & \quad \times \sum_{\indasub_p \in \{1,2,3\}} \bra{\phi^{(\indisub_p)}} \sigma_{\indasub_p} \ket{\phi^{(\indisub_p)}} \bra{\psi^{(\indisub_p)}} \sigma_{\indasub_p} \ket{\psi^{(\indisub_p)}} .
    \end{split}
\end{equation}

For each individual sum on the right hand side, we have 
\begin{equation} \label{eq:intermediate_before_swap}
    \sum_{\indasub_k \in \{1,2,3\}} \bra{\phi^{(\indisub_k)}} \sigma_{\indasub_k} \ket{\phi^{(\indisub_k)}} \bra{\psi^{(\indisub_k)}} \sigma_{\indasub_k} \ket{\psi^{(\indisub_k)}} =  \bra{\phi^{(\indisub_k)}} \bra{\psi^{(\indisub_k)}} \left( \sum_{\indasub_k \in \{1,2,3\}} \sigma_{\indasub_k} \otimes \sigma_{\indasub_k}  \right)\ket{\phi^{(\indisub_k)}} \ket{\psi^{(\indisub_k)}}.
\end{equation}

Defining the $\operatorname{SWAP}$ operator as 
\begin{equation}
    \operatorname{SWAP} = \begin{bmatrix} 1&0&0&0\\0&0&1&0\\0&1&0&0\\0&0&0&1\end{bmatrix},
\end{equation}
we note that
\begin{equation}\label{eq:sum_pauli_to_swap}
    \sum_{a \in \{1,2,3\}} \sigma^{a} \otimes \sigma^{a} =\begin{bmatrix} 0&0&0&1\\0&0&1&0\\0&1&0&0\\1&0&0&0\end{bmatrix}+\begin{bmatrix} 0&0&0&-1\\0&0&1&0\\0&1&0&0\\-1&0&0&0\end{bmatrix}+\begin{bmatrix} 1&0&0&0\\0&-1&0&0\\0&0&-1&0\\0&0&0&1\end{bmatrix}= 2(\operatorname{SWAP}) - I.
\end{equation}
Therefore, returning to \Cref{eq:intermediate_before_swap}:
\begin{equation}
\begin{split}
    \bra{\phi^{(\indisub_k)}} \bra{\psi^{(\indisub_k)}} \left( \sum_{\indasub_k \in \{1,2,3\}} \sigma_{\indasub_k} \otimes \sigma_{\indasub_k}  \right)\ket{\phi^{(\indisub_k)}} \ket{\psi^{(\indisub_k)}} &= \bra{\phi^{(\indisub_k)}} \bra{\psi^{(\indisub_k)}} \left[  2(\operatorname{SWAP}) - I \right]\ket{\phi^{(\indisub_k)}} \ket{\psi^{(\indisub_k)}} \\ 
    &= 2 \left|\bra{\phi^{(\indisub_k)}}\ket{ \psi^{(\indisub_k)}} \right|^2 - 1,
\end{split}
\end{equation}
and thus
\begin{equation}
    \sum_{\inda \in \{1,2,3\}^p} \bra{\phi} P_{\indi}^{\inda} \ket{\phi} \bra{\psi} P_{\indi}^{\inda} \ket{\psi} = \prod_{k=1}^p \left( 2 \left|\bra{\phi^{(\indisub_k)}}\ket{ \psi^{(\indisub_k)}} \right|^2 - 1\right).
\end{equation}

Returning to \Cref{eq:covariance_first_step}, we now have that
\begin{equation} \label{eq:variance_lemma_final}
    \begin{split}
        \E\left[\bra{\phi} H_{n,p} \ket{\phi} \bra{\psi} H_{n,p} \ket{\psi}\right] &= \frac{1}{\binom{n}{p}} \sum_{\indi \in \mathcal{I}_p^n} \sum_{\inda \in \{1,2,3\}^p} \bra{\phi} P_{\indi}^{\inda} \ket{\phi} \bra{\psi} P_{\indi}^{\inda} \ket{\psi} \\
        &= \frac{1}{\binom{n}{p}} \sum_{\indi \in \mathcal{I}_p^n} \prod_{k=1}^p \left( 2 \left|\bra{\phi^{(\indisub_k)}}\ket{ \psi^{(\indisub_k)}} \right|^2 - 1\right).
    \end{split}
\end{equation}

\end{proof}

As a consequence of the above, for any product state $\ket{\phi} \in \mathcal{S}_{\rm product}^n$, the variance is normalized so that $\E\left[\bra{\phi} H_{n,p} \ket{\phi} \bra{\phi} H_{n,p} \ket{\phi}\right] = 1$.

\subsection{Concentration}
\label{sec:concentration}

Throughout this study, we aim to study properties of ``typical" random Hamiltonians $H_{n,p}$ by studying these properties on average. In the case of the largest eigenvalue, standard Gaussian concentration inequalities can be applied to show that $\lambda_{\text{max}}(\sqrt{n} H_{n,p})$ is concentrated around its mean.
\begin{proposition}[Concentration of largest eigenvalue] \label{prop:spin_glass_concentration}
    For all $t\geq 0$:
    \begin{equation}
        \mathbb{P}\left[ \left| \lambda_{\text{max}}(\sqrt{n} H_{n,p}) - \E \lambda_{\text{max}}(\sqrt{n} H_{n,p}) \right| \geq tn \right] \leq 2 \exp\left( -\frac{t^2n}{2 \cdot 3^p} \right).
    \end{equation}
\end{proposition}

To prove this, we apply the Gaussian concentration inequality for Lipschitz functions. A function $f:\mathbb{R}^N \to \mathbb{R}$ is $L$-Lipschitz if $|f(x) - f(y)| \leq L \|x - y\|_2$ for all $x,y \in \mathbb{R}^n$. 

\begin{lemma}[Gaussian concentration inequality for Lipschitz functions; Theorem 2.26 of \cite{wainwright2019high}] \label{lem:gaussian_concentration}
    Let $X_1, \dots, X_N$ be i.i.d.\ standard Normal variables and $f:\mathbb{R}^N \to \mathbb{R}$ be an $L$-Lipschitz function. Then for any $t \geq 0$:
    \begin{equation}
        \mathbb{P}\left[ \left| f(X) - \E f(X) \right| \geq t \right] \leq 2 \exp\left( -\frac{t^2}{2L^2} \right).
    \end{equation}
\end{lemma}

\begin{proof}[Proof of \Cref{prop:spin_glass_concentration}]
    Applying \Cref{lem:gaussian_concentration} with $f(\bm \alpha) = \lambda_{\text{max}}(\sqrt{n} H_{n,p}(\bm \alpha))$, it suffices to show that the function $\lambda_{\text{max}}(\sqrt{n} H_{n,p}(\bm \alpha))$ has Lipschitz constant bounded by $n$. Note that for any given state $\ket{\phi}$ 
    \begin{equation} \label{eq:lipschitz_from_hamiltonian}
        \|\nabla_{\bm \alpha} \bra{\phi} \sqrt{n} H_{n,p}(\bm \alpha) \ket{\phi}  \|_2^2 = \frac{n}{\binom{n}{p}} \sum_{\indi \in \mathcal{I}_p^n} \sum_{\inda \in \{1,2,3\}^p} \left(\bra{\phi} P_{\indi}^{\inda} \ket{\phi}\right)^2.
    \end{equation}
    Since $\left(\bra{\phi} P_{\indi}^{\inda} \ket{\phi}\right)^2 \leq 1$ for any Pauli matrix $P_{\indi}^{\inda}$, this implies that  $\|\nabla_{\bm \alpha} \bra{\phi} \sqrt{n} H_{n,p} \ket{\phi}  \|_2^2 \leq 3^p n$ and the Lipschitz constant of the function $\bra{\phi} \sqrt{n} H_{n,p}(\bm{\alpha}) \ket{\phi} $ is bounded by $L \leq \sqrt{3^p n}$.
    
    The function $\lambda_{\text{max}}(\sqrt{n} H_{n,p}(\bm \alpha))$ also has the same Lipschitz bound $L \leq \sqrt{3^p n}$. This can be shown by the following argument. $\lambda_{\text{max}}(\sqrt{n} H_{n,p}(\bm \alpha))$ is the maximum of $\bra{\phi} \sqrt{n} H_{n,p}(\bm \alpha) \ket{\phi} $ over all states $\ket{\phi}$. For a given draw of coefficients $\bm \alpha$, let $\ket{\psi}$ be a state which achieves the optimum $\bra{\psi} \sqrt{n} H_{n,p}(\bm \alpha) \ket{\psi} = \lambda_{\text{max}}(\sqrt{n} H_{n,p}(\bm \alpha))$; then, 
    \begin{equation}
        \lambda_{\text{max}}(\sqrt{n} H_{n,p}(\bm \alpha)) \leq  \bra{\psi} \sqrt{n} H_{n,p}(\bm \alpha') \ket{\psi} + L \|\bm \alpha - \bm \alpha'\| \leq \lambda_{\text{max}}(\sqrt{n} H_{n,p}(\bm \alpha')) + L \|\bm \alpha - \bm \alpha'\|.
    \end{equation}
    Applying a symmetric bound, we similarly have that $\lambda_{\text{max}}(\sqrt{n} H_{n,p}(\bm \alpha)) \geq \lambda_{\text{max}}(\sqrt{n} H_{n,p}(\bm \alpha')) - L \|\bm \alpha - \bm \alpha'\|$.
    This then implies that $\left| \lambda_{\text{max}}(\sqrt{n} H_{n,p}(\bm \alpha)) - \lambda_{\text{max}}(\sqrt{n} H_{n,p}(\bm \alpha')) \right| \leq L\|\alpha- \alpha'\|$, which is the same Lipschitz bound. Taking $L^2=3^pn$ in \Cref{lem:gaussian_concentration} completes the proof.
\end{proof}

The concentration bound above can be strengthened by obtaining tighter control of the variance (see \Cref{sec:properties_entangled_states}). In fact, we can remove the dependence on $p$ when considering the maximum energy over product states. This arises due to the fact that the variance of the energy of any product state is equal to one (see \Cref{lem:covariance_product}) and independent of $p$.

\begin{proposition}[Concentration of maximum energy over product states] \label{prop:product_state_concentration}
    Let $E_{n,\rm Prod}^*(p)$ be a random variable quantifying the maximum normalized energy over product states defined as
    \begin{equation}
        E_{n,\rm Prod}^*(p) \coloneqq \frac{1}{\sqrt{n}} \max_{\ket{\phi} \in \mathcal{S}_{\rm product}^n} \bra{\phi}H_{n,p} \ket{\phi}.
    \end{equation}
    Then, for all $t\geq 0$:
    \begin{equation}
        \mathbb{P}\left[ \left| E_{n,\rm Prod}^*(p) - \E \left[ E_{n,\rm Prod}^*(p) \right] \right| \geq t \right] \leq 2 \exp\left( -\frac{t^2n}{2 } \right).
    \end{equation}
\end{proposition}
\begin{proof}
    We start with \Cref{eq:lipschitz_from_hamiltonian} and normalize it by a factor $1/n$ obtaining
    \begin{equation} 
        \left\|\nabla_{\bm \alpha} \bra{\phi} \frac{1}{\sqrt{n}} H_{n,p}(\bm \alpha) \ket{\phi}  \right\|_2^2 = \frac{1}{n\binom{n}{p}} \sum_{\indi \in \mathcal{I}_p^n} \sum_{\inda \in \{1,2,3\}^p} \left(\bra{\phi} P_{\indi}^{\inda} \ket{\phi}\right)^2.
    \end{equation}
    The above is equivalent up to a factor $1/n$ to the equation in the calculation of the variance in \Cref{lem:covariance_product} (see specifically \Cref{eq:covariance_first_step}). Following the proof of \Cref{lem:covariance_product} from \Cref{eq:covariance_first_step} to \Cref{eq:variance_lemma_final}, we can simplify the above as
    \begin{equation} 
        \left\|\nabla_{\bm \alpha} \bra{\phi} \frac{1}{\sqrt{n}} H_{n,p}(\bm \alpha) \ket{\phi}  \right\|_2^2 = \frac{1}{n\binom{n}{p}} \sum_{\indi \in \mathcal{I}_p^n} \prod_{k=1}^p \left( 2 \left|\bra{\phi^{(\indisub_k)}}\ket{ \phi^{(\indisub_k)}} \right|^2 - 1\right) = \frac{1}{n}.
    \end{equation}
    Thus, the energy of a product state has the Lipschitz bound $L\leq \frac{1}{\sqrt{n}}$. By the same reasoning in the proof of \Cref{prop:spin_glass_concentration}, this Lipschitz bound extends to the maximum over all product states. Applying $L= \frac{1}{\sqrt{n}}$ into \Cref{lem:gaussian_concentration} completes the proof.
\end{proof}

\section{Expected maximum energy for product states} 
\label{sec:quantum_p_spin_proof}

We now compute lower and upper bounds of the expected maximum energy achieved by product state in the limit of large $p$. By \Cref{prop:product_state_concentration} this is equivalent to bounds on the maximum product state energy that hold with high probability.

\subsection{Lower bound}

We begin with the lower bound. We first construct a set $\mathcal{S}_\delta^n \subset \mathcal{S}^n_{\rm product}$ of product states where pairs of states are sufficiently far apart and then applying the second moment method to the number $N_\epsilon$ of states $\ket{\bm{\mu}} \in \mathcal{S}_\delta^n$ above normalized energy threshold $C_\epsilon\sqrt{\log p}$ for some $\epsilon \geq 0$, where
\begin{equation}
    C_\epsilon\coloneqq\left(1-\epsilon\right)\sqrt{2}.
\end{equation}
Specifically,
\begin{equation}
    N_\epsilon = \left| \left\{ \ket{\bm{\mu}} \in \mathcal{S}_\delta^n: \sqrt{n}\bra{\bm{\mu}}H_{n,p}\ket{\bm{\mu}} \geq C_\epsilon\sqrt{\log\left(p\right)} n \right\} \right|.
\end{equation}
Throughout, we normalize the quantum Hamiltonian as $\sqrt{n} H_{n,p}$ so that the variance (over the distribution of coefficients $\bm{\alpha}$) of the energy expectation for a product state is equal to $n$. The second moment method lower bounds the probability that $N_\epsilon > 0$ via the Paley--Zygmund inequality:
\begin{equation}
    \mathbb{P}\left[ N_\epsilon > 0 \right] \geq \frac{\E[N_\epsilon]^2}{\E{[N_\epsilon^2]}}.
\end{equation}

\paragraph{Construction of $\mathcal{S}_\delta^n$.}
Consider a set $S_\delta=\left\{\bm{\hat{n}_\alpha}\right\}_{\alpha=1}^{\frac{q}{2}}\cup\left\{-\bm{\hat{n}_\alpha}\right\}_{\alpha=1}^{\frac{q}{2}}$ (for some constant $\delta$) of unit vectors $\bm{\hat{n}_\alpha} \in \mathbb{R}^3$ such that for $\alpha\neq\beta$:
\begin{equation}
    \left\lvert \langle \bm{\hat{n}_\alpha}, \bm{\hat{n}_\beta} \rangle \right\rvert\leq 1-p^{-\left(1-\delta\right)}.
\end{equation}
By \Cref{lem:packing_of_sphere} (derived from \cite{6773576}), for any $\delta$ small enough, $\frac{q}{2}=\frac{3}{4}p^{1-\delta}$ points can be chosen satisfying this constraint such that all $\bm{\hat{n}_i}$ lie in the spherical cap up to a polar angle of $\frac{\pi}{2}-0.01$. We thus take $q=\frac{3}{2}p^{1-\delta}$ in the following.

Given $S_\delta$, let us label pure product states via their Bloch sphere representation~\cite{nielsen_chuang_2010}:
\begin{equation}
    \ket{\bm{\mu}}\bra{\bm{\mu}}=\bigotimes\limits_{i=1}^n\left(\frac{\sigma^0}{2}+\frac{\bm{\hat{n}_{\mu_i}}\cdot\bm{\sigma}}{2}\right),
\end{equation}
for $\mu_k\in\left[q\right]$ and where $\bm{\sigma} = [\sigma^1, \sigma^2, \sigma^3]^\intercal$ is a vector consisting of pauli $X,Y,Z$ matrices. Since $\|\bm{\hat{n}_{\mu_i}}\|=1$, the above is guaranteed to be a pure state. Then,
\begin{equation} \label{eq:construction_of_logp_set}
    \mathcal{S}_\delta^n = \left\{  \ket{\bm{\mu}} \in   \mathcal{S}^n_{\rm product}:  \forall i \in [n], \mu_i \in S_\delta \right\}.
\end{equation}

By \Cref{lem:covariance_product}
this construction features the covariance:
\begin{equation} \label{eq:covariance_product_basis_states}
    \begin{aligned}
        \E\left[ \bra{\bm{\mu}} H_{n,p} \ket{\bm{\mu}} \bra{\bm{\tilde{\mu}}} H_{n,p} \ket{\bm{\tilde{\mu}}}\right]&= \frac{1}{\binom{n}{p}} \sum_{\indi \in \mathcal{I}_p^n}\prod\limits_{k=1}^p\left(2\tr\left(\left(\frac{\sigma^0}{2}+\frac{\bm{\hat{n}_{\mu_{i_k}}}\cdot\bm{\sigma}}{2}\right)\left(\frac{\sigma^0}{2}+\frac{\bm{\hat{n}_{\tilde{\mu}_{i_k}}}\cdot\bm{\sigma}}{2}\right)\right)-1\right)\\
        &= \frac{1}{\binom{n}{p}} \sum_{\indi \in \mathcal{I}_p^n}\prod\limits_{k=1}^p\bm{\hat{n}_{\mu_{i_k}}}\cdot\bm{\hat{n}_{\tilde{\mu}_{i_k}}}.
    \end{aligned}
\end{equation}
Note that $-1 \leq \bm{\hat{n}_{\mu_{i_k}}}\cdot\bm{\hat{n}_{\tilde{\mu}_{i_k}}} \leq 1$, and thus:
\begin{equation}
    \left| \left( \frac{1}{n} \sum_{i=1}^n\bm{\hat{n}_{\mu_i}}\cdot\bm{\hat{n}_{\tilde{\mu}_i}}\right)^p - \frac{p! {n \choose p}}{n^p} \frac{1}{{n \choose p}} \sum_{\substack{i_1 < \cdots < i_p}}\prod\limits_{k=1}^p\bm{\hat{n}_{\mu_{i_k}}}\cdot\bm{\hat{n}_{\tilde{\mu}_{i_k}}} \right| \leq \frac{n^p - p! {n \choose p}}{n^p} = O\left(\frac{1}{n}\right),
\end{equation}
so in particular:
\begin{equation} \label{eq:covar_product_mu}
    \E\left[ \bra{\bm{\mu}} H_{n,p} \ket{\bm{\mu}} \bra{\bm{\tilde{\mu}}} H_{n,p} \ket{\bm{\tilde{\mu}}}\right]=  \left( \frac{1}{n} \sum_{i=1}^n\bm{\hat{n}_{\mu_k}}\cdot\bm{\hat{n}_{\tilde{\mu}_k}}\right)^p+O\left(\frac{1}{n}\right).
\end{equation}

\paragraph{First and second moments.}
We begin with the first moment.
By \Cref{eq:covariance_product_basis_states}, it can be verified that $\sqrt{n}\bra{\bm{\mu}}H_{n,p}\ket{\bm{\mu}}$ is Gaussian with mean zero and variance $n$ for any $\ket{\bm{\mu}}$. Denoting the cumulative distribution function (CDF) of the standard normal distribution as $\Phi$, we have that the first moment is equal to:
\begin{equation}\label{eq:1st-moment}
\begin{split}
    \E N_\epsilon &= \sum_{\bm{\mu}\in\left[q\right]^n} \mathbb{P}[\sqrt{n} \bra{\bm{\mu}} H_{n,p}\ket{\bm{\mu}} \geq C_\epsilon \sqrt{\log (p)}n] \\
    &= q^n \left(1-\Phi(C_\epsilon\sqrt{\log\left(p\right)n})\right).
\end{split}
\end{equation}

For the second moment, we will use the shorthand:
\begin{equation}
    R\left(\bm{\mu},\bm{\tilde{\mu}}\right)\coloneqq\E\left[ \bra{\bm{\mu}} H_{n,p} \ket{\bm{\mu}} \bra{\bm{\tilde{\mu}}} H_{n,p} \ket{\bm{\tilde{\mu}}}\right]
\end{equation}
for this covariance. The joint probability is
\begin{equation}
\begin{split}
    \mathbb{P}&\left[\sqrt{n}\langle \bm{\mu}| H_{n,p} |\bm{\mu} \rangle\geq C_\epsilon \sqrt{\log\left(p\right)}n \cap \sqrt{n}\langle \bm{\tilde{\mu}}| H_{n,p} |\bm{\tilde{\mu}}\rangle \geq C_\epsilon \sqrt{\log\left(p\right)}n \right]  \\
    & \leq \mathbb{P}\left[\sqrt{n}\langle \bm{\mu}| H_{n,p} |\bm{\mu} \rangle \geq C_\epsilon \sqrt{\log\left(p\right)}n \right] \mathbb{P}\left[ \sqrt{n}\langle \bm{\tilde{\mu}}| H_{n,p} |\bm{\tilde{\mu}}\rangle \geq C_\epsilon \sqrt{\log\left(p\right)}n \;|\; \sqrt{n}\langle \bm{\mu}| H_{n,p} |\bm{\mu} \rangle =  C_\epsilon \sqrt{\log\left(p\right)}n \right] \\
    & = \left[1-\Phi(C_\epsilon \sqrt{\log\left(p\right)n} ) \right] \left[1- \Phi\left( C_\epsilon \left( 1- R(\bm{\mu},\bm{\tilde{\mu}})^2\right)^{-1/2} \left( 1- R(\bm{\mu},\bm{\tilde{\mu}})\right) \sqrt{\log\left(p\right)n}  \right)\right].
 \end{split}
\end{equation}
In the second line above, we use the monotonicity of the conditional CDF of the Gaussian for the inequality since the conditional CDF is minimized when conditioned on the smallest value possible. Applying the above formula, we have
\begin{equation}
    \E[ N_\epsilon^2] \leq \sum_{\bm{\mu},\bm{\tilde{\mu}}}\left[1-\Phi(C_\epsilon \sqrt{\log\left(p\right)n} ) \right] \left[1- \Phi\left( C_\epsilon \left( 1- R(\bm{\mu},\bm{\tilde{\mu}})^{2}\right)^{-1/2} \left( 1- R(\bm{\mu},\bm{\tilde{\mu}})\right) \sqrt{\log\left(p\right)n}  \right)\right].
\end{equation}
Eventually, we will show that the ratio $\frac{\E[N_\epsilon]^2}{\E{[N_\epsilon^2]}}$ is sufficiently bounded to obtain a contradiction with the concentration of $N_\epsilon$. The Lemma below establishes this bound.

\begin{lemma} \label{lem:pre_contradiction_quantum_lower_bound}
    For any small fixed $\epsilon>0$ and all $\delta>0$ such that $\delta<2\epsilon - \epsilon^2$ in the construction of $\mathcal{S}_\delta^n$ (see \Cref{eq:construction_of_logp_set}), the ratio $\frac{\E[N_\epsilon]^2}{\E{[N_\epsilon^2]}}$ obeys the asymptotic relation
    \begin{equation}
        \frac{\E[N_\epsilon]^2}{\E{[N_\epsilon^2]}} \geq \exp\left(-O_p( 1 )n - o_n(1)\right).
    \end{equation}
\end{lemma}
\begin{proof}
    We enumerate over pairs as follows. Fix $\bm{\mu},\bm{\tilde{\mu}}$ and let $f_1$ be the number of components $\mu_k$ such that:
\begin{equation}
    f_1 = \left|\{ k \in [n]: \left\lvert\bm{\hat{n}_{\mu_k}}\cdot\bm{\hat{n}_{\tilde{\mu}_k}}\right\rvert=1 \} \right|,
\end{equation}
with associated index set $K_1 = \{ k \in [n]: \left\lvert\bm{\hat{n}_{\mu_k}}\cdot\bm{\hat{n}_{\tilde{\mu}_k}}\right\rvert=1 \}$. Define:
\begin{equation}
    r_1\coloneqq\sum\limits_{k\in K_1}\bm{\hat{n}_{\mu_k}}\cdot\bm{\hat{n}_{\tilde{\mu}_k}},
\end{equation}
taking values in the set $\{-f_1, -f_1+2, \dots, f_1\}$. Note the combinatorial factor of the number of $\bm{\tilde{\mu}}$ which, given $\bm{\mu}$, have a given value of $f_1$ and $r_1$:
\begin{equation}
    \left(q-2\right)^{n-f_1}\binom{n}{f_1}\binom{f_1}{\frac{f_1+r_1}{2}}
\end{equation}
At any given $f_1,r_1$, we have the upper bound:
\begin{equation}
    R\left(\bm{\mu},\bm{\tilde{\mu}}\right)\leq\left(\frac{\left\lvert r_1\right\rvert}{n}+\left(1-\frac{f_1}{n}\right)\left(1-p^{-\left(1-\delta\right)}\right)\right)^p+O\left(\frac{1}{n}\right)\eqqcolon\tilde{R}\left(f_1,r_1\right)+O\left(\frac{1}{n}\right).
\end{equation}
In what follows, we will enumerate over ordered pairs by splitting pairs into three sets.
\begin{itemize}
    \item $S_{R_1}$: subset of $(\bm \mu,\bm{\tilde{\mu}})$ such that $R\left(\bm{\mu},\bm{\tilde{\mu}}\right) = 1$,
    \item $F_{p,\delta}$: subset of $(\bm \mu,\bm{\tilde{\mu}})$ with $f_1\leq n\left(1-p^{-\frac{\delta}{2}}\right)$,
    \item $(S_{R_1} \cup F_{p,\delta})^C$: subset of $(\bm \mu,\bm{\tilde{\mu}})$ such that $R\left(\bm{\mu},\bm{\tilde{\mu}}\right) \neq 1$ and $f_1> n\left(1-p^{-\frac{\delta}{2}}\right)$.
\end{itemize}
From here, we can perform the second moment calculation as:
\begin{equation}
    \begin{split}
        \frac{\mathbb{E}\left[N_\epsilon^2\right]}{\mathbb{E}\left[N_\epsilon\right]^2} &\leq \sum_{(\bm{\mu},\bm{\tilde{\mu}}) \in S_{R_1}} P(\bm{\mu}, \bm{\tilde{\mu}}) 
        + \sum_{(\bm{\mu},\bm{\tilde{\mu}}) \in F_{p,\delta}} P(\bm{\mu}, \bm{\tilde{\mu}}) 
        + \sum_{(\bm{\mu},\bm{\tilde{\mu}}) \in (S_{R_1} \cup F_{p,\delta})^C} P(\bm{\mu}, \bm{\tilde{\mu}}),
    \end{split}
\end{equation}
where
\begin{equation}
    P(\bm{\mu}, \bm{\tilde{\mu}}) 
    =  
    q^{-2n} \left[1-\Phi(C_\epsilon \sqrt{\log\left(p\right)n} ) \right]^{-1} \left[1- \Phi\left( 
    \frac{1- R(\bm{\mu},\bm{\tilde{\mu}})}{\sqrt{ 1- R(\bm{\mu},\bm{\tilde{\mu}})^{2}}}
     C_\epsilon \sqrt{\log\left(p\right)n}  \right)\right].
\end{equation}

Treating the three subsets separately, we start with $(\bm \mu,\bm{\tilde{\mu}}) \in S_{R_1}$, where $R\left(\bm{\mu},\bm{\tilde{\mu}}\right) = 1$ (completely correlated Gaussians); we have:
\begin{equation}
    \begin{aligned}
        \left[1-\Phi(C_\epsilon \sqrt{\log\left(p\right)n} ) \right]^{-1} &\left[1- \Phi\left( C_\epsilon \left( 1- R(\bm{\mu},\bm{\tilde{\mu}})^{2}\right)^{-1/2} \left( 1- R(\bm{\mu},\bm{\tilde{\mu}})\right) \sqrt{\log\left(p\right)n}  \right)\right]\\
        &=\left(1+o_n\left(1\right)\right)\exp\left(\frac{1}{2}C_\epsilon^2\log\left(p\right)n\right).
    \end{aligned}
\end{equation}
Here we have used the tail bound~\cite{savage1962mills}:
\begin{equation}
    1- \Phi(x) = \frac{1}{x + O\left(\frac{1}{x}\right)} \frac{1}{\sqrt{2\pi}}  \exp\left(-\frac{x^2}{2} \right). 
    \label{eq:savage}
\end{equation}
Thus, given $q=\frac{3}{2}p^{1-\delta}$ and $(1-\delta)>1/2$:
\begin{equation}
\begin{split}
        \sum_{(\bm{\mu},\bm{\tilde{\mu}}) \in S_{R_1}} P(\bm{\mu}, \bm{\tilde{\mu}}) 
        &= \left(1+o_n\left(1\right)\right)q^{-2n}\exp\left(\frac{1}{2}C_\epsilon^2\log\left(p\right)n\right) \\
        &= \left(1+o_n\left(1\right)\right)\exp\left[ ((1-\epsilon)^2-2(1-\delta)) \log(p) n - 2\log(3/2)n \right] \\
        &= o_n(1).
\end{split}
\end{equation}

For pairs in $F_{p,\delta}$ where $f_1\leq n\left(1-p^{-\frac{\delta}{2}}\right)$ and thus $r_1\leq n\left(1-p^{-\frac{\delta}{2}}\right)$, we have that the asymptotic covariance is bounded by $\xi_{p,\delta}$, which is equal to
\begin{equation}\label{eq:r_tilde_upper_bound}
\begin{split}    
    \tilde{R}\left(f_1,r_1\right)\leq \xi_{p,\delta} 
    &\coloneqq
    \left(1-p^{-\frac{\delta}{2}}+p^{-\frac{\delta}{2}}\left(1-p^{-\left(1-\delta\right)}\right)\right)^p \\
    &=\exp\left[-p^{\delta/2} +o_p(1)\right].
\end{split}
\end{equation}
Once again from the tail bound of \Cref{eq:savage} we have:
\begin{equation}
    \begin{split}
        \sum_{(\bm{\mu},\bm{\tilde{\mu}}) \in F_{p,\delta}}& P(\bm{\mu}, \bm{\tilde{\mu}}) \leq q^{-2n}\sum_{(\bm{\mu},\bm{\tilde{\mu}})\in F_{p,\delta}}\exp\left(\frac{C_\epsilon^2\log\left(p\right)}{1+R\left(\bm{\mu},\bm{\tilde{\mu}}\right)^{-1}}n\right)O_n\left(1\right)\sqrt{\frac{1+R\left(\bm{\mu},\bm{\tilde{\mu}}\right)}{1-R\left(\bm{\mu},\bm{\tilde{\mu}}\right)}}.
    \end{split}
\end{equation}
Splitting the above sum into the different possible values of $f_1, r_1$ with multiplicity $\left(q-2\right)^{n-f_1}\binom{n}{f_1}\binom{f_1}{\frac{f_1+r_1}{2}}$ for given $f_1, r_1$: 
\begin{equation}
    \begin{split}
        \sum_{(\bm{\mu},\bm{\tilde{\mu}}) \in F_{p,\delta}}& P(\bm{\mu}, \bm{\tilde{\mu}}) \leq O_n\left(1\right)q^{-n}\sum\limits_{f_1\leq n\left(1-p^{-\frac{\delta}{2}}\right),r_1}\left(q-2\right)^{n-f_1}\binom{n}{f_1}\binom{f_1}{\frac{f_1+r_1}{2}}\times \\
        &\quad \quad \quad \quad \quad \quad \quad \quad \quad \quad \quad \quad \quad \quad \quad \times \exp\left(\frac{C_\epsilon^2\log\left(p\right)}{1+\tilde{R}\left(f_1,r_1\right)^{-1}}n\right)\sqrt{\frac{1+\tilde{R}\left(f_1,r_1\right)}{1-\tilde{R}\left(f_1,r_1\right)}}\\
        &\leq O_n\left(1\right)\sqrt{\frac{1+\xi_{p,\delta}}{1-\xi_{p,\delta}}}\sum\limits_{f_1\leq n\left(1-p^{-\frac{\delta}{2}}\right),r_1}\exp_2\left(\left(H_2\left(\frac{f_1}{n}\right)+\frac{f_1}{n}H_2\left(\frac{1}{2}+\frac{r_1}{2f_1}\right)\right)n + o_n(1)\right) \times \\
        &\quad \quad \quad \quad \quad \quad \quad \quad \quad \quad \quad \quad \quad \quad \quad \times \exp\left(\left(-\frac{f_1\log\left(q\right)}{n\log\left(p\right)}+C_\epsilon^2\xi_{p,\delta}\right)\log\left(p\right)n \right).
    \end{split}
\end{equation}
In the above, we use the fact that $\tilde{R}\left(f_1,r_1\right)\leq \xi_{p,\delta}$ as defined earlier and note that $\frac{1}{1+\tilde{R}\left(f_1,r_1\right)^{-1}} \leq \frac{\xi_{p,\delta}}{\xi_{p,\delta}+1} \leq \xi_{p,\delta}$. We also use the bound on the binomial coefficient ${n \choose k} \leq \exp_2(H_2(k/n)n)$ where $H_2(\cdot)$ is the binary entropy function. Keeping track of the terms in the exponent of order $\log(p)$:
\begin{equation}
    \begin{split}
        \sum_{(\bm{\mu},\bm{\tilde{\mu}}) \in F_{p,\delta}}& P(\bm{\mu}, \bm{\tilde{\mu}}) \leq O_n\left(1\right)\sqrt{\frac{1+\xi_{p,\delta}}{1-\xi_{p,\delta}}}\sum\limits_{f_1\leq n\left(1-p^{-\frac{\delta}{2}}\right),r_1}\exp\left(C_\epsilon^2\xi_{p,\delta}\log\left(p\right)n + O_p(1)n + o_n(1)\right)\\
        &\leq O_n\left(n^2\right)\sqrt{\frac{1+\xi_{p,\delta}}{1-\xi_{p,\delta}}}\exp\left(C_\epsilon^2\xi_{p,\delta}\log\left(p\right)n + O_p(1)n + o_n(1)\right).
    \end{split}
\end{equation}
Noting that $\xi_{p,\delta}=\exp\left[-p^{\delta/2} +o_p(1)\right]$, we have:
\begin{equation}
    \begin{split}
        \sum_{(\bm{\mu},\bm{\tilde{\mu}}) \in F_{p,\delta}}& P(\bm{\mu}, \bm{\tilde{\mu}}) \leq  O_n\left(n^2\right)\exp\left[O_p(1)n + o_n(1)\right].
    \end{split}
\end{equation}

Finally, for the remaining subset $(S_{R_1} \cup F_{p,\delta})^C$, we use the same approach as before and expand the set of pairs into the possible values of $f_1, r_1$:
\begin{equation}
    \begin{aligned}
        \sum_{(\bm{\mu},\bm{\tilde{\mu}}) \in (S_{R_1} \cup F_{p,\delta})^C}& P(\bm{\mu}, \bm{\tilde{\mu}}) \\
        &\leq q^{-2n}\sum_{\bm{\mu};\bm{\tilde{\mu}} \in (S_{R_1} \cup F_{p,\delta})^C}\left[1-\Phi(C_\epsilon \sqrt{\log\left(p\right)n} ) \right]^{-1} \left[1- \Phi\left( 
        \frac{1- R(\bm{\mu},\bm{\tilde{\mu}})}{\sqrt{ 1- R(\bm{\mu},\bm{\tilde{\mu}})^{2}}}
        C_\epsilon \sqrt{\log\left(p\right)n}  \right)\right] \\
        &\leq q^{-2n}\sum_{\bm{\mu};\bm{\tilde{\mu}} \in (S_{R_1} \cup F_{p,\delta})^C}\left[1-\Phi(C_\epsilon \sqrt{\log\left(p\right)n} ) \right]^{-1}
        \\
        &\leq O_n(n)q^{-n}\sum\limits_{\substack{f_1>n\left(1-p^{-\frac{\delta}{2}}\right) \\ r_1\leq\min\left(f_1,n-1\right)}}\left(q-2\right)^{n-f_1}\binom{n}{f_1}\binom{f_1}{\frac{f_1+r_1}{2}} \exp\left(\frac{C_\epsilon^2\log\left(p\right)}{2}n\right).
    \end{aligned}
\end{equation}
Again, applying bound on the binomial coefficient ${n \choose k} \leq \exp_2(H_2(k/n)n)$, we have:
\begin{equation}
    \begin{aligned}
        \sum_{(\bm{\mu},\bm{\tilde{\mu}}) \in (S_{R_1} \cup F_{p,\delta})^C}& P(\bm{\mu}, \bm{\tilde{\mu}}) \\
        &\leq O_n(n)q^{-n}\sum\limits_{\substack{f_1>n\left(1-p^{-\frac{\delta}{2}}\right) \\ r_1\leq\min\left(f_1,n-1\right)}}\left(q-2\right)^{n-f_1}\exp_2\left(nH_2\left(\frac{f_1}{n}\right)+f_1H_2\left(\frac{1}{2}+\frac{r_1}{2f_1}\right)\right)\times \\
        & \quad \quad \quad \quad \quad \quad \quad \quad \quad \quad \quad
        \times \exp\left(\frac{C_\epsilon^2\log\left(p\right)}{2}n\right)\\
        &<O_n(n) \sum\limits_{\substack{f_1>n\left(1-p^{-\frac{\delta}{2}}\right) \\ r_1\leq\min\left(f_1,n-1\right)}}\exp\left(\left(-\frac{\log\left(q\right)}{\log\left(p\right)}+\frac{C_\epsilon^2}{2}+O_p\left(\frac{1}{\log\left(p\right)}\right)\right)\log\left(p\right)n+o_n\left(n\right)\right)\\
        &= O_n(n)\sum\limits_{\substack{f_1>n\left(1-p^{-\frac{\delta}{2}}\right) \\ r_1\leq\min\left(f_1,n-1\right)}}\exp\left(\left(-\frac{\log\left(q\right)}{\log\left(p\right)}+\left(1-\epsilon\right)^2+O_p\left(\frac{1}{\log\left(p\right)}\right)\right)\log\left(p\right)n+o_n\left(n\right)\right).
    \end{aligned}
\end{equation}
Noting that there are are most $n^2$ terms in the sum and that $q=\frac{3}{2}p^{1-\delta}$, we have for all $\delta<2\epsilon-\epsilon^2$ that
\begin{equation}
    -\frac{\log\left(q\right)}{\log\left(p\right)}+\left(1-\epsilon\right)^2 = -\log(3/2)\log(p)^{-1} - (1-\delta) + (1-\epsilon)^2 =   -\Omega_p(1),
\end{equation}
and then
\begin{equation}
    \sum_{(\bm{\mu},\bm{\tilde{\mu}}) \in (S_{R_1} \cup F_{p,\delta})^C} P(\bm{\mu}, \bm{\tilde{\mu}}) = O_n(n^3)\exp(-\Omega_p(1)n) = o_n(1).
\end{equation}

Combining the results from the three subsets, we have that the dominant contribution is from $F_{p,\delta}$, so
\begin{equation}
    \frac{\mathbb{E}\left[N_\epsilon\right]^2}{\mathbb{E}\left[N_\epsilon^2\right]} = \exp\left[-O_p(1)n - o_n(1)\right].
\end{equation}
\end{proof}

\paragraph{Completion of the lower bound proof.}
We now proceed to prove the lower bound in the main theorem (\Cref{theorem:Main-result-informal}) via a contradiction argument with the concentration properties of $E_{n,\rm prod}^*(p)$.

\begin{theorem}\label{thm:quantum_second_moment_lower_bound}
    For every $\epsilon>0$, there exists $p_0$ such that for all $p\geq p_0$, the limiting maximum energy over product states
    \begin{equation*}
        E_{n,\rm prod}^*(p)  \coloneqq   \frac{1}{\sqrt{n}} \max_{\ket{\phi} \in \mathcal{S}_{\rm product}^n} \bra{\phi}H_{n,p} \ket{\phi}
    \end{equation*}
    has the property that with high probability
    \begin{equation}
        E_{n,\rm prod}^*(p) \geq (1-\epsilon)\sqrt{2(\log p)}.
    \end{equation}
\end{theorem}
\begin{proof}
As a reminder, we use the notation $C_\epsilon\coloneqq\left(1-\epsilon\right)\sqrt{2}$. We claim that
\begin{equation}\label{eq:exp_lower_bound_quantum}
    \mathbb{E}\left[\max_{\ket{\bm{\mu}}}\sqrt{n}\bra{\bm{\mu}}H_{n,p}\ket{\bm{\mu}}\right]\geq C_{2\epsilon}\sqrt{\log\left(p\right)} n
\end{equation}
for every $\epsilon>0$ and sufficiently large $p$. Indeed, assume that this is not the case. Then by the concentration bound of \Cref{prop:product_state_concentration}:
\begin{equation}
    \begin{aligned}
        \mathbb{P}&\left(\max_{\ket{\bm{\mu}}}\sqrt{n}\bra{\bm{\mu}}H_{n,p}\ket{\bm{\mu}} \geq C_\epsilon\sqrt{\log\left(p\right)} n\right)\\
        &\leq\mathbb{P}\left(\max_{\ket{\bm{\mu}}}\sqrt{n}\bra{\bm{\mu}}H_{n,p}\ket{\bm{\mu}}-\mathbb{E}\left[\max_{\ket{\bm{\mu}}}\sqrt{n}\bra{\bm{\mu}}H_{n,p}\ket{\bm{\mu}}\right]\geq C_\epsilon\sqrt{\log\left(p\right)} n-C_{2\epsilon}\sqrt{\log\left(p\right)} n\right)\\
        &\leq 2\exp\left(-\frac{\left(C_\epsilon-C_{2\epsilon}\right)^2}{2}\log\left(p\right)n\right).
    \end{aligned}
\end{equation}
However, \Cref{lem:pre_contradiction_quantum_lower_bound} states that for sufficiently large $p$, we have by Paley-Zygmund
\begin{equation}
    \mathbb{P}\left(\max_{\ket{\bm{\mu}}}\sqrt{n}\bra{\bm{\mu}}H_{n,p}\ket{\bm{\mu}}  \geq C_\epsilon\sqrt{\log\left(p\right)} n\right) \geq\frac{\mathbb{E}\left[N_\epsilon\right]^2}{\mathbb{E}\left[N_\epsilon^2\right]} = \exp\left(-O_p( 1 )n - o_n(1)\right) 
\end{equation}
which results in the contradiction since for sufficiently large $p$ and $n$
\begin{equation}
    2\exp\left(-\frac{\left(C_\epsilon-C_{2\epsilon}\right)^2}{2}\log\left(p\right)n\right)<\exp\left(-O_p( 1 )n - o_n(1)\right).
\end{equation}
Thus, \eqref{eq:exp_lower_bound_quantum} holds. By applying the concentration bound again we thus have that with probability at least $1-\exp\left(-\Omega\left(n\right)\right)$
\begin{equation}
    \max_{\ket{\bm{\mu}}}\sqrt{n}\bra{\bm{\mu}}H_{n,p}\ket{\bm{\mu}}\geq\left(1-3\epsilon\right)\sqrt{2\log\left(p\right)}n.
\end{equation}
\end{proof}

\subsection{Matching upper bound}

To prove a matching upper bound, we follow a two step procedure. First, we apply Markov's inequality over the discrete net of product states studied previously of size $\exp(n \left( \log(p) + o(\log(p))\right))$. We will show that with high probability, any product state over this discrete net will have energy less than $\left(1+o_p\left(1\right)\right)\sqrt{2  \log (p) n}$. Then, we will extend this proof over the discrete net to cover all product states by applying a variant of Dudley's inequality~\cite{vershynin2010introduction}.

As before, we label pure product states via their Bloch sphere representation~\cite{nielsen_chuang_2010}:
\begin{equation}
    \ket{\bm{\mu}}\bra{\bm{\mu}}=\bigotimes\limits_{i=1}^n\left(\frac{\sigma^0}{2}+\frac{\bm{n_{\mu_i}}\cdot\bm{\sigma}}{2}\right),
\end{equation}
where $\bm{n_{\mu_i}} \in \mathbb{R}^3$ and $\|\bm{n_{\mu_i}}\| = 1$ to ensure the above is pure.
To simplify notation, 
for a product state $\ket{\bm{\mu}}$, let us denote
\begin{equation}
    X_{\bm{\mu}} = \bra{\bm{\mu}} H_{n,p} \ket{\bm{\mu}}
\end{equation}
as the (Gaussian) random variable corresponding to the energy of $\ket{\bm{\mu}}$. 
The random variables $(X_{\bm{\mu}})_{\ket{\bm{\mu}} \in \mathcal{S}_{\rm product}^n} $ 
form a centered Gaussian process indexed by the space 
of product states $\mathcal{S}_{\rm product}^n$. 
The canonical norm on $\mathcal{S}_{\rm product}^n$
induced by the Gaussian process is
\begin{equation}
    d(\bm{\mu}, \bm{\tilde{\mu}}) \coloneqq \sqrt{\mathbb{E}\left[ (X_{\bm{\mu}} - X_{\bm{\tilde{\mu}}})^2 \right]}.
\end{equation}

In fact, by applying \Cref{eq:covar_product_mu},
\begin{equation}
    d(\bm{\mu}, \bm{\tilde{\mu}}) =  \sqrt{2-2\left( \frac{1}{n} \sum_{i=1}^n\bm{n_{\mu_k}}\cdot\bm{n_{\tilde{\mu}_k}}\right)^p}+o_n\left(1\right).
\end{equation}

Dudley's inequality provides a bound on the supremum of a Gaussian process.
\begin{theorem}[Dudley's Inequality \cite{vershynin2010introduction} (see also Theorem 3.1.2 of \cite{ko2020free})] \label{thm:dudley_sup}
    For the centered Gaussian process $(X_{\bm{\mu}})_{\ket{\bm{\mu}} \in \mathcal{S}_{\rm product}^n}$ indexed by metric space $(\mathcal{S}_{\rm product}^n,d)$, there exists an absolute constant $L$ such that 
    \begin{equation}
        \E \sup_{d(\bm{\mu}, \bm{\tilde{\mu}}) \leq \Delta } \left| X_{\bm{\mu}} - X_{\bm{\tilde{\mu}}} \right| \leq L \int_0^{\Delta} \sqrt{\log \mathcal{N}\left(\mathcal{S}_{\rm product}^n,d,\epsilon\right)} \; \dd \epsilon,
    \end{equation}
    where $\mathcal{N}\left(\mathcal{S}_{\rm product}^n,d,\epsilon\right)$ is the covering number of $\mathcal{S}_{\rm product}^n$ with $\epsilon$-balls in the metric $d$.
\end{theorem}

There exists a cover $\{\bm{x}_i\}_{i=1}^N$ on the set of points on the sphere ($\bm{x}_i \in \mathbb{R}^3$ such that $\|\bm{x}_i\| = 1$)
of cardinality $O(1/\epsilon')$ such that given any point on the sphere $\bm{y}$, 
there exists a point in the cover such that $\bm{x_i} \cdot \bm{y} \geq 1-\epsilon'$ (see Lemma~\ref{lem:upper_bound_cov}). Using such a cover on all the qubits, we have that for any state $\ket{\bm{\mu}}$, there exists a state $\ket{\bm{\tilde \mu}}\neq\ket{\bm{\mu}}$ in the cover such that
\begin{equation}
\begin{split}
    \left(\frac{1}{n} \sum_{i=1}^n\bm{n_{\mu_k}}\cdot\bm{n_{\tilde{\mu}_k}} \right)^p &\geq \left( 1 - \epsilon' \right)^p \\
    &\geq 1-p\epsilon'.
\end{split}
\end{equation}
Setting $\epsilon = p\epsilon'$, this implies that 
\begin{equation}
    \log \mathcal{N}\left(\mathcal{S}_{\rm product}^n,d,\epsilon\right) \leq n \left( \log(p) + \log(1/\epsilon) + O(1) \right).\label{eq:entropy-bound}
\end{equation}

Similar to the proof of the lower bound, consider a specific $\widetilde \delta$-cover $\widetilde S_{\widetilde \delta}^n$ where $\widetilde \delta = C (\log p)^{-1}$ for some constant $C>0$. By our upper bound on the covering number we have that 
\begin{equation}
    \log |\widetilde S_{\widetilde \delta}^n| \leq n \left( \log(p) + \log \log (p) + O(1) \right).
\end{equation}
Now consider the number $\widetilde N_{\theta}$ of states $\ket{\bm{\mu}} \in \widetilde S_{\widetilde \delta}^n$ above normalized energy threshold
\begin{equation}
    C_\theta = \sqrt{2(\log(p) + \theta \log \log (p) ) }
\end{equation}
for some $\theta \in \mathbb{R}$. Specifically,
\begin{equation} 
    \widetilde N_\theta = \left| \left\{ \ket{\bm{\mu}} \in \widetilde S_{\widetilde \delta}^n: \sqrt{n}\bra{\bm{\mu}}H_{n,p}\ket{\bm{\mu}} \geq C_\theta n \right\} \right|.
\end{equation}

From here, we have as in \Cref{eq:1st-moment} that the expected number of states above normalized energy $C_\theta$ is
\begin{equation}
    \E \left[ \widetilde N_\theta \right] = \exp\left[ ( (1-\theta) \log\log(p) + O(1) )n \right] ,
\end{equation}
which for large enough $p$ and $\theta>1$ is exponentially decaying in $n$ and by Markov's inequality:
\begin{equation} \label{eq:prob_pre_dudley}
    \mathbb{P}\left[ \widetilde N_\theta \geq 1 \right] \leq \exp(- \Omega(\log\log(p)) n),
\end{equation}
again for large enough $p$ and $\theta>1$.

We now extend this from the discrete net $\tilde{S}_{\tilde{\delta}}^n$ to
all states by bounding fluctuations of $X_{\bm{\mu}}$ away from $\ket{\bm{\mu}}$ in the net. To do this, we use \Cref{thm:dudley_sup}, 
setting $\Delta = \tilde{\delta} = C(\log p)^{-1}$ so that
\begin{equation} \label{eq:moment_post_dudley}
    \begin{split}
        \E \sup_{d(\bm{\mu}, \bm{\tilde{\mu}}) \leq C(\log p)^{-1} } \left| X_{\bm{\mu}} - X_{\bm{\tilde{\mu}}} \right| &\leq 
        \sqrt{n}\left( \int_0^{C(\log p)^{-1}} \sqrt{\log(p/\epsilon)} \; \dd{\epsilon} + O_p(1) \right) \\
        &=  O_p\left(1\right)\sqrt{n},
    \end{split}
\end{equation}
where we have used the integral:
\begin{equation}
    \begin{aligned}
        \int_0^{C(\log p)^{-1}} \sqrt{\log(p/\epsilon)} \; \dd{\epsilon}&=\frac{p\sqrt{\pi}}{2}\operatorname{erfc}\left(\sqrt{\log\left(\frac{p}{C}\log\left(p\right)\right)}\right)+\frac{C}{\log\left(p\right)}\sqrt{\log\left(\frac{p}{C}\log\left(p\right)\right)}\\
        &=O_p\left(1\right)\frac{\sqrt{\log\left(p\log\left(p\right)\right)}}{\log\left(p\right)}\\
        &=o_p\left(1\right).
    \end{aligned}
\end{equation}
Combining \Cref{eq:prob_pre_dudley} and \Cref{eq:moment_post_dudley} along with the concentration bound in \Cref{sec:concentration}, we obtain that
\begin{equation}    
\mathbb{E} \left[ \max_{\ket{\bm{\mu}} \in \mathcal{S}_{\rm product}^n} n^{-1/2} \bra{\bm{\mu}}H_{n,p}\ket{\bm{\mu}}  \right] \leq \sqrt{2(\log(p) + O_p(\log \log (p)) ) } + O_p(1) \leq (1+o_p(1)) \sqrt{2 \log(p)}.
\end{equation}

\section{Conjectured upper bound of expected maximum energy} 
\label{sec:quantum_p_spin_upper_bound}

A natural question is what approximation ratio product states achieve
with respect to the expected maximum energy over all states $E^*_n(p)=\E[\lambda_{\text{max}}(H_{n,p})/\sqrt{n}]$. As was the case for product states, by \Cref{prop:spin_glass_concentration} this is equivalent to a bound on the maximum energy that holds with high probability.
From applying the trace method, we can obtain an upper bound that is dependent on
conjectured asymptotic properties of certain hypergraph matchings formally stated in \Cref{conj:main-0}.
A constant $\gamma(p)$ in this conjecture directly results in corresponding bounds on the approximation factor achievable by product states. Our formal result (with proof deferred to \Cref{app:proof_of_upper_bound}) states that
\begin{align}\label{eq:formal_lambda_ub}
\E[\lambda_{\text{max}}(H_{n,p})/\sqrt{n}]
\le
(1
+o_p(1))
\sqrt{2\gamma(p)\log p},
\end{align}
where $\gamma(p)$ is the constant given in \Cref{conj:main-0}.
At this stage, we do not have tight control over the values $\gamma(p)$ 
or its asymptotics for large $p$ so we defer proofs and further formal discussion to \Cref{app:proof_of_upper_bound}.

\paragraph{Conjectured independence property}

We now formulate the asymptotic independence ansatz from which the $\gamma(p)$ of \Cref{eq:formal_lambda_ub} derives. Given a positive integer $d$, consider any matching $M$ of the $2d$ elements in the set $[2d]$ (namely pairing them into $d$ groups of two elements each). Let $\Trace_{\rm sum}(M)$ denote $\frac{1}{2}\sum \Trace(\sigma_1\ldots\sigma_{2d})$ where the sum runs over all choices of each $\sigma_j \in \{\sigma^1, \sigma^2, \sigma^3\}$ of single qubit Paulis, subject to the constraint that within each pair $(i,j) \in M$ in the matching $M$ the two Pauli matrices $\sigma_i = \sigma_j$ are identical. 

Next, suppose $r$ tuples $\indi_1,\ldots,\indi_r$
are chosen uniformly at random independently where each $\indi_j$ is a random sequence of distinct $p$-subsets of $[n]$ chosen also uniformly at random from the set of all $n(n-1)\cdots (n-p+1)$  such sets. 
Equivalently, we have a random hypergraph on the set of nodes $[n]$,
where $\indi_1,\ldots,\indi_r$ describes the set of
$p$-uniform hyper edges. 

Consider a permutation $\pi$ of $2r$ elements $\indi_1,\ldots,\indi_r, \indi_1,\ldots,\indi_r$ (namely each tuple appears exactly twice), 
chosen uniformly at random. For each $j\in [n]$ consider the set of tuples containing $j$. The permutation $\pi$ also induces a random uniform matching $M_j$ of these tuples (as each of them appears twice in the $2r$ sequence). In particular we have $\Trace_{\rm}(M_j)$ for each $j$. Our conjecture postulates an asymptotic independence for the sequence of this trace sums in expectation when $r=O(n)$, up to a sub-exponential factor. Specifically, we conjecture the following. 

\begin{conjecture}[Asymptotic matching ansatz]\label{conj:main-0}
For every constant $C$ and $r\leq Cn$, there exists a constant $\gamma(p)$ depending on $p$ but independent of $n$ and $r$ such that
\begin{align*}
\E_{\indi_1, \dots, \indi_r} \E_{\pi}\left[ \prod_{j\in [n]}\Trace_{\rm sum}(M_j)\right]&\leq\gamma(p)^r \exp(O_p(1)n) \; \E_{\indi_1, \dots, \indi_r} \left[ \prod_{j\in [n]} \E_{M_j} \left[\Trace_{\rm sum}(M_j) \right]\right],
\end{align*}
where the expecation is with respect to the randomness of the set $\indi_1,\ldots,\indi_r$ and the permutation $\pi$ of the associated $2r$ set of elements.
\end{conjecture}

The factor $\gamma(p)$ loosely controls the amount of independence
in the random matchings upon going from a sequence of $r$ to $r+1$ tuples.
This factor allows us to upper bound the expected maximum energy as
\begin{align}
\E[\lambda_{\text{max}}(H_{n,p})/\sqrt{n}]
\le
(1
+o_p(1))
\sqrt{2\gamma(p)\log p}.
\end{align}
In \Cref{app:proof_of_upper_bound} we will show the conjecture is satisfied by taking $\gamma\left(p\right)=3^p$; the conjectural aspect is whether or not this choice of $\gamma\left(p\right)$ is optimal.

\section{Existence of maximum energy limit for product states} 
\label{sec:limit_exists}

Throughout this study, we are considering the limiting maximum energy over product states: 
\begin{equation}
        \lim_{n \to \infty}\E \frac{1}{\sqrt{n}}\max_{\ket{\phi} \in \mathcal{S}_{\rm product}^n} \bra{\phi} H_{n,p} \ket{\phi}.
\end{equation} 
To justify the existence of this limit (for even $p$), we extend classical techniques showing super-additivity of the energy in the Sherrington--Kirkpatrick model \cite{talagrand2010mean,guerra2002thermodynamic}, a condition which guarantees the existence of the limit by Fekete's Lemma \cite{fekete1923verteilung}. For our purposes, we will use a weaker notion of \emph{near super-additivity} which also guarantees the existence of the limit \cite{de1951some} (see \Cref{lem:fekete}).
\begin{definition}[Near super-additive sequence] \label{def:near_super_additive}
    A sequence of numbers ${a_n}$ is near super-additive if there exists some function $f:\mathbb{R}^+ \to \mathbb{R}^+$ satisfying 
    \begin{equation}
        \int_1^\infty \frac{f(t)}{t^2} \; dt < \infty
    \end{equation}
    such that $a_{n+m} \geq a_n + a_m - f(n+m)$ for all $n,m>0$.    
\end{definition}

Our proof will show that the maximum energy of the Hamiltonian $H_{n,p}$ is near super-additive. Formally, we are only able to prove the existence of the limit for even $p$ as our proof relies on convexity of the covariance function for even $p$.
\begin{restatable}[Existence of limit]{lemma}{limitexists} \label{lem:limit_exists}
    In the $p$-spin Hamiltonian model defined in \Cref{eq:hamiltonian_concise_form}, for even $p$, the limit 
    \begin{equation}
        \lim_{n \to \infty}\E \frac{1}{\sqrt{n}}\max_{\ket{\phi} \in\mathcal{S}_{\rm product}^n} \bra{\phi} H_{n,p} \ket{\phi} 
    \end{equation} 
    exists.
\end{restatable}

To prove the limit exists, we consider the interpolating Hamiltonian $H_t$ defined as
\begin{equation}
    H_t = \sqrt{t} \sqrt{n+m} H_{n+m,p} + \sqrt{1-t}\left( \sqrt{n}H_{n,p} + \sqrt{m} H_{m,p} \right),
\end{equation}
where w.l.o.g.\ we assume that $H_{n,p}$ is supported on the first $n$ qubits and $H_{m,p}$ is supported on qubits $n+1, \dots, n+m$. Here, the Hamiltonians $H_{n,p},H_{m,p},H_{n+m,p}$ all have coefficients drawn independently from each other.

This Hamiltonian interpolates between $H_0 = \sqrt{n}H_{n,p} + \sqrt{m} H_{m,p} $ and $H_1=\sqrt{n+m} H_{n+m,p}$. Given a finite set of states $\mathcal{S}$, we also define the quantities
\begin{equation}
    \varphi(t) = \frac{1}{n+m} \beta^{-1} \E \log Z_\beta(t) , \quad Z_\beta(t) = \sum_{\ket{\phi} \in \mathcal{S}} \exp\left(\beta \bra{\phi} H_t \ket{\phi} \right)
\end{equation}
as the normalized free energy and partition function of the interpolating Hamiltonian respectively. Later, we will set $\mathcal{S}$ equal to some $\epsilon$-cover over the set of product states. The interpolating Hamiltonian also allows one to interpolate between the free energy of the full system and subsystem.
Accordingly, we define the free energy terms of the full system as 
\begin{equation}
    F_{\beta, n+m} \coloneqq \beta^{-1} \E \log Z_{\beta, n+m} = \beta^{-1} \E\left[ \log \sum_{\ket{\phi} \in \mathcal{S}} \exp\left(\beta \bra{\phi} \sqrt{n+m} H_{n+m,p} \ket{\phi}\right) \right] .
\end{equation}
Similar statements hold for the subsystem free energies $F_{\beta, n}$ and $F_{\beta, m}$. Note that $F_{\beta, n+m} = (n+m) \varphi(1)$.

Classical proofs of the existence of the limit for the Sherrington--Kirkpatrick model show that the limit exists by proving the super-additivity of the free energy and applying Fekete's Lemma \cite{fekete1923verteilung}. We will use an enhancement of Fekete's Lemma from de Bruijn and Erdős~\cite{de1951some} which guarantees the existence of the limit for near super-additive sequences \cite{de1951some,furedi2020nearly} (see also Theorem 1.9.2 of \cite{steele1997probability}).

\begin{lemma}[de Bruijn--Erdős extension of Fekete's Lemma \cite{de1951some}] \label{lem:fekete}
    For every near super-additive sequence $\{a_n\}$ (see \Cref{def:near_super_additive}), the limit $\lim_{n \to \infty} \frac{a_n}{n}$ exists and is equal to
    \begin{equation}
        \lim_{n \to \infty} \frac{a_n}{n} = \sup_{n \geq 1} \frac{a_n}{n}.
    \end{equation}
\end{lemma}

Our goal is to prove near super-additivity of the free energy and extend this to show that the maximum energy is also near super-additive. Both of these will be consequences of the near convexity of the covariance functions for the Hamiltonian $H_{n,p}$.
\begin{lemma}[Near super-additivity of free energy (adapted from Theorem 1.1 of \cite{panchenko2013sherrington})] \label{lem:super_additive_free_energy}
    For a constant $c>0$, the normalized interpolating free energy $\varphi(t)$ satisfies the relation
    \begin{equation}
        \varphi(1) \geq \varphi(0) - \frac{\beta c}{n+m}.
    \end{equation}
\end{lemma}
\begin{proof}
    The derivative of $\varphi(t)$ with respect to $t$ satisfies
    \begin{equation}
        \varphi'(t) = (n+m)^{-1} \beta^{-1} \mathbb{E} \left[ \left\langle \beta \frac{\partial}{\partial t} \bra{\phi} H_t \ket{\phi} \right\rangle_{\ket{\phi} \in \mathcal{S}} \right],
    \end{equation}
    where 
    \begin{equation}
    \left \langle f(\ket{\phi}) \right\rangle_{\ket{\phi} \in \mathcal{S}} \coloneqq \sum_{\ket{\phi} \in \mathcal{S}} \frac{\exp\left(\beta \bra{\phi} H_t \ket{\phi} \right)}{Z_\beta(t)} f(\ket{\phi})
    \end{equation}
    denotes the average of the function $f:\mathcal{S} \to \mathbb{R}$ with respect to the Gibbs measure. We can evaluate the expectation of the Gibbs average via standard formulas derived by Gaussian integration by parts (see Lemma 1.1 in \cite{panchenko2013sherrington}), giving
    \begin{equation}
        \varphi'(t) = (n+m)^{-1} \beta \mathbb{E} \left[ \left\langle  \mathbb{E} \left[ \bra{\phi} H_t \ket{\phi} \frac{\partial}{\partial t} \bra{\phi} H_t \ket{\phi} \right] - \mathbb{E} \left[ \bra{\phi} H_t \ket{\phi} \frac{\partial}{\partial t} \bra{\psi} H_t \ket{\psi} \right]  \right\rangle_{\ket{\phi},\ket{\psi} \in \mathcal{S}} \right].
    \end{equation}
    Evaluating the derivative, we note that
    \begin{equation}
    \begin{split}
        \mathbb{E} \left[ \bra{\phi} H_t \ket{\phi} \frac{\partial}{\partial t} \bra{\psi} H_t \ket{\psi} \right] &= \frac{n+m}{2} \mathbb{E}\left[ \bra{\phi}H_{n+m,p}\ket{\phi} \bra{\psi}H_{n+m,p}\ket{\psi} \right] \\
        &\quad - \frac{n}{2} \mathbb{E}\left[ \bra{\phi}H_{n,p}\ket{\phi} \bra{\psi}H_{n,p}\ket{\psi} \right] - \frac{m}{2} \mathbb{E}\left[ \bra{\phi}H_{m,p}\ket{\phi} \bra{\psi}H_{m,p}\ket{\psi} \right].
    \end{split}
    \end{equation}
    \Cref{lem:super-additivity-covariance} evaluates the variance and covariance terms above. Namely, we have
    \begin{equation}
    \begin{split}
        \varphi'(t) = -(n+m)^{-1} \frac{\beta}{2} \E \biggl\langle  & \E\left[  (n+m)\bra{\phi} H_{n+m,p} \ket{\phi} \bra{\psi} H_{n+m,p} \ket{\psi}\right] \\
        &- \E\left[ n \bra{\phi} H_{n,p} \ket{\phi} \bra{\psi} H_{n,p} \ket{\psi}\right]  - \E\left[  m\bra{\phi} H_{m,p} \ket{\phi} \bra{\psi} H_{m,p} \ket{\psi}\right]\biggr\rangle_{\ket{\phi},\ket{\psi} \in \mathcal{S}}.
    \end{split}
    \end{equation}
    Applying the inequality in \Cref{lem:super-additivity-covariance} with a constant $c>0$,
    \begin{equation}
    \begin{split}
        \varphi'(t) \geq -(n+m)^{-1} \frac{\beta}{2}  c .
    \end{split}
    \end{equation}
    Integrating from $t=0$ to $t=1$ completes the proof.
\end{proof}

The previous Lemma proves that the free energy is near super-additive. Using standard bounds relating the free energy to the maximum energy, we can extend this near super-additivity to the maximum energy as well to prove \Cref{lem:limit_exists}.

\begin{proof}[Proof of \Cref{lem:limit_exists}]
    For a given finite set of product states $\mathcal{S}$, let $\hat{E}_{n+m}$ denote the maximum energy defined as
    \begin{equation}
        \hat{E}_{n+m} \coloneqq \E \max_{\ket{\phi} \in \mathcal{S}} \bra{\phi} \sqrt{n+m}H_{m+n,p} \ket{\phi}.
    \end{equation}
    Similar definitions hold for the maximum energies of the subsystems $\hat{E}_n$ and $\hat{E}_m$. We have the following standard bounds for $\hat{E}_{n+m}$ in relation to the free energy which hold for any $\beta>0$:
    \begin{equation}
       F_{\beta, n+m} \geq \hat{E}_{n+m} \geq F_{\beta, n+m} - \frac{\log | \mathcal{S}|}{\beta}.
    \end{equation}
    The first bound follows from noting $\hat{E}_{n+m} = \beta^{-1} \log \exp(\beta \hat{E}_{n+m}) \leq F_{\beta, n+m}$. The second bound follows from noting that $\beta^{-1} \log \exp(\beta\hat{E}_{n+m}) \geq \beta^{-1} \log(Z_\beta/|\mathcal{S}|). $ Furthermore, from \Cref{lem:super_additive_free_energy}, we have that for some $c>0$,
    \begin{equation}
        F_{\beta, n+m} - F_{\beta, n} - F_{\beta, m} \geq -\beta c.
    \end{equation}
    Combining the previous two equations,
    \begin{equation} \label{eq:intermediate_in_energy_max}
        \hat{E}_{n+m} - \hat{E}_{n} - \hat{E}_{m} \geq -\beta c - \frac{\log | \mathcal{S}|}{\beta}.
    \end{equation}

    We now let $S_\epsilon$ be an $\epsilon$ cover of the pure states on a single qubit. More precisely, for any state $\ket{\phi} \in \mathcal{S}_{\rm all}^1$ there exists $\ket{\phi'} \in S_\epsilon$ such that $\| \ket{\phi} - \ket{\phi'}\|_2 \leq \epsilon$. For product states on $n+m$ qubits, we extend the epsilon cover by taking the product of $n+m$ sets $S_\epsilon$:
    \begin{equation}
        S_\epsilon^{n+m} = \{ \ket{\phi} = \ket{\phi^{(1)}} \otimes \cdots \otimes \ket{\phi^{(n+m)}} \in \mathcal{S}_{\rm product}^{n+m}: \ket{\phi^{(i)}} \in S_\epsilon \; \; \forall i \in [n+m]\}.
    \end{equation}
    
    The $S_\epsilon^{n+m}$ guarantees the existence of a product state $\ket{\phi_\epsilon}$ within $1-\epsilon(n+m)$ in overlap squared $| \bra{\phi_\epsilon}\ket{ \phi}|^2$ to any product state $\ket{\phi}$. To see this, we choose a $\epsilon$-close state on each qubit $\ket{\phi_\epsilon} = \ket{\phi_\epsilon^{(1)}} \otimes \cdots \otimes \ket{\phi_\epsilon^{(n+m)}}$ and have
    \begin{equation}
        1-\epsilon\leq\left(1-\frac{\epsilon}{2}\right)^2\leq\Re\left(\bra{\phi_\epsilon^{(i)}}\ket{\phi^{(i)}}\right)^2\leq\Re\left(\bra{\phi_\epsilon^{(i)}}\ket{\phi^{(i)}}\right)^2+\Im\left(\bra{\phi_\epsilon^{(i)}}\ket{\phi^{(i)}}\right)^2=\left\lvert\bra{\phi_\epsilon^{(i)}}\ket{\phi^{(i)}}\right\rvert^2,
    \end{equation}
    which implies
    \begin{equation}
        \left| \bra{\phi_\epsilon}\ket{\phi} \right|^2 \geq (1-\epsilon)^{n+m} \geq 1-\epsilon(m+n).
    \end{equation}

    Using the fact that $\E \lambda_{\text{max}}(\sqrt{m+n}H_{m+n,p}) = \Theta(m+n)$,
    \begin{equation}
    \begin{split}
        \hat{E}_{n+m} &\coloneqq 
        \E \max_{\ket{\phi} \in \mathcal{S}_{\rm product}^{n+m}}\bra{\phi} \sqrt{n+m}H_{m+n,p} \ket{\phi} \\
        &\leq \E \max_{\ket{\phi_\epsilon} \in S_\epsilon^{n+m}} \bra{\phi_\epsilon} \sqrt{n+m}H_{m+n,p} \ket{\phi_\epsilon} + \Theta(\epsilon (m+n)^2),        
    \end{split}
    \end{equation}
    and similarly for the subsystems. Applying this in \Cref{eq:intermediate_in_energy_max},
    \begin{equation} 
        \hat{E}_{n+m} - \hat{E}_{n} - \hat{E}_{m} \geq -\beta c - \frac{\log | \mathcal{S_\epsilon}|^{n+m}}{\beta} - \Theta(\epsilon (m+n)^2).
    \end{equation}
    To achieve an $\epsilon$ cover, it suffices to have $|S_\epsilon| = (1-2/\epsilon)^3$ points (e.g., see Lemma 5.2 of \cite{vershynin2010introduction}). Therefore, setting $\epsilon = (n+m)^{-k}$ for any fixed $k$,
    \begin{equation}
        \log | \mathcal{S_\epsilon}|^{m+n} = O((m+n) \log(m+n)).
    \end{equation}
    Choosing $k>1$, we have
    \begin{equation} 
        \hat{E}_{n+m} - \hat{E}_{n} - \hat{E}_{m} \geq -\beta c - O(\beta^{-1}(m+n) \log(m+n)) - o(m+n).
    \end{equation}
    Setting $\beta = \beta' \sqrt{n+m}$ for some constant $\beta'$ achieves the required bounds for near super-additivity. Applying \Cref{lem:fekete} completes the proof.
\end{proof}

Finally, we complete the proof of \Cref{lem:limit_exists} by proving the near sub-additivity of the subsystem covariance, which was referred to in our proof of \Cref{lem:super_additive_free_energy}.

\begin{lemma}[Near sub-additivity of subsystem covariance]\label{lem:super-additivity-covariance}
    Given product states $\ket{\phi}, \ket{\psi} \in \mathcal{S}_{\rm product}^{n+m}$, there exists a constant $c>0$ such that
    \begin{equation}
    \begin{split}
        \E\left[  (n+m)\bra{\phi} H_{n+m,p} \ket{\phi} \bra{\psi} H_{n+m,p} \ket{\psi}\right] &\leq \E\left[ n \bra{\phi} H_{n,p} \ket{\phi} \bra{\psi} H_{n,p} \ket{\psi}\right]  \\
        &\quad + \E\left[  m\bra{\phi} H_{m,p} \ket{\phi} \bra{\psi} H_{m,p} \ket{\psi}\right] + c
    \end{split}
    \end{equation}
    for even $p$. Additionally, the variance terms satisfy
    \begin{equation}
    \begin{split}
        n+m &= \E\left[  (n+m)\bra{\phi} H_{n+m,p} \ket{\phi} \bra{\phi} H_{n+m,p} \ket{\phi}\right] \\
        &= \E\left[ n \bra{\phi} H_{n,p} \ket{\phi} \bra{\phi} H_{n,p} \ket{\phi}\right]  + \E\left[  m\bra{\phi} H_{m,p} \ket{\phi} \bra{\phi} H_{m,p} \ket{\phi}\right].
    \end{split}
    \end{equation}
\end{lemma}
\begin{proof}
    Let $\ket{\phi}, \ket{\psi} \in \mathcal{S}_{\rm product}^{n+m}$ where $\ket{\phi}=\ket{\phi^{(1)}} \otimes \cdots \otimes \ket{\phi^{(n+m)}}$ and $\ket{\psi}= \ket{\psi^{(1)}} \otimes \cdots \otimes \ket{\psi^{(n+m)}}$ denote two arbitrary product states. From \Cref{lem:covariance_product}, we have:
    \begin{equation}
        \E\left[\bra{\phi} H_{n+m,p} \ket{\phi} \bra{\psi} H_{n+m,p} \ket{\psi}\right] = \frac{1}{\binom{n+m}{p}} \sum_{\indi \in \mathcal{I}_p^{n+m}} \prod_{k=1}^p \left( 2 \left|\bra{\phi^{(\indisub_k)}}\ket{ \psi^{(\indisub_k)}} \right|^2 - 1\right).
    \end{equation}

Using the above, it is straightforward to note that 
\begin{equation}
    \mathbb{E}\left[(n+m) \left(\bra{\phi} H_{n+m,p} \ket{\phi}\right)^2 \right] = n+m.
\end{equation}

Similar statements hold for the subsystem Hamiltonians $H_{n,p}$ and $H_{m,p}$. For convenience, from here on in the proof, let $R_k = \left( 2 \left|\bra{\phi^{(k)}}\ket{ \psi^{(k)}} \right|^2 - 1\right)$. Then,
\begin{equation}
\begin{split}
    \E[ \bra{\phi} (\sqrt{n}H_{n,p} + \sqrt{m}H_{m,p}) \ket{\phi} & \bra{\psi} (\sqrt{n}H_{n,p} + \sqrt{m}H_{m,p}) \ket{\psi} ] \\ &= n {n \choose p}^{-1} E_p(R_1, \dots, R_n) +   m {m \choose p}^{-1} E_p(R_{n+1}, \dots R_{n+m}),
\end{split}
\end{equation}
where $E_p$ is the $p$-th degree elementary symmetric polynomial defined as 
\begin{equation}
    E_p(x_1, \dots, x_n) = \sum_{1 \leq i_1 < \cdots < i_p \leq n} x_{i_1} x_{i_2} \cdots x_{i_p}.
\end{equation}

Similarly, 
\begin{equation}
    \E\left[  (n+m)\bra{\phi} H_{n+m,p} \ket{\phi} \bra{\psi} H_{n+m,p} \ket{\psi}\right]  = (m+n) {m+n \choose p}^{-1} E_p(R_{1}, \dots, R_{n+m}).
\end{equation}

We are left with showing the following inequality holds:
\begin{equation}
    (m+n)\frac{E_p(R_{1}, \dots, R_{n+m})}{ {m+n \choose p}} \leq n\frac{E_p(R_{1}, \dots, R_{n})}{ {n \choose p}} + m\frac{E_p(R_{n+1}, \dots, R_{m+n})}{  {m \choose p}} + O(1).
\end{equation}

Note 
\begin{equation}
 \frac{(m+n)^p}{{m+n \choose p}}=p!+O(1/(m+n)).
 \end{equation}
Similar identities  hold for ${m \choose p}$ and ${n \choose p}$. Thus it suffices to prove that
\begin{equation}
    \frac{E_p(R_{1}, \dots, R_{n+m})}{(m+n)^{p-1}} 
    \leq \frac{E_p(R_{1}, \dots, R_{n}) }{n^{p-1}} + \frac{E_p(R_{n+1}, \dots, R_{m+n}) }
    {m^{p-1}}+O(1). \label{eq:super-add}
\end{equation}

We have 
\begin{equation}
p!E_p(R_{1}, \dots, R_{n+m})=\sum_{i_1\ne i_2\ne\cdots \ne i_p} R_{i_1} R_{i_2} \cdots R_{i_p}.
\end{equation}
Note that 
\begin{equation}
\sum_{i_1\ne i_2\ne\cdots \ne i_p} R_{i_1} R_{i_2} \cdots R_{i_p}=
\sum_{1\le i_1,\ldots,i_p\le n}R_{i_1} R_{i_2} \cdots R_{i_p}-\sum_{(i_1,\ldots,i_p)\in \mathcal{I}_{\rm repeat}}
R_{i_1} R_{i_2} \cdots R_{i_p},
\end{equation}
where $\mathcal{I}_{\rm repeat}$ denotes the set of tuples $(i_1,\ldots,i_p)$ with at least one element repeated.
We have that the first sum is $\left(\sum_{1\le i\le m+n} R_i\right)^{p}$. At the same time, since each $R_i$ is bounded by $1$
in absolute value, the sum 
$\sum_{(i_1,\ldots,i_p)\in \mathcal{I}_{\rm repeat}}$ is of the order $O\left( (m+n)^{p-1}\right)$. Thus, to show \Cref{eq:super-add},
it suffices to verify that
\begin{equation}
{\left(\sum_{1\le i\le m+n} R_i\right)^{p} \over (m+n)^{p-1}}\le 
{\left(\sum_{1\le i\le m} R_i\right)^{p} \over m^{p-1}}+
{\left(\sum_{m+1\le i\le m+n} R_i\right)^{p} \over n^{p-1}}.
\end{equation}
We have:
\begin{equation}
\begin{split}
{\left(\sum_{1\le i\le m+n} R_i\right)^{p} \over (m+n)^{p-1}}&=(m+n){\left(\sum_{1\le i\le m+n} R_i \over m+n\right)^{p}} \\
&=(m+n)\left({m\over m+n} {\sum_{1\le i\le m} R_i \over m}+{n\over m+n}{\sum_{m+1\le i\le m+n} R_i \over n}\right)^{p} \\
&\le 
(m+n){m\over m+n} \left({\sum_{1\le i\le m} R_i \over m}\right)^p+(m+n){n\over m+n}\left({\sum_{m+1\le i\le m+n} R_i \over n}\right)^{p} \\
&=
{\left(\sum_{1\le i\le m} R_i\right)^{p} \over m^{p-1}}+
{\left(\sum_{m+1\le i\le m+n} R_i\right)^{p} \over n^{p-1}},
\end{split}
\end{equation}
and the claim is established. The inequality uses the convexity of the function $x^p$ for even $p$.

\end{proof}

\section{Universality of the ground state values}
\label{sec:equivalence_distribution}

\input{equivalence_section}

\section{Statistical properties of entangled states}
\label{sec:properties_entangled_states}

One intuition for why product states can be effective at optimizing random spin glass Hamiltonians is because on average, product states have exponentially larger variance in comparison to completely random states. Similar intuition was given in \cite{Harrow2017extremaleigenvalues}, which studied product state approximations for the ground states of sparse models. Here, we consider the statistical properties of entangled states to more formally describe this intuition and detail the technical challenges in working with entangled states.

\paragraph{Variance of random states}
We first show that states chosen uniformly at random from the Haar measure on the set of pure states have variance on average exponentially smaller than that of product states.

\begin{lemma}[Variance of random states] \label{lem:variance_random_states}
    Let $\mathcal{D}_{\mathcal{S}_{\rm all}^n}$ denote the uniform distribution over the set of all possible states $\mathcal{S}_{\rm all}^n$ on $n$ qubits. The expected variance of the $p$-spin glass Hamiltonian $H_{n,p}$ over this distribution is
    \begin{equation}
        \E_{\ket{\phi} \sim \mathcal{D}_{\mathcal{S}_{\rm all}^n}} \E \left[ \left(\bra{\phi} H_{n,p} \ket{\phi}\right)^2\right] = \frac{3^p}{2^n+1}.
    \end{equation}
\end{lemma}
\begin{proof}
    For a given state $\ket{\phi} \in \mathcal{S}_{\rm all}^n$, we have that its variance is
\begin{equation} 
\begin{split}
    \E\left[ \left(\bra{\phi} H_{n,p} \ket{\phi}\right)^2\right] &= \frac{1}{\binom{n}{p}} \sum_{\indi \in \mathcal{I}_p^n} \sum_{\inda \in \{1,2,3\}^p} \bra{\phi} P_{\indi}^{\inda} \ket{\phi} \bra{\phi} P_{\indi}^{\inda} \ket{\phi} \\
    &= \frac{1}{\binom{n}{p}} \sum_{\indi \in \mathcal{I}_p^n}  \bra{\phi} \bra{\phi} \left[\sum_{\inda \in \{1,2,3\}^p} P_{\indi}^{\inda} \otimes P_{\indi}^{\inda} \right] \ket{\phi}  \ket{\phi} \\
    &= \frac{1}{\binom{n}{p}} \sum_{\indi \in \mathcal{I}_p^n} \sum_{\inda \in \{1,2,3\}^p} \Tr[ \ket{\phi}  \ket{\phi} \bra{\phi} \bra{\phi}  P_{\indi}^{\inda} \otimes P_{\indi}^{\inda} ].
\end{split}
\end{equation}

For a given term above, we can take the average over all possible states by replacing the uniform distribution over states with that of the Haar distribution $\mathcal{U}_{2^n}$ over the $2^n$ dimensional unitary group of matrices:
\begin{equation} \label{eq:intermediate_variance_average_state}
    \E_{\ket{\phi} \sim \mathcal{D}_{\mathcal{S}_{\rm all}^n}} \E\left[ \left(\bra{\phi} H_{n,p} \ket{\phi}\right)^2\right] = \frac{1}{\binom{n}{p}} \sum_{\indi \in \mathcal{I}_p^n} \sum_{\inda \in \{1,2,3\}^p} \Tr\left[ \E_{U \sim \mathcal{U}_{2^n}} \left[ U \otimes U \ket{0}  \ket{0} \bra{0} \bra{0} U^\dagger \otimes U^\dagger \right] P_{\indi}^{\inda} \otimes P_{\indi}^{\inda} \right].
\end{equation}

The matrix $\E_{U \sim \mathcal{U}_{2^n}} \left[ U \otimes U \ket{0}  \ket{0} \bra{0} \bra{0} U^\dagger \otimes U^\dagger \right]$ has a closed form that can be evalauted via the Weingarten calculus \cite{collins2006integration}. Namely, for our given form, we can use Equation 7.186 of \cite{watrous2018theory} to evaluate this matrix as 
\begin{equation}
    \begin{split}
        \E_{U \sim \mathcal{U}_{2^n}} \left[ U \otimes U \ket{0}  \ket{0} \bra{0} \bra{0} U^\dagger \otimes U^\dagger \right] &= \frac{1}{2 {2^n+1 \choose 2}} \left(  I \otimes I + \sum_{a,b=1}^{2^n} E_{a,b} \otimes E_{b,a} \right) ,
    \end{split}
\end{equation}
where $E_{a,b}$ is the indicator matrix equal to one in entry $a,b$ and zero in every other entry. Noting that Pauli matrices are traceless and 
\begin{equation}
    \Tr\left[\left( \sum_{a,b=1}^{2^n} E_{a,b} \otimes E_{b,a} \right) P_{\indi}^{\inda} \otimes P_{\indi}^{\inda} \right] = \Tr[ P_{\indi}^{\inda} P_{\indi}^{\inda} ] = 2^n,
\end{equation}
we therefore have for any Pauli $P_{\indi}^{\inda}$ that
\begin{equation}
    \begin{split}
        \Tr\left[ \E_{U \sim \mathcal{U}_{2^n}} \left[ U \otimes U \ket{0}  \ket{0} \bra{0} \bra{0} U^\dagger \otimes U^\dagger \right] P_{\indi}^{\inda} \otimes P_{\indi}^{\inda} \right] &= \frac{1}{2 {2^n+1 \choose 2}} 2^n = \frac{1}{2^n+1}.
    \end{split}
\end{equation}

Plugging in the above into \Cref{eq:intermediate_variance_average_state}, we arrive at the final result:
\begin{equation} 
    \E_{\ket{\phi} \sim \mathcal{D}_{\mathcal{S}_{\rm all}^n}} \E\left[ \left(\bra{\phi} H_{n,p} \ket{\phi}\right)^2\right] = \frac{3^p}{2^n+1}.
\end{equation}

\end{proof}

\Cref{lem:variance_random_states} shows that the variance of the $p$-spin Hamiltonian decays exponentially fast with the number of qubits on average for random states. In fact, noting that the variance function is $\sqrt{3^p}$ Lipschitz (see proof in \Cref{lem:gaussian_concentration}), the variance is also concentrated around this average with high probability. This can equivalently be stated in terms of Levy's lemma (see Theorem 7.37 of \cite{watrous2018theory}).
\begin{corollary}
    By Levy's lemma, for any $\epsilon>0$, there exists a constant $c(p)$ which depends on $p$ but which is independent of $n$ such that
    \begin{equation}
        \mathbb{P}_{\ket{\psi} \sim \mathcal{D}_{\mathcal{S}_{\rm all}^n}}\left[ \left| \E\left[\left(\bra{\psi} H_{n,p} \ket{\psi}\right)^2 \right]  - \E_{\ket{\phi} \sim \mathcal{D}_{\mathcal{S}_{\rm all}^n}} \E\left[ \left(\bra{\phi} H_{n,p} \ket{\phi}\right)^2\right]  \right|  \geq \epsilon \right] \leq 3 \exp\left( -c(p) \epsilon^2 2^n \right).
    \end{equation}
\end{corollary}

As an aside, the exponential decay in variance also holds over any distribution of states that forms a unitary 2-design \cite{harrow2023approximate,gross2007evenly}. Furthermore, this exponential decay in the variance has implications for optimizing random spin glass Hamiltonians, and is related to an untrainability phenomenon in quantum variational optimization where it is shown that relevant quantities of optimization typically decay exponentially for collections of random states \cite{mcclean2018barren,cerezo2021variational,anschuetz2022critical,anschuetz2022quantum}.

\paragraph{Variance as function of purity}

More generally, the variance of the Hamiltonian with respect to a state $\ket{\phi}$ can be written as a function of the average local purity of the state as detailed below. In this form, the relationship between the entanglement properties of the state and the variance are more apparent. Notably, unlike product states, entangled states can have widely ranging values for the variance. This fact, in part, makes it challenging to prove the existence of the limit of the ground state energy for all states.
\begin{proposition}[Variance as function of local purity]
    Given a state $\ket{\phi} \in \mathcal{S}_{\rm all}^n$, the variance of the energy of the random Hamiltonian $H_{n,p}$ for $\ket{\phi}$ is equal to 
    \begin{equation}
        \E\left[ \left(\bra{\phi} H_{n,p} \ket{\phi}\right)^2\right] =  \sum_{k=0}^p (-1)^{p-k}  2^k {p \choose k} A_k,
    \end{equation}
    where $A_k$ is the average purity over all possible subsystems of $k$ qubits. More formally, 
    \begin{equation}
        A_k = {n \choose k}^{-1} \sum_{S \in {[n] \choose k}} \Tr(\rho_S^2),\label{eq:avg_purity}
    \end{equation}
    where $\rho_S = \Tr_{[n] \setminus S}( \ket{\phi}\bra{\phi} )$ is the reduced density matrix of $\ket{\phi}$ for the qubits in $S$ and ${[n] \choose k}$ denotes the set of ordered $k$-tuples drawn without replacement from $[n]$.
\end{proposition}
\begin{proof}
    
For a given state $\ket{\phi} \in \mathcal{S}_{\rm all}^n$, we have that its variance is
\begin{equation} 
\begin{split}
    \E\left[ \left(\bra{\phi} H_{n,p} \ket{\phi}\right)^2\right] &= \frac{1}{\binom{n}{p}} \sum_{\indi \in \mathcal{I}_p^n} \sum_{\inda \in \{1,2,3\}^p} \bra{\phi} P_{\indi}^{\inda} \ket{\phi} \bra{\phi} P_{\indi}^{\inda} \ket{\phi} \\
    &= \frac{1}{\binom{n}{p}} \sum_{\indi \in \mathcal{I}_p^n}  \bra{\phi} \bra{\phi} \left[\sum_{\inda \in \{1,2,3\}^p} P_{\indi}^{\inda} \otimes P_{\indi}^{\inda} \right] \ket{\phi}  \ket{\phi} \\
    &= \frac{1}{\binom{n}{p}} \sum_{\indi \in \mathcal{I}_p^n}  \bra{\phi} \bra{\phi} \left[ \prod_{k=1}^p \left[2 (\operatorname{SWAP}) - I \right]_{\indisub_k, \indisub_k+n} \right] \ket{\phi}  \ket{\phi},
\end{split}
\end{equation}
where in the last line we used \Cref{eq:sum_pauli_to_swap} and use subscripts $\left[2 (\operatorname{SWAP}) - I \right]_{i,j}$ to indicate that the $\operatorname{SWAP}$ operator acts on qubits $i$ and $j$. As the state $\ket{\phi}$ is no longer necessarily a product state, we let $\rho_{S}$ denote its reduced state on the qubits in the set $S \subseteq [n]$. Furthermore, for $k\in [n]$, we let $A_k$ denote the average purity of the state $\ket{\phi}$ over all reduced states supported on $k$ qubits as formally defined in \Cref{eq:avg_purity}. With this notation, we can calculate the variance as
\begin{equation} 
\begin{split}
    \E\left[ \left(\bra{\phi} H_{n,p} \ket{\phi}\right)^2\right] &= \frac{1}{\binom{n}{p}} \sum_{k=0}^p \sum_{\indi \in \mathcal{I}_p^n} (-1)^{p-k} 2^k \sum_{\{ r_1, \dots, r_k \} \in {[p] \choose k}}  \Tr( \rho_{\{\indisub_{r_1}, \dots, \indisub_{r_k} \}}^2) \\
    &=  \sum_{k=0}^p (-1)^{p-k}  2^k {p \choose k} A_k.
\end{split}
\end{equation}
In the above, we use the fact that the sum over $\indi \in \mathcal{I}_p^n$ is uniform over all possible $p$-tuples, and thus $k$-tuples in turn drawn uniformly from the draws of the $p$-tuples are thus also uniform over all possible $k$-tuples.

\end{proof}

\paragraph{Product states maximize variance for large $n$}

The results here show that product states have large variance, though we should caution that product states do not necessarily achieve the maximum variance for the spin glass Hamiltonians $H_{n,p}$. As a simple example, for $n=2$ and $p=2$, the GHZ state $\frac{1}{\sqrt{2}}\left( \ket{00} + \ket{11} \right)$ has higher variance than any product state. Nonetheless, the $n=2$ setting is the only one where such a counter-example can be found for $p=2$. We prove this fact here. Our proof is based on a certain uncertainty relation which bounds the maximum variance via the Lovász theta function of a graph constructed from the Hamiltonian \cite{de2023uncertainty} (see also \cite{xu2023bounding,hastings2022optimizing} for similar constructions). Finding bounds of this form for general $p>2$ can be challenging, though numerically we find that similar bounds seem to hold for $p>2$. Thus, we believe our proof can be extended to spin glass models where $p>2$, but we do not concern ourselves with this strengthening here for simplicity. To present this bound, we first define the anti-commutativity graph of a set of Paulis.

\begin{definition}[Anti-commutativity graph \cite{de2023uncertainty}] \label{def:anti_comm_graph}
    Given a subset $S \subseteq \mathcal{P}_n$ of the $n$-qubit Pauli operators, the anti-commutativity graph $G_S$ of $S$ is a graph containing a node for each Pauli matrix in $S$ ($|S|$ total nodes) and for each pair of nodes $i$ and $j$ associated to Paulis $P_i$ and $P_j$, there is an edge between $i$ and $j$ if and only if $P_i P_j + P_j P_i = 0$ (i.e., $P_i$ and $P_j$ anti-commute).
\end{definition}

Bounds on the variance are given by the Lovász theta function, denoted as $\vartheta(G)$, and defined by the following semi-definite program.
\begin{definition}[Lovász theta function \cite{lovasz1979shannon,de2023uncertainty}] \label{def:lovasz}
    The Lovász theta function $\vartheta(G)$ of a graph $G$ on $n$ nodes is equal to
    \begin{equation}
        \begin{split}
        \vartheta(G) = \max_{\substack{B \in \mathbb{R}^{n \times n} }} & \lambda_{\max}(B)    \\ 
         \text{s.t. } &B = B^\top, B \succeq 0, \\ & B_{ij} = 0 \text{ for edges } i \sim j \text{ in $G$}, \\ & B_{ii} = 1 \text{ for all } i \in \{1,\dots, n\},
        \end{split}
    \end{equation}
    where $\lambda_{\max}(B)$ is the largest eigenvalue of $B$.
\end{definition}
The Lovász theta function also has an equivalent minimization formulation \cite{knuth1993sandwich} over the complement graph 
$\overline{G}$:
\begin{equation}
    \begin{split}
    \vartheta(G) = \min_{\substack{B \in \mathbb{R}^{n \times n} }} & \lambda_{\max}(B),    \\ 
     \text{s.t. } &B = B^\top, \\ & B_{ij} = 1 \text{ for edges } i \sim j \text{ in $\overline{G}$}, \\ & B_{ii} = 1 \text{ for all } i.
    \end{split}
\end{equation}

\begin{lemma}[Maximum variance over set of Paulis \cite{de2023uncertainty}] \label{lem:max_variance_lovasz}
    Given a subset $S \subseteq \mathcal{P}_n$ of the $n$-qubit Pauli operators, it holds that
        \begin{equation}
            \max_{\ket{\phi} \in \mathcal{S}_{\rm all}^n} \sum_{P \in S} \left( \bra{\phi} P \ket{\phi} \right)^2 \leq \vartheta(G_S),
        \end{equation}
        where $\vartheta(G_S)$ is the Lovász theta function of the anti-commutativity graph (see \Cref{def:anti_comm_graph,def:lovasz}).
\end{lemma}
\begin{proof}
    We follow the proof of \cite{de2023uncertainty}. For a state $\rho$, consider the operator $E_\rho$ equal to
    \begin{equation}
        E_\rho = \sum_{P_i \in S} \Tr[P_i \rho] P_i.
    \end{equation}
    Let $v \in \mathbb{R}^{|S|}$ be a vector with entry $i$ set to $v_i = \Tr[P_i \rho]$. Then, by convexity, we have 
    \begin{equation} \label{eq:lovasz_pre_pre}
        \Tr\left( E_\rho^2 \rho\right) \geq \Tr\left( E_\rho \rho\right)^2.
    \end{equation}
    The right hand side is equal to $\|v\|_2^4$. The left hand side is equal to 
    \begin{equation} \label{eq:lovasz_pre_matrix}
        \Tr\left( E_\rho^2 \rho\right) = \sum_{P_i, P_j \in S } v_i v_j \Tr[P_i P_j \rho] = \sum_{\substack{P_i, P_j \in S \\ P_i P_j + P_i P_j \neq 0} } v_i v_j \operatorname{Re}\left(\Tr[P_i P_j \rho]\right).
    \end{equation}
    In the last step, we used the fact that $\Tr[P_i P_j \rho]^* = \Tr[P_j P_i \rho]$. Denoting $M \in \mathbb{R}^{|S| \times |S|}$ with entries $M_{ij} = \operatorname{Re}\left(\Tr[P_i P_j \rho]\right)$, we note that $M$ is positive semi-definite and has diagonal entries set to $1$ since $\Tr[P_i P_j \rho] = 1$ when $P_i=P_j$. Furthermore $M_{ij} = 0$ whenever $P_i P_j + P_i P_j = 0$ so $M = G_S$ and is the anti-commutativity graph of $S$. Since $\sum_{\substack{P_i, P_j \in S \\ P_i P_j + P_i P_j \neq 0} } v_i v_j M_{ij} \leq \|v\|^2 \lambda_{\max}(M)$, where $\lambda_{\max}(M)$ is the largest eigenvalue of $M$, we have returning to \Cref{eq:lovasz_pre_pre} that for any state $\rho$:
    \begin{equation}
    \begin{split}
        \sum_{P_i \in S} \Tr[P_i \rho]^2 \leq \max_{\substack{B \in \mathbb{R}^{|S| \times |S|} }} & \lambda_{\max}(B),    \\ 
         \text{s.t. } &B = B^\top, B \succeq 0, \\ & B_{ij} = 0 \text{ for edges } i \sim j \text{ in $G_S$}, \\ & B_{ii} = 1 \text{ for all } i.
    \end{split}
    \end{equation}
    The proof is complete upon noting that the above is $\vartheta(G_S)$. 
\end{proof}

Finally, we can show that for $n>2$, the variance of the $p=2$ spin glass Hamiltonian for any state $\ket{\phi} \in \mathcal{S}_{\rm all}^n$ is at most $1$ using the above fact. 
\begin{proposition}[Maximum limiting variance of spin glass Hamiltonian]
    For $p=2$ and any $n \geq 3$, the maximum variance of the energy of the random Hamiltonian $H_{n,p}$ over $\ket{\phi} \in \mathcal{S}_{\rm all}^n$ is 
    \begin{equation}
        \max_{\ket{\phi} \in \mathcal{S}_{\rm all}^n}\E\left[ \left(\bra{\phi} H_{n,p} \ket{\phi}\right)^2\right] = 1.
    \end{equation}
\end{proposition}
\begin{proof}
    Proving that the variance is bounded by $1$ for any state suffices since we have previously shown that product states achieve that bound. For a given state $\ket{\phi} \in \mathcal{S}_{\rm all}^n$, we have that its variance is
\begin{equation} 
\begin{split}
    \E\left[ \left(\bra{\phi} H_{n,p} \ket{\phi}\right)^2\right] &= \frac{1}{\binom{n}{p}} \sum_{\indi \in \mathcal{I}_p^n} \sum_{\inda \in \{1,2,3\}^p} \bra{\phi} P_{\indi}^{\inda} \ket{\phi} \bra{\phi} P_{\indi}^{\inda} \ket{\phi}.
\end{split}
\end{equation}
We now apply \Cref{lem:max_variance_lovasz} and set $S$ equal to the set of ${n \choose 2}9$ many $2$-local Paulis on $n$ qubits. Denoting the commutation graph of this set of ${n \choose 2}9$ Paulis as $G_n$, we have
\begin{equation}
    \sum_{\indi \in \mathcal{I}_p^n} \sum_{\inda \in \{1,2,3\}^p} \bra{\phi} P_{\indi}^{\inda} \ket{\phi} \bra{\phi} P_{\indi}^{\inda} \ket{\phi} \leq \vartheta(G_n).
\end{equation}
We show in \Cref{app:Lovasz_bounds} that for $n\geq 4$, $\vartheta(G_n) \leq {n \choose 2}$ which gives a bound of one on the variance. The case of $n = 3$ can be verified numerically by solving for $\vartheta(G_3)$.
\end{proof}

\paragraph{Adjusted model}
If we consider a slightly adjusted model $H_{n, (1, \dots, p)}$ which includes all up-to-$p$-local Paulis and is defined as 
\begin{equation} 
     H_{n,(1, \dots, p)} = \frac{1}{2^{p/2}{n \choose p}^{1/2}   } \sum_{\indi \in \mathcal{I}_p^n} \sum_{\inda \in \{0,1,2,3\}^p} \alpha[\indi; \inda] \; P_{\indi}^{\inda},
\end{equation}
then it is straightforward to show that product states maximize the variance of this model for all $n$ and $p$. Note the subtle difference in the above model where the sum $\sum_{\inda \in \{0,1,2,3\}^p}$ includes identity as well. We also multiply by a constant factor $2^{-p/2}$ to normalize it as variance one for product states.

\begin{lemma} \label{lem:max_variance_p2}
    The variance of the energy of the adjusted Hamiltonian $H_{n,(1, \dots, p)}$ for any product state $\ket{\phi} \in \mathcal{S}_{\rm product}^n$ is $\E [(\bra{\phi} H_{n,(1, \dots, p)} \ket{\phi})^2] = 1$ and is the maximum possible variance over the set of all states, i.e., for any state $\ket{\psi} \in \mathcal{S}_{\rm all}^n$,
    \begin{equation}
        \E \left[(\bra{\psi} H_{n,(1, \dots, p)} \ket{\psi})^2\right] \leq \E \left[(\bra{\phi} H_{n,(1, \dots, p)} \ket{\phi})^2\right] = 1. 
    \end{equation}
\end{lemma}
\begin{proof}
    For a given state $\ket{\phi} \in \mathcal{S}_{\rm all}^n$ we have that its variance is
\begin{equation} 
\begin{split}
    \E\left[ \left(\bra{\phi} H_{n,(1,\dots,p)} \ket{\phi}\right)^2\right] &= \frac{1}{\binom{n}{p} 2^p} \sum_{\indi \in \mathcal{I}_p^n} \sum_{\inda \in \{0,1,2,3\}^p} \bra{\phi} P_{\indi}^{\inda} \ket{\phi} \bra{\phi} P_{\indi}^{\inda} \ket{\phi} \\
    &= \frac{1}{\binom{n}{p} 2^p} \sum_{\indi \in \mathcal{I}_p^n}  \bra{\phi} \bra{\phi} \left[\sum_{\inda \in \{0,1,2,3\}^p} P_{\indi}^{\inda} \otimes P_{\indi}^{\inda} \right] \ket{\phi}  \ket{\phi} \\
    &= \frac{1}{\binom{n}{p} 2^p} \sum_{\indi \in \mathcal{I}_p^n}  \bra{\phi} \bra{\phi} \left[ \prod_{k=1}^p \left[2 (\operatorname{SWAP})\right]_{\indisub_k, \indisub_k+n} \right] \ket{\phi}  \ket{\phi},
\end{split}
\end{equation}
where in the last line we used \Cref{eq:sum_pauli_to_swap} and use subscripts $\left[2 (\operatorname{SWAP}) \right]_{i,j}$ to indicate that the $\operatorname{SWAP}$ operator acts on qubits $i$ and $j$. Let  $\rho_S = \Tr_{[n] \setminus S}( \ket{\phi}\bra{\phi} )$ denote the reduced density matrix of $\ket{\phi}$ for the qubits in $S$. Then,
\begin{equation}
\begin{split}
    \E\left[ \left(\bra{\phi} H_{n,(1,\dots,p)} \ket{\phi}\right)^2\right] &= \frac{1}{\binom{n}{p}} \sum_{\indi \in \mathcal{I}_p^n}  \Tr( \rho_{\{\indisub_{1}, \dots, \indisub_{p} \}}^2).
\end{split}
\end{equation}
The above is the sum of purities $0 \leq \Tr(\rho_{\{\indisub_{1}, \dots, \indisub_{p} \}}^2) \leq 1$, where the upper bound is saturated for product states. 
\end{proof}

\section*{Acknowledgements}
The authors thank Seth Lloyd, Kunal Marwaha, John Preskill, Aram Harrow, Alexander Zlokapa, Anthony Chen, and Robbie King for very enlightening discussions. The authors also thank Brian Swingle for finding an error in a previous draft of this manuscript.

\clearpage
\appendix

\section{Spherical packing}

We first prove a statement that lower bounds the number of points that can be packed into a portion of the unit sphere in $\mathbb{R}^3$.
\begin{lemma} \label{lem:packing_of_sphere}
    Let $S_+^3 = \{\vx \in \mathbb{R}^3: \|\vx\|=1, \vx_3\geq0\}$ be the set of points on the sphere in $\mathbb{R}^3$ on the positive $z$ axis (i.e., the hemisphere). For $\epsilon$ sufficiently small, there exists a finite set $S_\epsilon \subset S_+^3$ of size $|S_\epsilon| \geq \frac{3}{4}\epsilon^{-1}$ where $\vx_3\geq 0.01$ for every $\vx \in S_\epsilon$ and for every $\vx,\vx' \in S_\epsilon$ where $\vx \neq \vx'$, it holds that $|\langle \vx, \vx' \rangle| \leq 1-\epsilon$.
\end{lemma}
\begin{proof}
    We follow a standard packing number bound argument as used in Section V of \cite{6773576}. Let $\operatorname{Vol}(S^3_{\epsilon})$ denote the area of the portion of the unit sphere lying above the plane $z=0.01$. Similarly, let $A(\epsilon)$ denote the area of the portion of the sphere up to polar angle $\cos^{-1}(\epsilon)$. Then, we have
    \begin{equation}
        |S_\epsilon| \geq \frac{\operatorname{Vol}(S^3_{\epsilon})}{A(\epsilon)},
    \end{equation}
    since every point $\vx \in S_\epsilon$ covers an area of $A(\epsilon)$. We have that 
    \begin{equation}
        \operatorname{Vol}(S^3_{\epsilon}) \geq \left(2 - 0.05 \right)\pi .
    \end{equation}
    Furthermore, denoting $\theta = \cos^{-1}(1-\epsilon)$ by the calculation in section V of \cite{6773576}, we have:
    \begin{equation}
        A(\epsilon) = 2\pi \int_0^\theta \sin(\phi) d\phi = 2\pi \epsilon. 
    \end{equation}
    Thus, 
    \begin{equation}
        |S_\epsilon| \geq \frac{2-0.05}{2}\epsilon^{-1} \geq \frac{3}{4}\epsilon^{-1}.
    \end{equation}
\end{proof}

We now prove a corresponding upper bound on the covering number, following the general bound given in \cite{Rankin_1955}.
\begin{lemma} \label{lem:upper_bound_cov}
    For $\epsilon$ sufficiently small, there exists a finite set $S_\epsilon \subset S^2$ of size $|S_\epsilon| \leq \frac{2.01}{\epsilon}$ where for every $\vx \in S^2$ there exists a $\vx'\in S_\epsilon$ such that $|\langle \vx, \vx' \rangle| \geq 1-\epsilon$.
\end{lemma}
\begin{proof}
    Taking $n=3$ in Theorem 2 of \cite{Rankin_1955} gives a bound:
    \begin{equation}
        |S_\epsilon|\leq\frac{\sin\left(\beta\right)\tan\left(\beta\right)}{\int\limits_0^\beta\dd{\phi}\sin\left(\phi\right)\left(\cos\left(\phi\right)-\cos\left(\beta\right)\right)}=\frac{2\sin\left(\beta\right)\tan\left(\beta\right)}{\left(1-\cos\left(\beta\right)\right)^2},
    \end{equation}
    where
    \begin{equation}
        \beta=\arcsin\left(\sqrt{2}\sin\left(\alpha\right)\right)
    \end{equation}
    and
    \begin{equation}
        \alpha=\arccos\left(1-\epsilon\right).
    \end{equation}
    In particular,
    \begin{equation}
        |S_\epsilon|\leq\frac{4\left(1-\left(1-\epsilon\right)^2\right)}{\sqrt{2\left(1-\epsilon\right)^2-1}\left(1-\sqrt{2\left(1-\epsilon\right)^2-1}\right)^2}.
    \end{equation}
    Using the inequalities (for sufficiently small $\epsilon$):
    \begin{equation}
        1-2\epsilon-2\epsilon^2\leq\sqrt{2\left(1-\epsilon\right)^2-1}\leq 1-2\epsilon
    \end{equation}
    we have:
    \begin{equation}
        |S_\epsilon|\leq\frac{8\epsilon}{\left(1-2\epsilon-2\epsilon^2\right)4\epsilon^2}.
    \end{equation}
    For sufficiently small $\epsilon$ this is at most $\frac{2.01}{\epsilon}$.
\end{proof}

\section{Conjectured upper bound on the expected maximum eigenvalue}
\label{app:proof_of_upper_bound}

In this subsection we obtain a conjectured upper bound on $E^*_n(p)=\E[\lambda_{\text{max}}(H_{n,p})/\sqrt{n}]$ by resorting to the trace method. Finding an upper bound appears to be a challenging task and our proofs will hold upon assuming the validity of \Cref{conj:main-0} which characterizes the asymptotic independence structure of certain random matchings of Paulis placed in sequence within the trace. The amount of independence in this structure is controlled by a constant $\gamma(p)$ in \Cref{conj:main-0}. Bounds on this constant will imply corresponding bounds on the approximation factor achievable by product states. Namely, our goal is to prove a bound
\begin{align}
\E[\lambda_{\text{max}}(H_{n,p})/\sqrt{n}]
\le
(1
+o_p(1))
\sqrt{2\gamma(p)\log p}. \label{eq:ground-state-upper}
\end{align}

Fix $\beta>0$. We
have 
\begin{align}
\E[\lambda_{\text{max}}(H_{n,p})/\sqrt{n}]\le {1\over \beta n}
\mathbb{E}\left[\log\Tr\left(\exp\left(\beta\sqrt{n}H_{n,p}\right)\right)\right]
\le 
 {1\over \beta n}
 \log
\mathbb{E}\left[\Tr\left(\exp\left(\beta\sqrt{n}H_{n,p}\right)\right)\right],
\end{align}
and the bound will be obtained via a more tractable right-hand side expression corresponding to the annealed average.
It will be convenient to switch from the Gaussian disorder $\alpha\in \mathcal{N}(0,1)$ to the Rademacher disorder $\alpha=\pm 1$ with probability $1/2$ each
for our upper bound derivation. This is justified by \Cref{theorem:equivalence}. It is straightforward to see that the condition in \Cref{eq:tail-D_n} is verified in particular.

Consider the first moment of the partition function:
\begin{equation} \label{eq:first_moment_preliminary-0}
\mathbb{E}\left[\Tr\left(\exp\left(\beta\sqrt{n}H_{n,p}\right)\right)\right]=\sum\limits_{m=0}^\infty\frac{\left(\beta\sqrt{n}\right)^m}{\binom{n}{p}^{\frac{m}{2}}m!}\sum\limits_{\substack{\bm{\indi_1},\ldots,\bm{\indi_m}\\\bm{\inda_1},\ldots,\bm{\inda_m}}}\mathbb{E}\left[\alpha\left[\bm{\indi_1};\bm{\inda_1}\right]\ldots\alpha\left[\bm{\indi_m};\bm{\inda_m}\right]\right]
\Tr\left(P_{\bm{\indi_1}}^{\bm{\inda_1}}\ldots P_{\bm{\indi_m}}^{\bm{\inda_m}}\right).
\end{equation}

Since the Rademacher disorder is symmetric, the expectation 
$\mathbb{E}\left[\alpha\left[\bm{\indi_1};\bm{\inda_1}\right]\ldots\alpha\left[\bm{\indi_m};\bm{\inda_m}\right]\right]$
is non-zero only when each tuple pair $[\indi,\inda]$ appears in the product an even number of times. In this case the expectation is $1$ (this is where
working with Rademacher disorder is convenient). Thus we rewrite \Cref{eq:first_moment_preliminary-0} as 
\begin{equation} \label{eq:pre_unique_pre_superlinear}
\mathbb{E}\left[\Tr\left(\exp\left(\beta\sqrt{n}H_{n,p}\right)\right)\right]=\sum\limits_{m=0}^\infty\frac{\left(\beta\sqrt{n}\right)^m}{\binom{n}{p}^{\frac{m}{2}}m!}\sum\limits_{\substack{\bm{\indi_1},\ldots,\bm{\indi_m}\\\bm{\inda_1},\ldots,\bm{\inda_m} \\ \left[\bm{\indi};\bm{\inda}\right] \text{ even}} }
\Tr\left(P_{\bm{\indi_1}}^{\bm{\inda_1}}\ldots P_{\bm{\indi_m}}^{\bm{\inda_m}}\right),
\end{equation}
where the sum is restricted to $m$ many pairs $\left(\left[\bm{\indi_1};\bm{\inda_1}\right],\ldots,\left[\bm{\indi_m};\bm{\inda_m}\right]\right)$
such that each term $\left[\bm{\indi};\bm{\inda}\right]$ appears an even number of times in the $m$ pairs; i.e., any pair $\left[\bm{\indi_j};\bm{\inda_j}\right]$ must be matched with another pair $\left[\bm{\indi_k};\bm{\inda_k}\right]$ for the trace term to contribute. This is formalized further by the notion of a matching and pair ordered matchings over given sets.

\begin{definition}[Matching and Pair Ordered Matching] \label{def:matching}
    A matching $M$ of length $2d$ is a set of $d$ pairs $M=\{(a_1, b_1), \dots, (a_d, b_d)\}$ where $a_1, \dots, a_d, b_1, \dots, b_d \in [2d]$ are uniquely chosen so that $\{a_1, \dots, a_d, b_1, \dots, b_d\} = [2d]$. Given a multi-set $S=\multiset{s_1, s_2, \dots, s_d}$ of size $|S|=d$ consisting of elements in $[d]$, a pair ordered matching $\widetilde{M}(S)$ of $S$ is a permutation of the sequence $(s_1,s_1, s_2, s_2, \dots, s_d, s_d)$. When input $S$ is unspecified, it is assumed to be equal to $[d]$. 
    We denote $\mathcal{M}_d$ as the set of all matchings of length $2d$ and $\widetilde{\mathcal{M}}_d(S)$ as the set of all unique pair ordered matchings on the $d$ elements of $S$ (again, $S=[d]$ when unspecified).
\end{definition}

In handling the terms in \Cref{eq:pre_unique_pre_superlinear}, we prove two statements that allow us to simplify the calculations. First, in \Cref{sec:superlinear_restriction}, we show that terms in the sum corresponding to $m=\omega(n)$ --  growing superlinearly with $n$ -- 
have a negligible contribution allowing us to restrict the sum to those where $m \leq Cn$ for some constant $C$ depending on $\beta, p, \epsilon$. 
Also, in \Cref{sec:unique_terms_only}, we show that contributions associated to instances with duplicated tuples, i.e., those where $\indi_c=\indi_d$ for $c\neq d$, are sub-dominating with respect to instances without duplicated tuples. Thus, we sum over unique tuples $\indi_1, \dots, \indi_r$ drawn from the set ${\mathcal{I}_p^n \choose r}$.
Furthermore, tuples can only be appropriately matched when $m$ is even, so we let $m=2r$ and equivalently have 
\begin{equation} \label{eq:tr_exp_expanded_pre_conjecture}
\mathbb{E}\left[\Tr\left(\exp\left(\beta\sqrt{n}H_{n,p}\right)\right)\right]
=
\exp(o(n))
\sum\limits_{r=0}^{Cn}
\frac{\left(\beta^2 n\right)^r}{\binom{n}{p}^{r}(2r)!}\sum\limits_{\substack{\{ \indi_1, \dots, \indi_r \} \in {\mathcal{I}_p^n \choose r}\\ \inda_1, \dots, \inda_r \in \{1,2,3\}^p}} \sum_{\widetilde{M} \in \widetilde{\mathcal{M}}_r}
\Tr\left(P_{\bm{\indi_{\widetilde{M}_1}}}^{\inda_{\widetilde{M}_1}}\ldots P_{\bm{\indi_{\widetilde{M}_{2r}}}}^{\inda_{\widetilde{M}_{2r}}}\right)
+ o(1).
\end{equation}
The above sums over all $(2r)!/2^r$ unique pair ordered matchings as stated in \Cref{def:matching}. 
The $o(1)$ added term in the above handles the contribution from terms associated with $r\geq Cn$ as shown in \Cref{sec:superlinear_restriction}, and the $\exp(o(n))$ pre-factor accounts for contributions from instances where $\{\indi_1,\dots, \indi_r\}$ have duplicated tuples as shown in \Cref{sec:unique_terms_only}.

\subsection{Superlinear terms} \label{sec:superlinear_restriction}
Here, we establish that the contribution from superlinear terms $m=\omega(n)$ in \Cref{eq:pre_unique_pre_superlinear} is negligible.

\begin{lemma}\label{lemma:Super-linear}
For every $\beta,p$ and $\epsilon>0$  there exists $C=C(\beta,p,\epsilon)$ such that the sum in \Cref{eq:pre_unique_pre_superlinear} restricted to the range where $m\ge Cn$
is at most $\exp\left(-\epsilon n\right)$. 
\end{lemma}

\begin{proof} 
The total number $r$ of distinct tuple pairs $[\indi,\inda]$
in each set of $m$ pairs is at most $m/2$.
Note that the total number of distinct pairs is 
$3^p{n\choose p}$ corresponding to the product of the $3^p$ choices for tuples $\inda$ and ${n \choose p}$ choices for tuples $\indi$.
Let $m_0=\min(m/2, 3^p{n\choose p})$ which
is an upper bound on the number of distinct pairs
in each set of $m$ pairs.
Over all $r\le m_0$
the number of choices for $r$-sets of such distinct pairs  
 is then at most
\begin{align}
\sum_{r\le m_0}
{3^p{n\choose p}\choose r}
\le 
m_0{\left(3^p{n\choose p}\right)^{m_0}\over m_0!},
\end{align}
where monotonicity
in $r\le m_0$ was used.
Let $m_1,\ldots,m_r$
denote the multiplicities of each distinct tuple pair among the sequence of $m$ pairs. 
By a stars and bars argument, the total number of choices of $m_1,\ldots,m_r$ is at most ${m+r \choose m} \leq 2^{m+r}$ corresponding to the number of ways to place $m$ bars within $r$ stars noting that each bar must be placed next to at least one other bar. An upper bound on the number
of permutations for a sequence of $m$ pairs
is $m!/2^{m/2}$ corresponding to the number of pair ordered matchings upper bounded by the setting where there are $m/2$ unique tuple pairs. Thus we obtain
\begin{align}
\sum\limits_{m\geq Cn} \frac{\left(\beta\sqrt{n}\right)^m}{\binom{n}{p}^{\frac{m}{2}}m!}\sum\limits_{\substack{\bm{\indi_1},\ldots,\bm{\indi_m}\\\bm{\inda_1},\ldots,\bm{\inda_m} \\ \left[\bm{\indi};\bm{\inda}\right] \text{ even}} }
\Tr\left(P_{\bm{\indi_1}}^{\bm{\inda_1}}\ldots P_{\bm{\indi_m}}^{\bm{\inda_m}}\right) 
&\le
2^n\sum\limits_{m\ge Cn}\frac{\left(\beta\sqrt{n}\right)^m}{\binom{n}{p}^{\frac{m}{2}}m!}
m_0{\left(3^p{n\choose p}\right)^{m_0}\over m_0!}m!2^{m+r-m/2} \\
&\leq
2^n\sum\limits_{m\ge Cn}\frac{\left(\beta\sqrt{n}\right)^m}{\binom{n}{p}^{\frac{m}{2}}}
{\left(3^p{n\choose p}\right)^{m_0}\over (m_0-1)!}2^{m}.
\end{align}
We split the sum into two cases. First consider the range
$m\ge 2(3^p){n\choose p}$. In this case
$m_0=3^p{n\choose p}$ and our range
is $m\ge 2m_0$.
We have 
\begin{align}
\sum_{m\ge 2m_0}
{\left(\beta\sqrt{n}\right)^m \over \binom{n}{p}^{m\over 2}}
2^{m} = 
(1+o(1))
\left({2\beta\sqrt{n} \over {n\choose p}^{1\over 2}}\right)^{2m_0}.
\end{align}
We thus obtain a bound
\begin{align}
\sum\limits_{m\ge 2(3^p){n\choose p}} \frac{\left(\beta\sqrt{n}\right)^m}{\binom{n}{p}^{\frac{m}{2}}m!}\sum\limits_{\substack{\bm{\indi_1},\ldots,\bm{\indi_m}\\\bm{\inda_1},\ldots,\bm{\inda_m} \\ \left[\bm{\indi};\bm{\inda}\right] \text{ even}} }
\Tr\left(P_{\bm{\indi_1}}^{\bm{\inda_1}}\ldots P_{\bm{\indi_m}}^{\bm{\inda_m}}\right) 
\leq 2^n
(1+o(1))
{\left(2\beta (3^{p\over 2})\sqrt{n}\right)^{2m_0}
\over (m_0-1)!}.
\end{align}
Since $m_0=\Theta(n^p)$ with $p\ge 2$,
the value above is bounded
as $2^n2^{-\Omega(m_0)}=2^{-\Omega(n^p)}$.

Next we consider the range
$Cn\le m\le 2(3^p){n\choose p}$.
In this case $m_0=m/2$.
We obtain a sum
\begin{align}
2^n\sum_{Cn \le m\le 2(3^p){n\choose p}}\frac{\left(\beta\sqrt{n}\right)^m}{\binom{n}{p}^{\frac{m}{2}}}
{\left(3^p{n\choose p}\right)^{m\over 2}\over (m/2-1)!}2^{m} =
2^n\sum_{Cn \le m\le 2(3^p){n\choose p}}{\left(4\beta^2n3^p\right)^{m\over 2}\over (m/2-1)!}.
\end{align}
This sum is of the form
$\sum_{r\ge Cn} (\lambda n)^r/r!$
for some constant $\lambda$. 
By the tail bound of the Poisson distribution, this
sum is at most $\exp(-\epsilon n)$ for
sufficiently large $C$.
\end{proof}

\subsection{Matched products}
To evaluate \Cref{eq:tr_exp_expanded_pre_conjecture}, we will consider the matchings induced on each individual qubit and expand the trace over all qubits as the product of traces over Paulis acting on each of the $n$ different qubits separately. Here, we evaluate the trace over a sequence of matched Paulis acting on a single qubit.

In what follows, we fix $d$ as a positive integer to denote the number of unique Paulis. Consider products $\sigma_{2d}\cdots \sigma_1$ where $\sigma_i\in \{\sigma^1, \sigma^2, \sigma^3\}$, where as a reminder $\sigma^1, \sigma^2, \sigma^3$ are the single qubit Pauli $X,Y,Z$ matrices respectively. 
Fix a matching $M$ on $1,\ldots,2d$. 
Say $i\sim j, 1\le i<j\le 2d$ if $i$ and $j$ are matched by $M$.
Let  $\Trace_{\rm sum}(M)\triangleq (1/2)\Trace(\sum_{\sigma:M} \prod_{1\le i\le 2d}\sigma_i)$ where the sum is over all products of $2d$ Paulis $\sigma^1, \sigma^2, \sigma^3$ 
with the restriction that $\sigma_i=\sigma_j$ 
whenever $i\sim j$. Suppose the matching $M$ is chosen uniformly at random from $\mathcal{M}_d$ (see \Cref{def:matching}).

\begin{restatable}[]{lemma}{trsumlemma}
\label{lemma:expected-trace-sum}
When $M$ is drawn uniformly at random from $\mathcal{M}_d$ for any integer $d\geq 1$, it holds that 
\begin{equation}
    \E_{M}[\Trace_{\rm sum}(M)]=2d+1.
\end{equation}
\end{restatable}

\begin{proof}

The proof proceeds by induction in $d$. When $d=1$ the statement holds true since $\Trace_{\rm sum}(M)=3$. 
We assume the validity of the statement for $d'\le d-1$ and now we prove it for $d$. 

Fix any matching $M$.
Define the matching $M_{(12)}$ as follows. If qubits $1$ and $2$ are matched then $M_{(12)}=M$. 
Otherwise $M_{(12)}$ is obtained from $M$ by swapping $1$ and $2$, namely
if qubit $1$ is matched with qubit $j\ne 2$ and qubit $2$ is matched with qubit $k\ne 1$ in $M$, then in $M_{(12)}$ qubit $1$ is matched with $k$ and qubit $2$ with $j$.
Note that if $M$ is chosen uniformly at random from $\mathcal{M}_d$, the same applies to $M_{(12)}$.

In the event qubits $1$ and $2$ are matched, which occurs with probability $1/(2d-1)$
and which we denote by $\mathcal{E}_{1,2}$,
we obtain $\Trace_{\rm sum}(M)=3\Trace_{\rm sum}(\hat M)$ where $\hat M$ is a matching induced by $M$ on $3,4,\ldots,2d$ obtained after eliminating $1,2$. 
$\hat M$ is distributed uniformly at random from $\mathcal{M}_{d-1}$, and thus by the inductive assumption we have
$\E[\Trace_{\rm sum}(\hat M)]=2(d-1)+1$. 
Thus $\E\left[\Trace_{\rm sum}(M)|\mathcal{E}_{1,2}\right]=3(2(d-1)+1)$.

Consider now the event $\mathcal{E}_{1,2}^c$ that $1$ and $2$ are not matched which occurs with probability $1-1/(2d-1)$. For any $a,b\in \{X,Y,Z\}$, let
 $\Trace_{\rm sum,a,b}(M)$ denote the part of $\Trace_{\rm sum}(M)$ subject to the restriction that $\sigma_1=\sigma^a,\sigma_2=\sigma^b$.
Define $\Trace_{\rm sum,a,b}(M_{(12)})$ similarly. We have 
$\Trace_{\rm sum, a,a}(M)=\Trace_{\rm sum,a,a} (M_{(12)})$ for all $a \in \{1,2,3\}$ and $\Trace_{\rm sum, a,b}(M)=-\Trace_{\rm sum, a,b}(M_{(12)})$ for any $a,b \in \{1,2,3\}$ where $a \ne b$ (i.e. swapping anti-commuting Paulis introduces a negative sign).
Then, under the event $\mathcal{E}_{1,2}^c$,
\begin{align}
\Trace_{\rm sum}(M)=\sum_a \Trace_{\rm sum, a,a}(M)+\sum_{a\ne b}\Trace_{\rm sum, a,b}(M),
\end{align}
and we obtain 
\begin{align}
\Trace_{\rm sum}(M)+\Trace_{\rm sum}(M_{(12)})=2 \sum_a \Trace_{\rm sum, a,a}(M).
\end{align}
Taking the expectation over matchings where both $M$ and $M_{(12)}$ are chosen uniformly at random,
we obtain $\E[\Trace_{\rm sum}(M)]=\E[\sum_a \Trace_{\rm sum, a,a}(M)]$. Let $\hat M$ be the matching on  $3,\ldots,2d$ with
$j$ and $k$ matched, where $j$ and $k$ are the matches of qubits $1$ and $2$ respectively in $M$. Note that $\hat M$ is a uniform random matching drawn from $\mathcal{M}_{d-1}$.
 Thus by the inductive assumption $\E{\Trace_{\rm sum}(\hat M)}=2(d-1)+1$. 
Also note
that $\sum_a \Trace_{\rm sum, a,a}(M)=\Trace_{\rm sum}(\hat M)$.

Combining the two cases, we obtain
\begin{align}
r_{2d}=\left(3{1\over 2d-1}+
1-{1\over 2d-1}\right)\left(2(d-1)+1\right)
=2d+1.
\end{align}

\end{proof}

\subsection{Conjectured Asymptotic Independence Structure}
We consider \Cref{eq:tr_exp_expanded_pre_conjecture} replacing the sum over tuples $\indi_1, \dots, \indi_r$ as an expectation over u.a.r. draws from the set ${\mathcal{I}_p^n \choose r}$ of $r$ many $p$-tuples drawn without replacement:
\begin{equation} \label{eq:Trace-sum-m-2}
    \sum\limits_{\substack{\{ \indi_1, \dots, \indi_r \} \in {\mathcal{I}_p^n \choose r}\\ \inda_1, \dots, \inda_r \in \{1,2,3\}^p}} \sum_{\widetilde{M} \in \widetilde{\mathcal{M}}_r}\Tr\left(P_{\indi_{\widetilde{M}_1}}^{\inda_{\widetilde{M}_1}}\ldots P_{\indi_{\widetilde{M}_{2r}}}^{\inda_{\widetilde{M}_{2r}}}\right) 
    = 
    {{n\choose p}\choose r}
    \E_{\indi_1, \dots, \indi_r}\left[ \sum\limits_{\substack{\inda_1, \dots, \inda_r \in \{1,2,3\}^p \\ \widetilde{M} \in \widetilde{\mathcal{M}}_r } }\Tr\left(P_{\indi_{\widetilde{M}_1}}^{\inda_{\widetilde{M}_1}}\ldots P_{\indi_{\widetilde{M}_{2r}}}^{\inda_{\widetilde{M}_{2r}}}\right) \right],
\end{equation}
interpreted as follows. We generate a  sequence of distinct tuples $ \indi_1,\ldots,\indi_r$ u.a.r. without replacement from the total choice of 
${{n\choose p} \choose r}$ such sequences.
For each $j\in [n]$ and each $\indi_k$ containing $j$, 
let $\ell(k,j)$ be the coordinate of $j$ in $\indi_k$. Namely, $\ell(k,j)=\ell\in [p]$ if $\indi_k=(i_{1},\ldots,i_{p})$ and $i_{\ell}=j$. 
With this notation, we sum over all choices of sequences of tuples $\inda_1,\ldots, \inda_r$
the following object:
\begin{align}
\Tr\left(P_{\bm{\indi_{\widetilde{M}_1}}}^{\bm{\inda_{\widetilde{M}_1}}}\ldots P_{\bm{\indi_{\widetilde{M}_{2r}}}}^{\bm{\inda_{\widetilde{M}_{2r}}}}\right)=\prod_{j\in [n]}\Trace\left(\prod_{\substack{k \in \{1, \dots, 2r\} \\ j\in \indi_{\widetilde{M}_{k}}}} \sigma^{\inda_{\ell(\widetilde{M}_{k},j)}}\right). 
\end{align}
Thus, due to product structure of the trace, we can rewrite the double sum 
$\sum_{\bm{\inda_1},\ldots,\bm{\inda_r}} \sum_{\widetilde{M} \in \widetilde{\mathcal{M}}_r}$ appearing in 
 \Cref{eq:Trace-sum-m-2} as follows:
\begin{align}\label{eq:trace-product}
\sum_{\widetilde{M} \in \widetilde{\mathcal{M}}_r} \prod_{j\in [n]}\left(\sum_{\substack{\sigma_1,\ldots,\sigma_{2\Delta(j)}\in \{\sigma^1, \sigma^2, \sigma^3\} \\ \sigma_a = \sigma_b \text{ if matched}}} \Trace\left(\sigma_1,\ldots,\sigma_{2\Delta(j)}\right)\right).
\end{align}
This expression is explained as follows. 
The outer sum is over all $(2r)!/2^{r}$ unique pair ordered matchings of tuples $\indi_1,\ldots, \indi_r$, each replicated twice.  Let $\indj_1,\ldots,\indj_m, m=2r$
denote the set of these replicated tuples. Let
$\Delta(j)$ be the cardinality
of the number of tuples $\indi_k$ containing qubit $j$. Each such ordering induces a matching $M_j$ on the $2\Delta(j)$-set of tuples $\indj_k$
containing  qubit $j$, where elements $s,t\in [2\Delta(j)]$ are matched if they correspond to two replicas $\indj',\indj''$ of the same tuple $\indi$.
The inner sum  is then over all ordered sequences Pauli operators
$\sigma_1,\ldots,\sigma_{2\Delta(i)}\in \{\sigma^1, \sigma^2, \sigma^3\}$ (as a reminder, $\sigma^1, \sigma^2, \sigma^3$ are the Pauli X, Y, and Z operators respectively) respecting the matching, namely $\sigma_s=\sigma_t$ if $s$ and $t$ are matched.
The sum of traces over any fixed matching $M$ is determined by the matching alone, which, as in \Cref{lemma:expected-trace-sum}, was  denoted by $\Trace_{\rm sum}(M)$. We then
rewrite \Cref{eq:trace-product} as 
\begin{align}\label{eq:trace-product-2}
2^n\sum_{\widetilde{M} \in \widetilde{\mathcal{M}}_r} \prod_{j\in [n]}\Trace_{\rm sum}(M_j).
\end{align}
Then we can further rewrite the right hand side of \Cref{eq:Trace-sum-m-2} as 
\begin{align}
 2^n{{n\choose p} \choose r} \E_{\indi_1, \dots, \indi_r}\left[\sum_{\widetilde{M} \in \widetilde{\mathcal{M}}_r} \prod_{j\in [n]}\Trace_{\rm sum}(M_j)\right].
\label{eq:Trace-sum-m-4}
\end{align}
We further replace the sum by $(2r)!/2^{r}\E_{\widetilde{M}}[\cdot]$ where the expectation is also now with respect to the permutation of $m=2r$ elements chosen uniformly at random
up to multiplicities of the replicated pairs. We  obtain
\begin{align}
2^n{{n\choose p} \choose r} {(2r)! \over 2^r} \E_{\indi_1, \dots, \indi_r} \E_{\widetilde{M}}\left[ \prod_{j\in [n]}\Trace_{\rm sum}(M_j)\right].
\label{eq:Trace-sum-m-5}
\end{align}
Here the expectation is both with respect to the random choice of $\indi_1,\ldots,\indi_r$ and random pair ordered matching $\widetilde{M}$ over $m=2r$ elements.

We note that  for each $i\in [n]$, conditioned on the cardinality of $\Delta(i)$, the matching $M_i$ is generated uniformly at random.

We now assume Conjecture~\ref{conj:main-0}, which states:
\begin{align*}
\E_{\indi_1, \dots, \indi_r} \E_{\widetilde{M}}\left[ \prod_{j\in [n]}\Trace_{\rm sum}(M_j)\right]&\leq\gamma(p)^r \exp(O_p(1)n) \; \E_{\indi_1, \dots, \indi_r} \left[ \prod_{j\in [n]} \E_{M_j} \left[\Trace_{\rm sum}(M_j) \right]\right] \\
&=\gamma(p)^r \exp(O_p(1)n) \; \E_{\indi_1, \dots, \indi_r}\left[\prod_{j\in [n]}(2\Delta(j)+1)\right],
\end{align*}
where we recall that $\Delta(j)$ is the (random) cardinality of the number of tuples $\indi_k, k\in [r]$ containing qubit $j$.
Here in the second expression the outer expectation is with respect to the random choice of  tuples  $\indi_1,\ldots,\indi_r$ drawn without replacement 
and the inner expectation is with respect to the distribution $\operatorname{Unif}(\mathcal{M}_{\Delta(j)})$ which is the random matching induced on qubit $j$.

The second part of the conjecture is already verified by \Cref{lemma:expected-trace-sum} since the marginal 
distribution of the matching $M_j$ is u.a.r.. It is the first part of the conjecture which is its essence. The constant $\gamma(p)$ essentially bounds the multiplicative growth in the ratio between the actual versus conjectured values upon adding a tuple $\indi_{r+1}$ at random to a random sequence of matchings of size $r$. Note that it is immediately clear from $(2\Delta(j)+1)\geq 1$ and the number of terms in $\Trace_{\rm sum}(M_j)$ that we have the conjecture is true for $\gamma\left(p\right)=3^p$; what is conjectural is that this can be improved to a more slowly growing function.

\subsection{Upper bound assuming conjectured ansatz}

In what proceeds, we will assume that \Cref{conj:main-0} holds. Our goal is to obtain an upper bound on the quantity
\begin{equation}
    \E_{\indi_1, \dots, \indi_r}\left[\prod_{i\in [n]}(2\Delta(i)+1)\right].
\end{equation}

Suppose $\indi_1,\ldots,\indi_r$ are chosen uniformly at random. This induces a $p$-uniform hypergraph on nodes in the set $[n]$ corresponding to the qubits.
Let $\Delta(i), i\in [p]$ be the number of tuples $\indi_k, 1\le k\le r$ which contain $i$. $\Delta(i)$ is the degree of the node $i$ in this hypergraph. 
We now show that the joint distribution of $(\Delta(i), i\in [n])$ is a product of independent r.v. with $\Pois(pr/n)$ distribution up to $O(\sqrt{r})$ correction.

\begin{lemma}\label{lem:poisson_from_conjecture}
Suppose $r\le Bn$ for some constant $B>0$ and $p \geq 3$. Let $\lambda=pr/n$.
There exists a constant $c(p,B)$ which depends on $B$ and $p$ alone such that
for every fixed  $d_1,\ldots,d_n\ge 0$ satisfying $\sum_i d_i=pr$ it holds
\begin{align}
\pr\left(\Delta(i)=d_i, i\in [n]\right) \le {\sqrt{r}\over c(B,p)}\prod_{i\in [n]}{\lambda^{d_i}\over d_i !}\exp(-\lambda d_i). \label{eq:d-i}
\end{align}
As a result
\begin{align}
\E_{\indi_1, \dots, \indi_r}\left[\prod_{i\in [n]}(2\Delta(i)+1)\right]\le {\sqrt{r}\over c(B,p)} \left({2pr \over n}+1\right)^n.
\end{align}
\end{lemma}

\begin{proof}
Let $D$ be a random variable distributed as $\Pois(pr)$. Consider the process where each element in the set $1,\ldots,D$ is assigned a slot
$i$ independently and uniformly at random. Denote by $D_i$ the number of elements assigned to slot $i$. Then $\sum_i D_i=D$
and by the Poisson splitting theorem, $(D_i, i\in [n])$ is a collection of independent Poisson r.v. with parameter $\lambda=pr/n$. Observe that
the distribution of $(\Delta(i), i\in [n])$ is the distribution of $(D_i, i\in [n])$ conditioned on $D=pr$, and further conditioned on the following
event: (a) tuples $\indi_1\triangleq (i_1,\ldots,i_p),\indi_2\triangleq (i_{p+1},\ldots,i_{2p}),\ldots,\indi_r\triangleq (i_{p(r-1)+1},\ldots,i_{pr})$  
are distinct up to order (that is, viewed as sets) and (b)
each of the $p$ elements in the tuple $\indi_k=(i_{(k-1)p+1},\ldots,i_{kp})$ are assigned to distinct $p$ out of $n$ slots, for each $k\in [r]$. 
We denote this event by 
$\mathcal{E}_{\rm distinct}$.

Thus
\begin{equation} \label{eq:D-lower}
\begin{split}
\pr\left(\Delta(i)=d_i, i\in [n]\right)&=\pr\left(D_i=d_i, i\in [n] | D=pr, \mathcal{E}_{\rm distinct} \right) \\
&={\pr\left(D_i=d_i, i\in [n] , D=pr, \mathcal{E}_{\rm distinct}\right) \over \pr\left(D=pr, \mathcal{E}_{\rm distinct}\right)}  \\
&\le
{\pr\left(D_i=d_i, i\in [n] , D=pr\right) \over \pr\left(D=pr, \mathcal{E}_{\rm distinct}\right)} \\
&=\pr^{-1}\left(D=pr,\mathcal{E}_{\rm distinct}\right)\prod_{i\in [n]}{\lambda^{d_i}\over d_i !}\exp(-\lambda d_i) 
\end{split}
\end{equation}
We claim that $\pr\left(\mathcal{E}_{\rm distinct}|D=pr\right)=\Omega(1)$ in terms of $n$. Namely, it is bounded away from zero uniformly in $n$.
Indeed, regarding (a) note that the probability that two sets $\indi_{k_1}$ and $\indi_{k_2}$ coincide for fixed $k_1\ne k_2$, 
occurs with probability $O(1/n^p)$. Thus, by the union bound, the probability that all tuples are distinct is at least $1-O(n^{p-2})=1-o(1)$,
as $p\ge 3$.
Furthermore, regarding (b),
recall that $r\le Bn$. The probability that the $j$-th tuple $(i_{(j-1)p+1},\ldots, i_{jp})$ does not contain a repeated item is 
$n(n-1)\cdots (n-p+1)/n^p=1-p^2/(2n)+o(1/n)$. The probability that this occurs across all $r$ tuples is then at least
$ \left(1-p^2/(2n)+o(1/n)\right)^{Bn}=\exp(-Bp^2/2+o(1))=\Omega(1)$ as claimed. 

Then 
\begin{equation}
\begin{split}
\pr\left(D=pr,\mathcal{E}_{\rm distinct}\right) &=\pr\left(\mathcal{E}_{\rm distinct} | D=pr\right)\pr\left(D=pr\right) \\
&\ge \left(1-o(1)\right)\exp(-Bp^2/2+o(1))\pr\left(D=pr\right) \\
&=\exp(-Bp^2/2+o(1))\pr\left(D=pr\right).
\end{split}
\end{equation}
We have 
\begin{align}
\pr\left(D=pr\right) ={(pr)^{pr}\over (pr)!}\exp(-pr).
\end{align}
Recalling Stirling's formula $n!=(1+o_n(1))\sqrt{2\pi n}n^ne^{-n}$, we obtain
\begin{align}
\pr\left(D=pr\right) ={1\over (1+o_r(1))\sqrt{2\pi pr}}\ge {c(p)\over \sqrt{r}}
\end{align}
for some  constant $c(p)$ which depends on $p$ alone. We conclude also 
\begin{align}
\pr\left(D=pr,\mathcal{E}_{\rm distinct}\right) \ge {c(B,p)\over \sqrt{r}},
\end{align}
by incorporating the constant $\exp(-Bp^2/2+o(1))$ and $c(p)$ into $c(B,p)$.
Applying this to \Cref{eq:D-lower} we obtain \Cref{eq:d-i}.
The upper bound on $\EE{\prod_i (2\Delta(i)+1)}$ is then obtained by independence of $D_i$ and $\EE{D_i}=\lambda=pr/n$.
\end{proof}

Returning to the expression in \Cref{eq:tr_exp_expanded_pre_conjecture} copied below:
\begin{equation} 
\mathbb{E}\left[\Tr\left(\exp\left(\beta\sqrt{n}H_{n,p}\right)\right)\right]
=
\exp(o_n(n))
\sum\limits_{r=0}^{Cn}
\frac{\left(\beta^2 n \right)^r}{\binom{n}{p}^{r}(2r)!}\sum\limits_{\substack{\{ \indi_1, \dots, \indi_r \} \in {\mathcal{I}_p^n \choose r}\\ \inda_1, \dots, \inda_r \in \{1,2,3\}^p}} \sum_{\widetilde{M} \in \widetilde{\mathcal{M}}_r}
\Tr\left(P_{\bm{\indi_{\widetilde{M}_1}}}^{\inda_{\widetilde{M}_1}}\ldots P_{\bm{\indi_{\widetilde{M}_{2r}}}}^{\inda_{\widetilde{M}_{2r}}}\right)
+ o_n(1),
\end{equation}
we obtain from applying the above Lemma to the terms in the sum:
\begin{equation} \label{eq:bound_from_conjecture}
    \sum\limits_{\substack{\{ \indi_1, \dots, \indi_r \} \in {\mathcal{I}_p^n \choose r}\\ \inda_1, \dots, \inda_r \in \{1,2,3\}^p}} \sum_{\widetilde{M} \in \widetilde{\mathcal{M}}_r}
\Tr\left(P_{\bm{\indi_{\widetilde{M}_1}}}^{\bm{\inda_{\widetilde{M}_1}}}\ldots P_{\bm{\indi_{\widetilde{M}_{2r}}}}^{\bm{\inda_{\widetilde{M}_{2r}}}}\right)
\le 
2^{O_p(1)n} \gamma(p)^r
{{n\choose p} \choose r} {(2r)! \over 2^r} 
{\sqrt{r} \over c(B,p)} \left({2pr \over n}+1\right)^n.
\end{equation}

Returning to the expectation of the exponentiated Hamiltonian $\EE{\exp(\beta\sqrt{n}H)}$ and applying the bound above we obtain 
\begin{equation}
\begin{split}
\EE{\Trace\left(\exp(\beta\sqrt{n}H)\right)}
&\leq 2^{O_p(1)n+o_n(n)} 
\sum_{r=0}^{Cn}{(\beta^2 n \gamma(p))^{r}  \over (2r)! {n\choose p}^{r}}
{{n\choose p} \choose r} {(2r)!  \over 2^r}
{\sqrt{r} \over c(B,p)} \left({2pr \over n}+1\right)^n 
+ o(1)\\
&\le 
2^{O_p(1)n+o_n(n)}+\sum_{r\ge 1}{(\beta^2 n \gamma(p))^{r}  \over (2r)! {n\choose p}^{r}}
2^{O_p(1)n+o_n(n)}
{{n\choose p} \choose r} {(2r)!  \over 2^r}
{\sqrt{r} \over c(B,p)} \left({2pr \over n}+1\right)^n.
\end{split}
\end{equation}
Using ${{n\choose p} \choose r}\le {{n\choose p}^r \over r!}$ and simplifying, we obtain an upper bound
\begin{align}
2^{O_p(1)n+o_n(n)}+
{2^{O_p(1)n+o_n(n)} \over c(B,p)}
\sum_{r\ge 1}\left({\beta^2 n \gamma(p) \over 2}\right)^{r}
{1 \over r!} 
\sqrt{r}\left({2pr \over n}+1\right)^n
\end{align}
Using $\sqrt{r}/r!\le 1/(r-1)!$, we have
\begin{align}
\sum_{r\ge 1}\left({\beta^2 n \gamma(p)\over 2}\right)^{r}
{1 \over r!} 
\sqrt{r}\left({2pr \over n}+1\right)^n 
\le 
(\beta^2n\gamma(p)/2)
\sum_{r\ge 0}\left({\beta^2 n \gamma(p)\over 2}\right)^{r}
{1 \over r!} 
\left({2p(r+1) \over n}+1\right)^n.
\end{align}
Let $Z$ be chosen as $\Pois\left({\beta^2 n \gamma(p)\over 2}\right)$. We recognize the right-hand side as 
\begin{align}
\left({\beta^2 n \gamma(p)\over 2}\right)\exp\left({\beta^2 n \gamma(p)\over 2}\right)\E\left[\left({2p(Z+1) \over n}+1\right)^n\right].
\end{align}
We have $\pr(Z>\beta^2 n \gamma(p)/2+n^{2\over 3})=\exp(-\omega(n))$ as subexponential bounds imply that the probability that a Poisson mean $O(n)$ r.v.
exceeds its mean by order $n^{2/3}$ is superexponentially small in $n$. 

On the other hand, under the event $Z\le \beta^2 n \gamma(p)/2+n^{2\over 3}$ we obtain an upper bound 
\begin{align}
\left({\beta^2 n \gamma(p)\over 2}\right)\exp\left({\beta^2 n \gamma(p)\over 2}\right)\left({2p(\beta^2 n \gamma(p)/2 +n^{2\over 3}+1) \over n}+1\right)^n
=
\exp\left( n \left(\beta^2 \gamma(p)/2+\log(1+p\beta^2)
\right)+o_n(n)\right).
\end{align}
We obtain 
\begin{align}
\EE{\Trace\left(\exp(\beta\sqrt{n}H)\right)} 
\le \exp(O_p(1)n + o_n(n))+\exp\left( n \left({\beta^2 \gamma(p) \over 2}+\log(1+p\gamma\left(p\right)\beta^2) + O_p(1)\right)+o_n(n)\right).
\end{align}
From this we obtain
\begin{align}
(\beta n)^{-1}\log \EE{\Trace\left(\exp(\beta\sqrt{n}H)\right)} \le
{ C\over \beta}+{ \beta \gamma(p) \over 2}+{\log(1+p\gamma\left(p\right)\beta^2) \over \beta}+o_n(1)\triangleq g(\beta,p)+o_n(1)
\end{align}
for some constant $C>\log 2$ which bounds the $O_p(1)$ term.

The proof of the claimed bound $(1+o_p(1))\sqrt{2 \gamma(p) (\log p)}$ is obtained after the lemma below. 
\begin{lemma}
Let $C>1$ be a constant and $\gamma(p) \geq 1$ for all $p$. The following property holds for $g(\beta,p) \triangleq {C\over \beta}+{\beta \gamma(p) \over 2}+{\log(1+p\gamma\left(p\right)\beta^2) \over \beta}$:
\begin{align}
\min_{\beta\ge 0}g(\beta,p) \leq
(1+o_p(1)) \sqrt{2 \gamma(p) \log p}.
\end{align}
\end{lemma}

\begin{proof} 
    Consider the value at $\beta = \sqrt{{2 \log p \over \gamma(p)}}$ where 
    \begin{equation}
    \begin{split}
        {C\over \beta}+{\beta \gamma(p) \over 2}+{\log(1+p \gamma(p) \beta^2) \over \beta} &= {C\sqrt{2\gamma(p)} \over \sqrt{\log p}} + {\sqrt{\gamma(p) \log p} \over \sqrt{2}} + \frac{\log(1+2p\log(p)) \sqrt{\gamma(p)}}{\sqrt{2 \log p}} \\
        &\leq {C\sqrt{2\gamma(p)} \over \sqrt{\log p}} + {\sqrt{\gamma(p) \log p} \over \sqrt{2}} + \frac{\log(4p\log(p)) \sqrt{\gamma(p)}}{\sqrt{2 \log p}} \\
        & = (1+o_p(1))\left( \sqrt{2\gamma(p)\log p} \right).
    \end{split}
    \end{equation}
\end{proof}

\subsection{Bounding contribution of non-unique terms}
\label{sec:unique_terms_only}

To show that we can restrict our analysis to the setting where we only have unique tuples, 
we now consider the expression
in \Cref{eq:pre_unique_pre_superlinear} copied below and restricted to $m=O(n)$ but without the restriction that each tuple $\indi_1,\ldots,\indi_m$ appears exactly twice.
\begin{equation}
\mathbb{E}\left[\Tr\left(\exp\left(\beta\sqrt{n}H_{n,p}\right)\right)\right]=\sum\limits_{m=0}^{Cn}\frac{\left(\beta\sqrt{n}\right)^m}{\binom{n}{p}^{\frac{m}{2}}m!}\sum\limits_{\substack{\bm{\indi_1},\ldots,\bm{\indi_m}\\\bm{\inda_1},\ldots,\bm{\inda_m} \\ \left[\bm{\indi};\bm{\inda}\right] \text{ even}} }
\Tr\left(P_{\bm{\indi_1}}^{\bm{\inda_1}}\ldots P_{\bm{\indi_m}}^{\bm{\inda_m}}\right) + o(1).
\end{equation}
The $o(1)$ above accounts for contributions from $m> Cn$ bounded in \Cref{sec:superlinear_restriction}.

For each $r\le m/2$
let 
\begin{equation}
    \Sigma_m(r) \coloneqq \sum\limits_{\substack{\bm{\indi_1},\ldots,\bm{\indi_m}\\\bm{\inda_1},\ldots,\bm{\inda_m} \\ \left[\bm{\indi};\bm{\inda}\right] \text{ even} \\ |\{ \indi_1,\ldots,\indi_m \}| = r }}
\Tr\left(P_{\bm{\indi_1}}^{\bm{\inda_1}}\ldots P_{\bm{\indi_m}}^{\bm{\inda_m}}\right)
\end{equation}
denote elements of this sum with the restriction that the number of distinct tuples appearing in the sequence $\indi_1,\ldots,\indi_m$ is $r$. Previously, we showed that when all tuples appear twice and $r=m/2$, we have the bound from \Cref{eq:bound_from_conjecture} which says that $\Sigma_m(m/2) \leq \Sigma_m^{\max}$ where
\begin{equation} \label{eq:def_sigma_m_max}
   \Sigma_m^{\max} \coloneqq 2^{O_p(1)n+o_n(n)}
{{n\choose p} \choose m/2} {(m)! \over 2^{m/2}} 
{\sqrt{m/2} \over c(B,p)} \left({pm \over n}+1\right)^n.
\end{equation}
Our main result shows that the bound above still holds when including contributions from duplicated tuples up to a $\exp(o_n(n))$ prefactor.

\begin{lemma}\label{lem:unique_terms_only}
    Let $r_{\rm max}=m/2=O(n)$. Then, assuming the validity of \Cref{conj:main-0}, it holds that
    \begin{equation}
        \sum_{r=1}^{r_{\rm max}} \Sigma_m(r) \leq \exp(o_n(n)) \Sigma_m^{\max}.
    \end{equation}
\end{lemma}

The proof of the above will immediately follow from showing that $\Sigma_m(r) \leq \exp(o(n)) \Sigma_m^{\max}$ in the cases where $r \leq m/2 - n/\sqrt{\log n}$ (\Cref{lemma:non-unique-many}) and $m/2 \geq r \geq m/2 - n/\sqrt{\log n}$ (\Cref{lemma:non-unique-few}). We first begin by considering the setting where $r \leq m/2 - n/\sqrt{\log n}$.

\begin{lemma}\label{lemma:non-unique-many}
Suppose $r\le m/2-n/\sqrt{\log n}$. Then the sum $\Sigma_m(r)$ is at most
\begin{align}
    2^{-\Omega(n\sqrt{\log n})}
{n\choose p}^{m/2} {m!\over (m/2)!}.
\end{align}

\end{lemma}

\begin{proof}
For each fixed choice of a sequence of  
pairs $[\indi_k,\inda_k], 
k\in [m]$ we have
$\Tr{\prod_{1\le k\le m}P_{\indi_k}^{\inda_k}}\le 2^n$. For each fixed choices of 
$\indi_1,\ldots,\indi_m$, the number of choices for $\inda_1,\ldots,\inda_m$ is at most
$3^{pm}=2^{O(n)}$ since $m=O(n)$. Thus $\Sigma_m(r)$ is bounded by $2^{O(n)}$ times
the number
of ordered sequences of tuples $\indi_1,\ldots,\indi_m$ such that the number of distinct tuples is  $r\le m/2-n/\sqrt{\log n}$.

With some abuse of notation, we denote the set of distinct $p$-sets appearing in the $m$-tuple 
$\indi_1,\ldots,\indi_m$ by $\indi_1,\ldots,\indi_r$
and the corresponding multiplicities by $m_1,\ldots,m_r$, so that $m_j\ge 2$ are even and $\sum_{1\le j\le r} m_j=m$. 
By assumption we have $r\le m/2-n/\sqrt{\log n}$.
The number of choices
for $r$ such distinct sets is 
\begin{align}
{{n\choose p} \choose r}\le {{n\choose p}^r \over r!}.
\end{align}
For each multi-set consisting of $r$ distinct $p$-sets with multiplicities $(m_1,\ldots,m_r)$, the number of distinct orderings of them is 
\begin{align}
{m! \over \prod_{1\le j\le r} m_j!}.
\end{align}
We obtain an upper bound:
\begin{align}\label{eq:orders}
\sum_{m_1,\ldots,m_r}
{{n\choose p}^r \over r!}{m!\over \prod_j m_j!}.
\end{align}

The total number of terms in the sum 
$\sum_{m_1,\ldots,m_r}$
is $2^{O(m)}=2^{O(n)}$. Ignoring the $m_j!$ terms, we obtain an upper bound:
\begin{align}
2^{O(n)}m! {{n\choose p}^r\over r!}.
\end{align}
Since $r\le m=O(n)$ and ${n\choose p}=\Theta(n^p)$, with $p\ge 2$, we have
$k/{n\choose  p}=O(1/n^{p-1})$ for $r\le k\le m/2$,
and therefore
\begin{align}
{{n\choose p}^r\over r!}
=
n^{-\Omega(n/\sqrt{\log n})}
{{n\choose p}^{m/2}\over (m/2)!} 
=
2^{-\Omega(n\sqrt{\log n})}
{{n\choose p}^{m/2}\over (m/2)!} .
\end{align}
\end{proof}

Next we consider the case when 
\begin{align}
r\ge m/2-n/\sqrt{\log n}. \label{eq:r-large}
\end{align}
Here, we will show that contributions are equivalent to that from $\sigma_m(m/2)$ up to an $\exp(o_n(n))$ pre-factor.
\begin{lemma} \label{lemma:non-unique-few}
    Suppose $r\ge m/2-n/\sqrt{\log n}$. Then assuming the validity of \Cref{conj:main-0}, it holds that
\begin{align}
    \Sigma_m(r) = \exp(o_n(n)) \Sigma_m^{\max},
\end{align}
\end{lemma}
where $\Sigma_m^{max}$ is the quantity defined in \Cref{eq:def_sigma_m_max}.
\begin{proof}
Let $r_0(\indi_1,\ldots,\indi_m)\le r$ correspond to the number of tuples $\indi$ appearing
exactly twice in the sequences $\indi_1,\ldots,\indi_m$ contained within the sum $\Sigma_m(r)$ which as a reminder takes the form
\begin{equation}
    \begin{split}        
    \Sigma_m(r) &\coloneqq \sum\limits_{\substack{\bm{\indi_1},\ldots,\bm{\indi_m}\\\bm{\inda_1},\ldots,\bm{\inda_m} \\ \left[\bm{\indi};\bm{\inda}\right] \text{ even} \\ |\{ \indi_1,\ldots,\indi_m \}| = r }}
    \Tr\left(P_{\bm{\indi_1}}^{\bm{\inda_1}}\ldots P_{\bm{\indi_m}}^{\bm{\inda_m}}\right) \\
    &= 
    \sum_{\left([\indi_1,\inda_1],\ldots,[\indi_m,\inda_m]\right)\in\mathcal{R}} \Tr\left(P_{\bm{\indi_1}}^{\bm{\inda_1}}\ldots P_{\bm{\indi_m}}^{\bm{\inda_m}}\right).
    \end{split}
\end{equation}
In the above, we introduce the set $\mathcal{R}$ to more easily denote the support of the summation corresponding 
the set of tuples $[\indi_1,\inda_1],\ldots,[\indi_m,\inda_m]$ which have even multiplicities and having the property that $|\{ \indi_1,\ldots,\indi_m \}| = r$.
To ease notation, we will also write $r_0$ without inputs when its inputs are obvious from context. Given the bound on $r$, we have for any sequence of $\indi_1, \dots, \indi_m$ in the sum $\Sigma_m(r)$ above, the following bound on $r_0(\indi_1, \dots, \indi_m)$:
\begin{align}
m\ge 2r_0+4(r-r_0)=2r+2(r-r_0)\ge m-2n/\sqrt{\log n}+2(r-r_0).
\end{align}
This then implies that $r-r_0\le n/\sqrt{\log n}$ and
\begin{align}
m-2r_0\le m-2(r-n/\sqrt{\log n}) \le 4n/\sqrt{\log n} \label{eq:non-paired}.
\end{align}
We now rewrite $\Sigma_m(r)$  as follows. For any subset $A\subset [m]$ with cardinality
$|A|=2r_0$, let $\mathcal{R}(A)\subset \mathcal{R}$ be the subset
of $\mathcal{R}$ corresponding to sequences such that the tuples
$\indi$ which appear exactly twice in the sequence 
$[\indi_1,\inda_1],\ldots,[\indi_m,\inda_m]$ (of which there are precisely $2r_0$) appear precisely in positions
corresponding to  $A\subset [m]$. We define
\begin{align}
\Sigma_m(r,A)\coloneqq \sum_{\left([\indi_1,\inda_1],\ldots,[\indi_m,\inda_m]\right)\in\mathcal{R}(A)} \Tr\left(P_{\bm{\indi_1}}^{\bm{\inda_1}}\ldots P_{\bm{\indi_m}}^{\bm{\inda_m}}\right),
\end{align}
so that
\begin{align}
\Sigma_m(r)=\sum_{1 \le r_0 \le r}\sum_{A\subset [m], |A|=2r_0} \Sigma_m(r,A).
\end{align}

We now upper bound $\Sigma_m(r,A)$  in terms of $\Sigma_{2r_0}(r_0,[2r_0])=\Sigma_{2r_0}(r_0)$ which is 
the sum corresponding to $r$ replaced by $r_0$ and $[m]$ as well as $A$ replaced by $[2r_0]$. 
Equivalently, it corresponds to the sum where there are $r_0$ distinct tuples
$\indi$ and each pair $[\indi,\inda]$ appears exactly twice---the case considered in previous subsections and particularly in \Cref{conj:main-0}.

For this purpose fix $j \in [m]$ where $j\notin A$. The corresponding tuple $\indi_j$ appears at least four times (as it must appear an even number of times) and the corresponding pair $[\indi_j,\inda_j]$ appears
at least twice, say also as $[\indi_\ell,\inda_\ell]$. Namely, $\indi_j=\indi_\ell, \inda_j=\inda_\ell$.
In the corresponding
Pauli product $P_{\bm{\indi_1}}^{\bm{\inda_1}}\ldots P_{\bm{\indi_m}}^{\bm{\inda_m}}$ consider
swapping the location of $P_{\bm{\indi_j}}^{\bm{\inda_j}}$ to a location in the trace adjacent to $P_{\bm{\indi_\ell}}^{\bm{\inda_\ell}}$.
Completing this procedure, 
the product of $P_{\bm{\indi_j}}^{\bm{\inda_j}}$ and 
$P_{\bm{\indi_\ell}}^{\bm{\inda_\ell}}$ is then the identity as they are neighboring. Thus the 
swapping results
in the same trace value up to a sign since each swap operation applies a multiplicative factor of $-1$ or $+1$ depending on whether the swapping terms anti-commute or commute. 
Since there are $3^p$ choices
for $\inda_j$, by removing entries $[\indi_j,\inda_j]$ and $[\indi_\ell,\inda_\ell]$ from the sequence we obtain the new trace value which is at most the 
original trace value times $3^p$ up to a sign. Repeating this for all indices
$j\notin A$ we obtain a value which is at most the 
original trace value times $3^{p(m-2r_0)/2}$ up to a sign. Applying \Cref{eq:non-paired} we have
\begin{align}
3^{p(m-2r_0)/2}\le 3^{p 2n/\sqrt{\log n}}=2^{o(n)}.
\end{align}

Next we account for the choices of tuples $\indi$ for positions in $[m]\setminus A$. We have at most
${{n\choose p}\choose r-r_0}$ choices
for distinct tuples appearing there (in fact smaller since these should be distinct from the remaining
$r_0$ tuples), 
and (crudely) at most 
$(2(r-r_0))!$ orderings of them in the $m-2r_0$ locations of
$[m]\setminus A$. We conclude 
\begin{align}
\Sigma_m(r,A)\le 2^{o(n)}{{n\choose p}\choose r-r_0}(2(r-r_0))!\Sigma_{2r_0}(r_0,[2r_0]).
\end{align}
We note that there are ${m\choose 2r_0}$ choices for the set $A$. 
Since $m=O(n)$ and $m-2r_0=O(n/\sqrt{\log n})$ this term is  $2^{o(n)}$.
We obtain a bound
\begin{align}
\Sigma_m(r)\le \sum_{r_0\le r}
2^{o(n)}{{n\choose p}\choose r-r_0}(2(r-r_0))!\Sigma_{2r_0}(r_0,[2r_0]).
\end{align}
Noting that $\Sigma_{2r_0}(r_0,[2r_0])$ can be handled by the setting of \Cref{conj:main-0} and adopting the bound in \Cref{eq:bound_from_conjecture} for $\Sigma_{2r_0}(r_0,[2r_0])$:
\begin{align}
\Sigma_m(r)\le \sum_{r_0\le r}
2^{o_n(n)}{{n\choose p}\choose r-r_0}(2(r-r_0))!
2^{O_p(1)n+o_n(n)}
{{n\choose p} \choose r_0} {(2r_0)! \over 2^{r_0}} 
{\sqrt{r_0} \over c(B,p)} \left({2pr_0 \over n}+1\right)^n.
\end{align}
We have 
\begin{equation}
\begin{split}
{{n\choose p}\choose r-r_0}(2(r-r_0))!
{{n\choose p} \choose r_0} (2r_0)!
&= {{n\choose p}\choose r}(2r)! \prod_{k=0}^{r_0-1} \frac{{n \choose p}-k}{{n \choose p}-k+r_0-r}\frac{2r_0 - k}{r-k+r_0} \\
&\leq
{{n\choose p}\choose r}(2r)!.
\end{split}
\end{equation}
since the fractions $\frac{{n \choose p}-k}{{n \choose p}-k+r_0-r}\frac{2r_0 - k}{r-k+r_0} \leq 1$ for $r=O(n)$ and $r-r_0=O(n)$. We have $2^{r_0}=2^{r+o(n)}$ for the same reason. The term
${\sqrt{r_0} \over c(B,p)} \left({2pr_0 \over n}+1\right)^n$ is non-decreasing in $r_0$. As a result,
\begin{align}
    \Sigma_m(r)\le 
    2^{O_p(1)n+o_n(n)}
{{n\choose p} \choose r} {(2r)! \over 2^{r}} 
{\sqrt{r} \over c(B,p)} \left({2pr \over n}+1\right)^n \leq \exp(o_n(n)) \Sigma_m^{max}.
\end{align}
In the above, we note $r\le m/2$ and the first bound on  $\Sigma_m(r)$ is monotonically increasing in $r$ which allows us to replace $r$ by $m/2$ in the final expression. This is the bound which was used earlier in \Cref{eq:bound_from_conjecture} where $r$ stood for $m/2$.
Thus in the case of \Cref{eq:r-large}, we obtain up to $2^{o(n)}$ the same upper bound as in the case when every
tuple $\indi$ appears exactly twice, analyzed in the earlier subsection.
    
\end{proof}

\section{Alternative proof of \Cref{lemma:expected-trace-sum}}
Here, we present an alternative proof of \Cref{lemma:expected-trace-sum} based on a recursive formula over the input matchings. For convenience, we restate the lemma below.

\trsumlemma*

\begin{proof}
Fix for now an arbitrary matching $M$ on $2d$ elements. Fix any sequence of $2d$ Paulis $\sigma_1,\ldots,\sigma_{2d}$. 
Then by Lemma 6 of \cite{erdHos2014phase}:
\begin{align}\label{eq:recursion}
(1/2)\sum_{\sigma:M}\Trace\left(\prod_{1\le i\le 2d}\sigma_i\right)=(1/2)\sum_{\sigma:M}\sum_{2\le j\le 2d}{\bf 1}(\sigma_1=\sigma_j)(-1)^j \Trace(2,-j),
\end{align}
where the sum on the left and the outer sum on the right is over all $2d$ sequences of Paulis matched according to $M$ as above, and
where $\Trace(2,-j)$ is trace of product $\sigma_2\cdots\sigma_{2d}$ with term $\sigma_j$ omitted. 

Suppose $1$ is matched with $k$ according to $M$,
that is $1\sim k$, 
so that $\sigma_1=\sigma_k$. 
Fix $j\ne k$ and find   $r\notin \{1,k\}$ which is the match of $j$. Consider the modified matching $M_j$ on $\{2,\ldots,2d\} \setminus \{j\}$ with $r$ matched with $k$. 
We claim that 
\begin{align}\label{eq:R-recursion}
\Trace_{\rm sum}(M)=\sum_{2\le j\le 2d}(-1)^j \Trace_{\rm sum}(M_j)+2(-1)^k\Trace_{\rm sum}(M_k).
\end{align}
Indeed,  consider the contribution to \Cref{eq:recursion} of the sum $(1/2){\bf 1}(\sigma_1=\sigma_j)\sum\Trace(2,-j)$, 
where the sum is now over all choices of $\sigma_2,\ldots,\sigma_{2d}\in \{X,Y,Z\}$ with term $\sigma_j$ omitted and respecting matching $M$. 
When $\sigma_1\ne \sigma_j$ the contribution of this sum is zero. When $\sigma_1=\sigma_j$ we  have $\sigma_1=\sigma_k=\sigma_j=\sigma_r$.
And the contribution is $(-1)^j\Trace_{\rm sum}(M_j)$. 

On the other hand,  consider the contribution to \Cref{eq:recursion} when $j=k$. In this case we have $\sigma_1=\sigma_j$. Let $M_k$ be the induced matching
on $\{2,\ldots,2d\}\setminus k$. The contribution is then $3(-1)^k\Trace_{\rm sum}(M_k)$, where the factor $3$ is due
to the choices $\sigma_1=\sigma_k \in \{\sigma^1, \sigma^2, \sigma^3\}$. We conclude  that \Cref{eq:R-recursion} holds. 

Now suppose $M$ is chosen u.a.r. Observe that for each $j\ne k$,  $M_j, 2\le j\le 2d$ obtained after rewiring is also generated uniformly
at random. Also the match $k$ of 1 is an even number (thus $(-1)^k=1$)  with probability $d/(2d-1)$, and is an odd number  (thus $(-1)^k=-1$) 
with probability $(d-1)/(2d-1)$. We obtain
\begin{align*}
\E_{M}[\Trace_{\rm sum}(M)]=\sum_{2\le j\le 2d}(-1)^j r_{2d-2}+ 2\left({d\over 2d-1}-{d-1\over 2d-1}\right) r_{2d-2} = r_{2d-2}+{2\over 2d-1}r_{2d-2} = r_{2d-2}{2d+1\over 2d-1}.
\end{align*}
We have $r_2=1/2 \left(\Trace(XX)+\Trace(YY)+\Trace(ZZ)\right)=3$. Thus 
\begin{align*}
\E_{M}[\Trace_{\rm sum}(M)]=2d+1.
\end{align*}
\end{proof}

\section{Bounding the $p=2$-local Pauli anticommutation graph Lovász number}
\label{app:Lovasz_bounds}

Here, we consider the Lovász theta function of the commutation graph in the proof of \Cref{lem:max_variance_p2}. Namely, we set $S$ equal to the set of ${n \choose 2}9$ many $2$-local Paulis on $n$ qubits and denote the anti-commutativity graph of this set of ${n \choose 2}9$ Paulis as $G_n$. 
We bound the Lovasz number of the graph by using the properties of the automorphism group of the graph.

\begin{definition}[Vertex symmetric graph \cite{knuth1993sandwich}] \label{def:vertex_symmetric}
    Given a graph $G$ on $n$ nodes labeled $\{1,\dots, n\}$, let $\operatorname{Aut}(G) \subset S_n$ be the automorphism group of the graph. $G$ is vertex symmetric if for any two nodes $u, v \in \{1,\dots, n\}$, there exists a permutation $P \in \operatorname{Aut}(G)$ such that $P(u) = v$.
\end{definition}

\begin{theorem}[Theorem 25 of \cite{knuth1993sandwich}] \label{thm:vertex_transitive_lovasz_upper}
    For any vertex symmetric graph $G$ and its complement $\bar G$ on $N$ nodes,
    \begin{equation}
        \vartheta(G) \vartheta(\bar G) = N.
    \end{equation}
\end{theorem}

Finally, we are ready to prove the main Lemma of this section.

\begin{lemma}
    Let $G_n$ denote the anti-commutativity graph of the set of $9{n \choose 2}$ 2-local Pauli matrices (see \Cref{def:anti_comm_graph}). Then, for all $n \geq 4$ it holds that
    \begin{equation}
        \vartheta(G_n) \leq {n \choose 2}.
    \end{equation}
\end{lemma}
\begin{proof}
We first note that the graph $G_n$ is vertex symmetric (see \Cref{def:vertex_symmetric}). Indeed, each node in $G_n$ is labeled by qubit indices $(i_1, i_2) = \indi \in \mathcal{I}^n_2$ and Paulis $(a_1, a_2) = \inda \in \{1,2,3\}^2$. Note, that the following permutations leave commutation relations invariant and consist of permutations in the automorphism group of $G_n$:
\begin{itemize}
    \item Permutations $P_\pi$ which for any qubit $i$, permute the Paulis for qubit $i$ corresponding to a permutation $\pi \in S_3$ of three elements acting on qubit $k$:
    \begin{equation}
        P_\pi \left\{i_1, i_2, a_1, a_2\right\} = \left\{i_1, i_2, \delta_{k,i_1} \pi a_1 +  (1- \delta_{k,i_1})a_1, \delta_{k,i_2} \pi a_2 +  (1- \delta_{k,i_2}) a_2\right\}.
    \end{equation}
    \item Permutations $P_{\pi'}$ which relabel the qubit indices $i_1, i_2$ based on a permutation $\pi' \in S_n$:
    \begin{equation}
        P_{\pi'} \left\{i_1, i_2, a_1, a_2\right\} = \left\{\pi' i_1, \pi' i_2, a_1, a_2\right\} .
    \end{equation}
\end{itemize}
Combining permutations of the form above, it is easy to see that the graph $G_n$ is vertex symmetric so by \Cref{thm:vertex_transitive_lovasz_upper}:
\begin{equation}
    \vartheta(G_n) \vartheta(\bar G_n) = 9{n \choose 2} .
\end{equation}

The graph $\bar G_n$ contains an edge between nodes $u$ and $v$ whenever the Pauli operators corresponding to $u$ and $v$ commute. Note, that the following set of Paulis correspond to nodes which constitute an independent set of the graph $\bar G_n$:
\begin{equation}
    \begin{split}
        \biggl\{ & X \otimes X \otimes I \otimes I \otimes I^{\otimes n-4}, \quad X \otimes Y \otimes I \otimes I \otimes I^{\otimes n-4}, \quad X \otimes Z \otimes I \otimes I \otimes I^{\otimes n-4} \\
        & Y \otimes I \otimes X \otimes I \otimes I^{\otimes n-4}, \quad Y \otimes I \otimes Y \otimes I \otimes I^{\otimes n-4}, \quad Y \otimes I \otimes Z \otimes I \otimes I^{\otimes n-4}  \\
        & Z \otimes I \otimes I \otimes X \otimes I^{\otimes n-4}, \quad Z \otimes I \otimes I \otimes Y \otimes  I^{\otimes n-4}, \quad X \otimes I \otimes I \otimes Z \otimes I^{\otimes n-4} \biggr\}
    \end{split}
\end{equation}
In fact, each pair of the Pauli operators above share a single index for which the Paulis in that index are not the same. This implies that any such pair does not commute and there are no edges between them in $\bar G_n$. Therefore, since the largest independent set of $G_n$ is a lower bound to $\vartheta(\bar G_n)$, then $9 \leq \vartheta(\bar G_n)$. This implies:
\begin{equation}
    \vartheta(G_n)  =  9{n \choose 2} \vartheta(\bar G_n)^{-1} \leq {n \choose 2} .
\end{equation}

\end{proof}

\clearpage
\bibliographystyle{alpha}
\bibliography{main.bib}

\end{document}

%% file: math_commands.tex
\usepackage{amsmath,amsfonts,amsthm,bm}

\def\eqref#1{equation~\ref{#1}}

\def\1{\bm{1}}

\def\vx{{\bm{x}}}
\def\vy{{\bm{y}}}

\DeclareMathAlphabet{\mathsfit}{\encodingdefault}{\sfdefault}{m}{sl}
\SetMathAlphabet{\mathsfit}{bold}{\encodingdefault}{\sfdefault}{bx}{n}

\newcommand{\E}{\mathbb{E}}
\newcommand{\EE}[1]{\mathbb{E}\!\left[#1\right]}

\DeclareMathOperator{\Tr}{Tr}

\usepackage{latexsym}
\usepackage{graphics}
\usepackage{amsmath}
\usepackage{dsfont}
\usepackage{xspace}
\usepackage{amssymb}
\usepackage{psfrag}
\usepackage{epsfig}
\usepackage{amsthm}
\usepackage{pst-all}

\newcommand{\ignore}[1]{\relax}

%% file: equivalence_section.tex
We primarily study random Hamiltonians whose set of coefficients $\bm \alpha$ are drawn i.i.d.\ from the standard normal distribution. Nevertheless, if one is interested in analyzing the asymptotic values of quantities such as the ground state energy or free energy, these coefficients can be changed to a different set $\bm \alpha'$ with a different distribution to study equivalent models. The equivalence statements we prove will hold for a broad range of random Hamiltonians which can be constructed as follows.
\begin{definition}[$\{m_n,H_n,\left\{\mathcal{D}_{n,i}\right\}_i\}_n$ sequence of random Hamiltonians] \label{def:sequence_of_random_hamiltonians}
    A $\{m_n,H_n,\left\{\mathcal{D}_{n,i}\right\}_i\}_n$ sequence of random Hamiltonians is constructed as follows. First, choose a sequence of positive integers $m_n = \omega(n)$. Then, for each value of $n$ in the sequence corresponding to the number of qubits, choose a collection of $m_n$ Hermitian operators $\{H_n^{(i)}\}_{i=1}^{m_n}$, each of bounded operator norm $\|H_n^{(i)}\| \leq 1$. A random Hamiltonian is drawn by placing these $m_n$ operators in a sum weighted by random coefficients $\bm \alpha_n = \{\alpha_{n,1}\sim\mathcal{D}_{n,1}, \dots, \alpha_{n,m_n}\sim\mathcal{D}_{n,{m_n}}\}$:
    \begin{equation}
        H_n(\bm \alpha_n) = \frac{1}{m_n^{1/2}} \sum_{i=1}^{m_n} \alpha_{n,i} H_n^{(i)}, \quad \quad \| H_n^{(i)} \| \leq 1.
    \end{equation} 
    The disorder distributions $\mathcal{D}_{n,i}$ are assumed to be standard normal $\mathcal{N}\left(0,1\right)$ when not specified.
\end{definition}

Setting $m_n = {n \choose p}3^p$ and the operators $\{H_n^{(i)}\}_{i=1}^{m_n}$ to be the set of $p$-local Pauli matrices corresponds to the spin glass model model (\Cref{eq:hamiltonian_concise_form}). The normalization term $m_n^{-1/2}$ above is chosen so that the expected maximum eigenvalue scales as $\sqrt{n}$ for typical models when the disorder terms $\alpha_{n,i}$ are i.i.d. Gaussian\footnote{This can informally be observed by noting that with this normalization, the variance of the energy with respect to a state is roughly $O_n(1)$. Then the maximum energy over $2^n$ eigenstates should be expected to scale loosely as the maximum of $2^n$ independent Gaussians which is $O(\sqrt{\log 2^n}) = O(\sqrt{n})$.}. The above class of Hamiltonians encompasses the spin glass Hamiltonians we study here as well as others including the SYK model \cite{feng2019spectrum,rosenhaus2019introduction} and quantum versions of $k$-SAT (denoted $k$-QSAT) \cite{bravyi2009bounds,bravyi2011efficient,ambainis2012quantum,laumann2009phase} and quantum max-cut \cite{gharibian2019almost,hwang2023unique}.

\Cref{theorem:equivalence} states sufficient conditions for which equivalence holds. Our proof extends the techniques outlined in \cite{chatterjee2005simple} to the quantum setting. \cite{crawford2007thermodynamics} previously has shown equivalence of the limiting mean free energy for quantum spin models, albeit under a stronger assumption requiring that the third moment $\E[|\alpha'|^3] < \infty$ be finite for the equivalent distribution. This importantly does not encompass the sparse random hypergraph setting described by \Cref{eq:sparse} below. 

As an example of such equivalence in the spin glass setting, any collection of  distributions where the coefficients $\alpha'$ are drawn independently from a bounded random variable (e.g. Rademacher) with variance 1 and mean 0 suffices. Perhaps more interestingly, one can also consider sparse random Hamiltonians where the active terms are supported on a random hypergraph (e.g., Erdős–-Rényi) model. Here, for example, the coefficients can be distributed according to 
\begin{equation}
    \alpha' \propto \begin{cases}
    0 & \text{ with probability } 1- 2 p_n \\
    +1 & \text{ with probability }  p_n  \label{eq:sparse}\\
    -1 & \text{ with probability }  p_n \\
    \end{cases}
    \end{equation}
where $p_n \in [0,1]$ is the probability that any given Pauli term appears and we require $p_n = \omega(n^{-\left(p-1\right)})$ for equivalence to hold due to the boundedness condition in \Cref{theorem:equivalence} copied below. 

\equivalence*
 
As a direct corollary of the above, we have the following equivalence statement for the particular spin glass Hamiltonians $H_{n,p}$ that we study here.
\begin{corollary}
    \label{prop:spin_glass_disorder_equivalence}
    Given the random Hamiltonian $H_{n,p}(\bm \alpha)$ (\Cref{eq:hamiltonian_concise_form}) with coefficients $\alpha[\indi, \inda] \sim \mathcal{N}(0,1)$, let $H_{n,p}(\bm \alpha')$  denote a random Hamiltonian where the coefficients $\alpha[\indi, \inda]$ are replaced by random variables $\alpha'[\indi,\inda] \sim \mathcal{D}_{\indi, \inda}$ drawn independently from a collection of distributions $\mathcal{D}_{\indi,\inda}$. Over the set of quantum states $\mathcal{S}_{\rm all}^n$ on $n$ qubits, the limiting average maximum energies are equal under the two distributions:
    \begin{equation}
        \lim_{n \to \infty} \E_{\alpha[\indi, \inda] \sim \mathcal{N}(0,1)} \left[ \frac{1}{\sqrt{n}} \max_{\ket{\phi} \in \mathcal{S}_{\rm all}^n} \bra{\phi}H_{n,p}(\bm \alpha) \ket{\phi} \right] = \lim_{n \to \infty} \E_{\alpha'[\indi,\inda] \sim \mathcal{D}_{\indi,\inda}} \left[ \frac{1}{\sqrt{n}} \max_{\ket{\phi} \in \mathcal{S}_{\rm all}^n} \bra{\phi}H_{n,p}(\bm \alpha') \ket{\phi} \right] 
    \end{equation}
    whenever the distributions $\mathcal{D}_{\indi,\inda}$ meet the three conditions:
    \begin{equation} \label{eq:spin_glass_universality_requirements}
        \begin{split}
            (\text{first moment}): & \quad \E_{\alpha'[\indi,\inda] \sim \mathcal{D}_{\indi,\inda}}[\alpha'[\indi,\inda]] =  0, \\
            (\text{second moment}): & \quad \E_{\alpha'[\indi,\inda] \sim \mathcal{D}_{\indi,\inda}}[\alpha'[\indi,\inda]^2]  = 1, \\
            (\text{boundedness}): & \quad \forall \epsilon > 0, \; \lim_{n \to \infty} \frac{1}{{n \choose p}}\sum_{\indi \in \mathcal{I}_p^n} \sum_{\inda \in \{1,2,3\}^p}  \mathbb{E}[\alpha'[\indi, \inda]^2 ;  |\alpha'[\indi, \inda]|  > n^{(p-1)/2} \epsilon ] = 0.
        \end{split}
    \end{equation}
\end{corollary}

A canonical model of spin glass Hamiltonians where
the above applies
is the setting of a sparse random graph. 
Construct the distribution $\mathcal{D}_n$ from sparse random variables $\bm{\tilde\alpha}$ distributed as follows: entries of $\bm{\tilde\alpha}=
\left\{\tilde\alpha[\indi; \inda]: \indi\in \mathcal{I}_p^n, \inda\in \{1,2,3\}^p\right\}$
are Rademacher distributed and take values $\pm 1$ with probability $(1/2)3^{-p}d_n/{n\choose p-1}$
and take value $0$ otherwise, independently across $\indi$ and $\inda$. Namely, non-zero terms are supported on a sparse random hypergraph with average degree (number of non-zero disorders per node)
equal to $(d_n 3^{-p}/{n\choose p-1})3^p{n-1\choose p-1}=d_n$, where $d_n$ is assumed to be an arbitrarily slowly growing function of $n$. The $3^p$ factors account for the $3^p$ possible Pauli matrices on each hyperedge. 

In order to bring $\tilde\alpha$ to the setting of \Cref{eq:spin_glass_universality_requirements} we rescale the distribution $\mathcal{D}_n$ as follows: $\alpha'=\sqrt{3^p{n\choose p-1}\over d_n}\tilde\alpha$.
Then $\E[\alpha']=0$ and $\E[\alpha'^2]=1$. We have
\begin{align}
\mathbb{P}\left(|\alpha'|>
n^{p-1\over 2}\epsilon \right)
=
\mathbb{P}\left(|\tilde\alpha|>
\sqrt{d_n 3^{-p} (p-1)!}\epsilon +o(1)\right)
\end{align}
which becomes zero as soon as $d_n$ is large enough. This means that the third condition in 
 \Cref{eq:spin_glass_universality_requirements} holds trivially for large enough $n$ and thus \Cref{prop:spin_glass_disorder_equivalence} holds. We obtain the identity:
 \begin{align}
\lim_{n \to \infty} \E_{\alpha[\indi, \inda] \sim \mathcal{N}(0,1)} \left[ \frac{1}{\sqrt{n}} \max_{\ket{\phi} \in \mathcal{S}_{\rm all}^n} \bra{\phi}H_{n,p}(\bm \alpha) \ket{\phi} \right] = \lim_{n \to \infty} \E_{\alpha'[\indi,\inda] \sim \mathcal{D}_n} \left[ {1\over \sqrt{n}}\max_{\ket{\phi} \in \mathcal{S}_{\rm all}^n} \bra{\phi} H_{n,p}( \bm \alpha') \ket{\phi} \right].
\end{align}

\subsection{Proof of equivalence}
\label{app:equivalence}

As a reminder, we let $\bm \alpha_n = \{\alpha_{n,i}\}_{i \in [m_n]}$ denote the set of random variables which constitute the coefficients of the random Hamiltonian $H_n(\bm \alpha_n)$. All random variables discussed here take values in an open interval $I \subseteq \mathbb{R}$. For simplicity, we will sometimes drop the subscript and denote the set of random variables as $\bm \alpha$ where the number of qubits $n$ is obvious from context.

To prove the equivalence between distributions, we will make use of two linear algebraic facts that allow us to calculate derivatives and bound the free energy in its matrix form.
\begin{lemma}[Derivative of matrix exponential \cite{wilcox1967exponential}] \label{lem:deriv_of_exp}
    Given a Hermitian matrix $H_\theta \in \mathbb{C}^{n \times n}$ which depends on a parameter $\theta \in \mathbb{R}$, 
    \begin{equation}
        \partial_\theta \exp(H_\theta) = \int_0^1 \exp( t H_\theta ) (\partial_\theta H_\theta) \exp ( (1-t)H_\theta )  \; dt.
    \end{equation}
\end{lemma}

\begin{lemma}[Theorem IV.2.5 of \cite{bhatia2013matrix}]\label{lem:sing_val_majorization}
    Given a matrix $A \in \mathbb{C}^{n \times n}$, let $s(A) \in \mathbb{R}^n$ denote the vector of singular values of $A$ in descending order. Then for any two matrices $A,B \in \mathbb{C}^{n \times n}$ and any $r>0$,
    \begin{equation}
        s(A)^r s(B)^r \succ_w s(AB)^r,
    \end{equation}
     where $ \succ_w$ indicates that $s(AB)$ is ``weakly majorized'' (or dominated) by $s(A)s(B)$, i.e.,
     \begin{equation}
         \sum_{i=1}^k s(A)^r_i s(B)^r_{i} \geq \sum_{i=1}^k s(AB)^r_i \; \; \; \text{ for all } k\in [n].
     \end{equation}
\end{lemma}

\cite{chatterjee2005simple} offers a general approach to proving such equivalences of random models, recapitulated in \Cref{thm:invariance} below. Before stating this theorem, we define the quantities $\lambda_r$ which bound the magnitude of influence of any given variable on a function or function class.
\begin{definition}[Definition 1.1 in \cite{chatterjee2005simple}] \label{def:lambda_bound}
    For any open interval $I$ containing 0, any positive integer $n$, any function $f :I^n \to \mathbb{R}$ which is three times differentiable in each coordinate, and $1 \leq r \leq 3$, let 
    \begin{equation}
        \lambda_r(f) \coloneqq \sup \left\{| \partial_i^p f(\vx)|^{r/p} : 1 \leq i \leq n, 1 \leq p \leq r, \vx \in I^n \right\},
    \end{equation}
    where $\partial_i^p$ denotes $p$-fold differentiation with respect to the $i$-th coordinate. For a collection $\mathcal{F}$ of such functions, define $\lambda_r(\mathcal{F}) \coloneqq \sup_{f\in \mathcal{F}} \lambda_r(f)$.
\end{definition}

\begin{theorem}[Adapted from Theorem 1.1 of \cite{chatterjee2005simple}] \label{thm:invariance}
    Let $\vx = (x_1, \dots, x_n) \in I^n$ and $\vy = (y_1, \dots, y_n) \in I^n$ be two vectors of independent random variables whose first and second moments agree: $\mathbb{E}[x_i]=\mathbb{E}[y_i]$ and $\mathbb{E}[x_i^2]=\mathbb{E}[y_i^2]$ for all $i$. Let $f:I^n \to \mathbb{R}$ be a three times differentiable function. Then for any $K>0$,
    \begin{equation}
        \begin{split}
            |\mathbb{E}[f(\vx)] - \mathbb{E}[f(\vy)]| \leq &\lambda_2(f)\sum_{i=1}^n \left[ \mathbb{E}[\vx_i^2 ;  |x_i|  > K ] + \mathbb{E}[\vy_i^2 ;  |y_i|  > K ] \right] \\ 
            &+ \lambda_3(f)\sum_{i=1}^n \left[ \mathbb{E}[|\vx_i|^3 ;  |x_i| \leq K ] + \mathbb{E}[|\vy_i|^3 ;  |y_i|  \leq K ] \right],
        \end{split}
    \end{equation}
    where $\lambda_r(f)$ is as in \Cref{def:lambda_bound}. 
\end{theorem}
The above theorem directly follows upon setting the function $g:\mathbb{R} \to \mathbb{R}$ in Theorem 1.1 of \cite{chatterjee2005simple} to the identity function and noting that $C_1(g)= 1$ and $C_2(g) = 1/6 < 1$ in that corresponding formulation of the theorem.
    
This theorem will be applied to the set of random variables corresponding to the possible energies of the Hamiltonian $H_{n}$ over the set of states $\mathcal{S}_{\rm all}^n$. 
The function class $\mathcal{F}$ in our setting will be the energies of these states. More formally, let $f_\phi:\mathbb{R}^{m_n} \to \mathbb{R}$ denote the normalized energy of a state as a function of the coefficients $\bm \alpha_n$, defined as
\begin{equation}
    f_\phi(\bm \alpha_n) = \frac{1}{\sqrt{n}} \bra{\phi} H_{n}(\bm \alpha_n) \ket{\phi}.
\end{equation}
The function class $\mathcal{F}_{\mathcal{S}_{\rm all}^n}$ contains all such possible functions given the set of states $\mathcal{S}_{\rm all}^n$:
\begin{equation}
    \mathcal{F}_{\mathcal{S}_{\rm all}^n} \coloneqq \{ f_\phi : \ket{\phi} \in \mathcal{S}_{\rm all}^n \}.
\end{equation}

In order to prove \Cref{theorem:equivalence}, we first prove \Cref{thm:max_invariance}, using the result of \Cref{thm:invariance} on the free energy.

\begin{theorem}[adapted from Theorem 1.4 of \cite{chatterjee2005simple}] \label{thm:max_invariance}
Given the set of quantum states $\mathcal{S}_{\rm all}^n$ on $n$ qubits, denote $U=\max_{f \in \mathcal{F}_{\mathcal{S}_{\rm all}^n}} f(\bm \alpha_n) $ and $V=\max_{f \in \mathcal{F}_{\mathcal{S}_{\rm all}^n}} f(\bm \alpha'_n)$ as functions of random variables $\bm \alpha_n$ and $\bm \alpha'_n$. Then for any $K \geq 0$ and $\beta \geq 1$,
\begin{equation}
\begin{split}
    \left| \mathbb{E} U - \mathbb{E} V \right| \leq& 2 \beta^{-1} \log(2^n)   \\
    &+ 2 \beta \frac{1}{n m_n} \sum_{i \in [m_n]}\left[ \mathbb{E}[\alpha_{n,i}^2 ;  |\alpha_{n,i}|  > K ] + \mathbb{E}[(\alpha'_{n,i})^2 ;  |\alpha'_{n,i}|  > K ] \right]\\
    &+ 7 \beta^2 \left(\frac{1}{n m_n}\right)^{3/2} \sum_{i \in [m_n]}\left[ \mathbb{E}[|\alpha_{n,i}|^3 ;  |\alpha_{n,i}| \leq K ] + \mathbb{E}[|\alpha'_{n,i}|^3 ;  |\alpha'_{n,i}|  \leq K ] \right].
\end{split}
\end{equation}
\end{theorem}

\begin{proof}[Proof of \Cref{thm:max_invariance}]
    The proof of this statement follows from \cite{chatterjee2005simple} replacing sums with traces.  For completeness, we replicate the proof and make the relevant changes below. The one major distinction is that the non-commuting nature of the Hamiltonians requires a re-derivation of the bounds for $\lambda_r$ for the free energy function. To keep this proof self-contained, we bound those quantities later in \Cref{lem:free_energy_constants}.
    
    For $\beta \geq 1$, we have that
    \begin{equation}
        \begin{split}
            \max_{\ket{\phi} \in \mathcal{S}_{\rm all}^n} \frac{1}{\sqrt{n}} \bra{\phi} H_n(\bm \alpha_n)\ket{\phi}  &= \beta^{-1} \log \left[ \exp( \beta \max_{\ket{\phi} \in \mathcal{S}_{\rm all}^n} f_\phi(\bm \alpha_n)) \right] \\
            &\leq  \beta^{-1} \log \Tr[ \exp(\frac{1}{\sqrt{ n}}\beta H_n(\bm \alpha_n)) ] \\
            & \leq \beta^{-1} \log\left[ 2^n \exp( \beta \max_{\ket{\phi} \in \mathcal{S}_{\rm all}^n} f_\phi(\bm \alpha_n) ) \right].
        \end{split}
    \end{equation}
    The quantity in the second line $F_\beta(\bm \alpha_n) \coloneqq \beta^{-1} \log  \Tr[ \exp(\frac{1}{\sqrt{ n}}\beta H_n(\bm \alpha_n)) ] $ is the free energy. Denoting $F(\bm \alpha_n) \coloneqq \max_{\ket{\phi} \in \mathcal{S}_{\rm all}^n} f_\phi(\bm \alpha_n)$, we have that  
    \begin{equation}
        \left| F(\bm \alpha_n) - F_\beta(\bm \alpha_n) \right| \leq \beta^{-1} \log(2^n).
    \end{equation}
    The above implies that
    \begin{equation}\label{eq:153}
        \left| \mathbb{E} U - \mathbb{E} V \right| \leq 2 \beta^{-1} \log(2^n) + \left| \E F_\beta(\bm \alpha_n) - \E F_\beta(\bm \alpha_n') \right|.
    \end{equation}

    We bound the above by using \Cref{thm:invariance} on the free energy. Using the bounds in \Cref{lem:free_energy_constants} below for $\lambda_2$ and $\lambda_3$ then yields
    \begin{equation}
    \begin{split}
        \left| \mathbb{E} U - \mathbb{E} V \right| \leq& 2 \beta^{-1} \log(2^n)   \\
        &+ 2 \beta \frac{1}{n m_n} \sum_{i \in [m_n]}\left[ \mathbb{E}[\alpha_{n,i}^2 ;  |\alpha_{n,i}|  > K ] + \mathbb{E}[(\alpha'_{n,i})^2 ;  |\alpha'_{n,i}|  > K ] \right]\\
        &+ 7 \beta^2 \left(\frac{1}{n m_n}\right)^{3/2} \sum_{i \in [m_n]}\left[ \mathbb{E}[|\alpha_{n,i}|^3 ;  |\alpha_{n,i}| \leq K ] + \mathbb{E}[|\alpha'_{n,i}|^3 ;  |\alpha'_{n,i}|  \leq K ] \right],
    \end{split}
    \end{equation}
    which is the desired result.
\end{proof}

\begin{lemma}\label{lem:free_energy_constants}
    Given $\beta \geq 1$ and the free energy $F_\beta(\bm \alpha) \coloneqq \beta^{-1} \log  \Tr[ \exp(\frac{1}{\sqrt{ n}}\beta H_n(\bm \alpha)) ]$ corresponding to the random spin glass Hamiltonian, parameters $\lambda_r(F_\beta)$ as defined in \Cref{def:lambda_bound} are bounded by
    \begin{equation}
        \begin{split}
            \lambda_1(F_\beta) \leq \left(\frac{1}{n m_n}\right)^{1/2}, \quad \quad
            \lambda_2(F_\beta) \leq 2\beta \left(\frac{1}{n m_n}\right), \quad \quad
            \lambda_3(F_\beta) \leq 7\beta^2 \left(\frac{1}{n m_n}\right)^{3/2}.
        \end{split}
    \end{equation}
\end{lemma}
\begin{proof}
    Throughout this, we let $Z_\beta$ denote the partition function defined as
    \begin{equation}
        Z_\beta \coloneqq \Tr[ \exp\left( \frac{\beta}{\sqrt{n}} H_n(\bm \alpha_n) \right) ].
    \end{equation}
    For convenience, we will sometimes drop the subscript and denote disorder terms as just $\bm \alpha$. Using \Cref{lem:deriv_of_exp} and calculating the first derivative of the free energy $F_\beta(\bm \alpha)$ with respect to parameter $\alpha_{n,i}$, we have
    \begin{equation} 
        \begin{split}
            \partial_{\alpha_{n,i}} F_\beta(\bm \alpha) &= \frac{1}{\sqrt{nm_n}  Z_\beta} \int_0^1\Tr[ \exp\left( t \frac{\beta}{\sqrt{n}} H_n \right) H_n^{(i)} \exp\left( (1-t) \frac{\beta}{\sqrt{n}} H_n \right) ] \; dt\\
            &=  \frac{1}{\sqrt{nm_n}  Z_\beta} \Tr[ H_n^{(i)} \exp\left( \frac{\beta}{\sqrt{n}} H_n \right) ].
        \end{split}
    \end{equation}
    The second step above uses the cyclic property of the trace. Applying \Cref{lem:sing_val_majorization} and noting that the singular values of $H_n^{(i)}$ are all bounded by $1$ by assumption, we have
    \begin{equation} \label{eq:post_majorization}
        \begin{split}
            \left| \partial_{\alpha_{n,i}} F_\beta(\bm \alpha) \right|
            & \leq \frac{1}{\sqrt{nm_n}  Z_\beta} \Tr[  \exp\left( \frac{\beta}{\sqrt{n}} H_n \right) ] \\
            & = \frac{1}{\sqrt{nm_n} }.
        \end{split}
    \end{equation}
    Similarly, calculating the second derivative using \Cref{lem:deriv_of_exp}:
    \begin{equation} \label{eq:second_deriv_free_energy}
        \begin{split}
            \partial_{\alpha_{n,i}}^2 F_\beta(\bm \alpha) &= \frac{1}{\sqrt{nm_n}  Z_\beta}  \left( \partial_{\alpha_{n,i}}  \Tr[ H_n^{(i)} \exp\left( \frac{\beta}{\sqrt{n}} H_n \right) ] - \frac{1}{Z_\beta} \Tr[ H_n^{(i)} \exp\left( \frac{\beta}{\sqrt{n}} H_n \right) ] \partial_{\alpha_{n,i}} Z_\beta \right) \\
            &= \underbrace{\frac{1}{\sqrt{nm_n}  Z_\beta}  \Tr[ H_n^{(i)} \partial_{\alpha_{n,i}} \exp\left( \frac{\beta}{\sqrt{n}} H_n \right) ]}_{(I)} - \underbrace{\frac{1}{Z_\beta^2} \frac{\beta}{n m_n} \Tr[ H_n^{(i)} \exp\left( \frac{\beta}{\sqrt{n}} H_n \right) ]^2}_{(II)},
        \end{split}
    \end{equation}
    where $\partial_{\alpha_{n,i}} \exp\left( \frac{\beta}{\sqrt{n}} H_n \right) = \frac{\beta}{\sqrt{nm_n} } \int_0^1 \exp\left( \frac{\beta t}{\sqrt{n}} H_n \right) H_n^{(i)} \exp\left( \frac{\beta(1-t)}{\sqrt{n}} H_n \right)  \; dt $ from applying \Cref{lem:deriv_of_exp}. We use a majorization argument to bound the magnitude of the two terms in the above. For the first term, we use \Cref{lem:sing_val_majorization} and denote $s_k(A)$ as the $k$-th largest singular value of matrix $A$. Since $s_k(H_n^{(i)}) \leq 1$, we have that:
    \begin{equation}
        \begin{split}
            |(I)| &\leq \biggl| \frac{\beta}{n m_n Z_\beta} \int_0^1  \Tr[ H_n^{(i)}   \exp\left( \frac{\beta t}{\sqrt{n}} H_n \right) H_n^{(i)} \exp\left( \frac{\beta(1-t)}{\sqrt{n}} H_n \right)   ] \; dt \biggr| \\
            &\leq \left| \frac{\beta}{n m_n Z_\beta} \int_0^1 \left[ \sum_{k=1}^{2^n}   s_k\left( \exp\left( \frac{\beta t}{\sqrt{n}} H_n \right) \right)  s_k\left(\exp\left( \frac{\beta(1-t)}{\sqrt{n}} H_n \right) \right)  \right] \; dt \right| \\
            &= \frac{\beta}{n m_n}.
        \end{split}
    \end{equation}
    In the above we used the fact that
    \begin{equation}
        Z_\beta = \sum_{k=1}^{2^n}   s_k\left( \exp\left( \frac{\beta t}{\sqrt{n}} H_n \right) \right)  s_k\left(\exp\left( \frac{\beta(1-t)}{\sqrt{n}} H_n \right) \right)
    \end{equation}
    for any $t$. A similar argument shows that the second term in \Cref{eq:second_deriv_free_energy} is also bounded in magnitude by the same quantity. Thus, we have $\left| \partial_{\alpha_{n,i}}^2 F_\beta(\bm \alpha) \right| \leq 2\frac{\beta}{n m_n}$.
    Finally, for the third derivative, we have that:
    \begin{equation}
    \begin{split}
        \partial_{\alpha_{n,i}}^3 F_\beta(\bm \alpha) &= \frac{\beta}{n m_n Z_\beta}    \Tr[ H_n^{(i)} \left(\int_0^1 \partial_{\alpha_{n,i}}\left[\exp\left( \frac{\beta t}{\sqrt{n}} H_n \right) \right] H_n^{(i)} \exp\left( \frac{\beta(1-t)}{\sqrt{n}} H_n \right) \; dt \right) ]\\
        & \quad +     \frac{\beta}{n m_n Z_\beta}    \Tr[ H_n^{(i)} \left(\int_0^1 \exp\left( \frac{\beta t}{\sqrt{n}} H_n \right) H_n^{(i)} \partial_{\alpha_{n,i}}\left[\exp\left( \frac{\beta(1-t)}{\sqrt{n}} H_n \right)\right]  \; dt \right) ]   \\
        & \quad - \frac{\beta}{n m_n Z_\beta^2}  \partial_{\alpha_{n,i}} [Z_\beta]   \Tr[ H_n^{(i)} \left(\int_0^1 \exp\left( \frac{\beta t}{\sqrt{n}} H_n \right)  H_n^{(i)} \exp\left( \frac{\beta(1-t)}{\sqrt{n}} H_n \right) \; dt \right) ]\\
        & \quad + \frac{2\beta}{n m_n Z_\beta^2} \Tr[ H_n^{(i)} \exp\left( \frac{\beta}{\sqrt{n}} H_n \right) ] \partial_{\alpha_{n,i}} \Tr[ H_n^{(i)} \exp\left( \frac{\beta}{\sqrt{n}} H_n \right) ] \\
        & \quad + \frac{2\beta}{n m_n Z_\beta^3} \Tr[ H_n^{(i)} \exp\left( \frac{\beta}{\sqrt{n}} H_n \right) ]^2 \partial_{\alpha_{n,i}} Z_\beta.         
    \end{split}
    \end{equation}
    The same majorization argument applied earlier can be applied to each term in the sum above to show that $\left| \partial_{\alpha_{n,i}}^3 F_\beta(\bm \alpha) \right| \leq \frac{7\beta^2}{n^{3/2} m_n^{3/2}}$. Plugging in the bounds on the derivative magnitudes into $\lambda_r(F_\beta)$ in \Cref{thm:invariance}, we arrive at the final result.
\end{proof}

Finally, we are ready to prove \Cref{theorem:equivalence}.

\begin{proof}[Proof of \Cref{theorem:equivalence}]
    As in the setting of \Cref{thm:max_invariance}, denote $U=\max_{f \in \mathcal{F}_{\mathcal{S}_{\rm all}^n}} f(\bm \alpha_n) $ and $V=\max_{f \in \mathcal{F}_{\mathcal{S}_{\rm all}^n}} f(\bm \alpha'_n)$ to denote the random variables corresponding to the maximum energy under disorder $\bm \alpha_n$ and $\bm \alpha_n'$ respectively. Applying \Cref{thm:max_invariance} with $K = \epsilon \sqrt{m_n/n}$ and $\beta = A n$, and using the bound $\mathbb{E}[|x_i|^3 ;  |x_i| \leq K ] \leq K \mathbb{E}[x_i^2]$, we get 
    \begin{equation} \label{eq:sum_of_terms_in_invariance}
    \begin{split}        
        \left| \mathbb{E} U - \mathbb{E} V \right| &\leq C_1 A^{-1} \\
        & \quad +C_2 \frac{A}{m_n} \sum_{i \in [m_n]} \left( \mathbb{E}[\alpha_{n,i}^2 ;  |\alpha_{n,i}|  > \sqrt{m_n/n} \epsilon ] + \mathbb{E}[(\alpha'_{n,i})^2 ;  |\alpha'_{n,i}|  > \sqrt{m_n/n} \epsilon ] \right) \\
        & \quad + C_3  A^2  \epsilon
    \end{split}
    \end{equation}
    for some positive constants $C_1$, $C_2$, and $C_3$ independent of $n$. 
    Thus, taking the limit of \Cref{eq:sum_of_terms_in_invariance} with $n \to \infty$ and using the boundedness assumption on $\bm \alpha'$ along with the assumption that $m_n = \omega(n)$ so that $\mathbb{E}[\alpha_{n,i}^2 ;  |\alpha_{n,i}|  > \sqrt{m_n/n} \epsilon ] \to 0$, we have that:
    \begin{equation}
        \lim_{n \to \infty} \left| \mathbb{E} U - \mathbb{E} V \right| \leq C_1 A^{-1} + C_3  A^2 \epsilon.
    \end{equation}
    Since $A$ and $\epsilon$ are arbitrary, the result can be shown by properly taking limits of $A$ and $\epsilon$.
\end{proof}